%% file: STNMaxMPG.tex
\RequirePackage{fix-cm}
\documentclass[nospthms]{svjour3}       
\smartqed  
\input{preamble}

\journalname{}
\begin{document}

\title{Hyper Temporal Networks}
\subtitle{A Tractable Generalization of Simple Temporal Networks and its relation to Mean Payoff Games}

\author{Carlo~Comin\and Roberto~Posenato\and Romeo~Rizzi}
\institute{
	Carlo Comin \at Department of Mathematics, 
	University of Trento, Trento, Italy \emph{and}
	LIGM, Universit{\'e} Paris-Est in Marne-la-Vall{\'e}e, Paris, France.
	Supported by Department of Computer Science, University of Verona 
	under PhD grant “Computational Mathematics and Biology”
	on a co-tutelle agreement with Universit\'e Paris-Est 
	in Marne-la-Vall\'ee. \\ \email{carlo.comin@unitn.it}
	\and
	Roberto Posenato, Romeo Rizzi \at Department of Computer Science, University of Verona, Verona, Italy \\ \email{\{roberto.posenato, romeo.rizzi\}@univr.it}
}

\date{}
\maketitle

\begin{abstract}
Simple Temporal Networks (\STN{s}) provide a powerful and general tool for representing conjunctions of maximum delay constraints over ordered pairs of temporal variables.
In this paper we introduce Hyper Temporal Networks (\TN{s}), a strict generalization of \STN{s}, to overcome the limitation of considering only conjunctions of constraints but maintaining a practical efficiency in the consistency check of the instances.
In a Hyper Temporal Network a single temporal hyperarc constraint may be defined as a set of two or more maximum delay constraints which is satisfied when at least one of these delay constraints is satisfied.
\TN{s} are meant as a light generalization of \STN{s} offering an interesting compromise.
On one side, there exist practical pseudo-polynomial time algorithms for checking consistency and computing feasible schedules for \TN{s}.
On the other side, \TN{s} offer a more powerful model accommodating natural constraints that cannot be expressed by \STN{s} like \textit{``Trigger off exactly $\delta$ min before (after) the occurrence of the first (last) event in a set.''}, which are used to represent synchronization events in some process aware information systems\slash workflow models proposed in the literature.


\end{abstract}

\section{Introduction}%
\label{sect:introduction}
\input{sectIntroduction}

\section{Motivating Examples}%
\label{sect:motivating}
\input{sectMotivatingExample}

\section{Background and Notation}%
\label{sect:background}
\input{sectBackground}

\section{Hyper Temporal Networks and Consistency Property}%
\label{sect:problemFormulation}
\input{sectProblemFormulation}

\section{Mean Payoff Games}%
\label{sect:meanpayoffgames}
\input{sectMPG}

\section{The Reductions}%
\label{sect:reductions}
\input{sectReductions}

\section{Computational Experiments}%
\label{sect:experiments}
\input{sectExperiments}

\section{Related Work}\label{sect:relatedwork}
\input{sectRelatedWork}

\section{Conclusions and Future Work}%
\label{sect:conclusions}

In the literature, there are different frameworks and approaches aimed to extend the \STN model 
allowing the representation of disjunctive temporal constraints~\cite{DechterMP91,StergiouK00}, 
but at cost of an exponential-time consistency check procedure.  
The only extension with a polynomial time consistency check procedure
we are aware of is the one of Kumar~\cite{Kumar05} mentioned in Section~\ref{sect:relatedwork}.

In this paper, we proposed a novel extension, called Hyper Temporal Network (\TN), 
where it is possible to represent a new kind of disjunctive constraint, hyper constraint, 
and to check the consistency of a network in pseudo-polynomial time.
A hyper constraint is a suitable set of \STN distance constraints and it is satisfied if at least one distance constraint is satisfied.
There could be two kinds of hyperarc: multi-head and multi-tail.
In a multi-head hyperarc, its distance constraints are between a common source timepoint and different destination timepoints.
In a multi-tail hyperarc, its distance constraints are between different source timepoints and a common destination timepoint. 

A \TN is said consistent if it is possible to determine an assignment for all its timepoints such that all hyperarcs are satisfied.
The computational complexity of the consistency problem of a \TN is \NP-complete when instances contain both kinds of hyperarc.

On instances containing either only multi-tail hyperarcs, or only multi-head hyperarcs, 
the consistency problem can be solved by reducing it,
in a very efficient way, to the search of a winning strategy in an equivalent Mean Payoff Game (\MPG), 
and exploiting the known winning-strategy search algorithms for \MPG{s}.

Moreover, we presented an empirical analysis of the efficiency of the resulting consistency check algorithm.
The empirical analysis shows that the proposed algorithm can be effectively 
used in real cases and confirms the general robustness of our approach.

As future work we are investigating the frontier 
of practical efficient consistency checking for possible generalizations of the \TN model as, 
for example, those including contingent constraints~\cite{VidalF99} or conditional ones~\cite{TVP03}.


\bibliographystyle{spmpsci}


\bibliography{biblio}

\end{document}

%% file: sectIntroduction.tex
In many areas of Artificial Intelligence (AI), including planning, scheduling and workflow management systems, the representation and management of quantitative temporal aspects is of crucial importance~\cite{Pani:2001tb,SmithFJ00,EderPR99,BettiniWJ02,CombiGPP12,CombiGMP14}. 
Examples of possible quantitative temporal aspects are: constraints on the earliest start time and latest end time of activities, constraints over the minimum and maximum temporal distance between activities, etc.

In many cases these constraints can be represented as an instance of a \textit{Simple Temporal Network (\STN)}~\cite{DechterMP91}, a directed weighted graph where each node represents a time-point variable (timepoint), usually corresponding to the beginning or the end of an activity,
and each arc specifies a binary constraint on the scheduling times to be assigned to its endpoints.
In~\cite{DechterMP91}, each arc is labeled with a closed interval of real values:
for example, the labeled arc $u\stackrel{[x,y]}{\longrightarrow}v$ encodes the binary constraint $x \leq v - u \leq y$ over its endpoints $u$ and $v$. 
A more uniform and elementary representation of an \STN is provided by its \textit{distance graph}\footnote{Distance graph is also called \textit{constraint graph} by other authors~\cite{Cormen01}. Moreover, Bellman~\cite{Bellman58} was the first to describe the relation between shortest paths and difference constraints in a constraint graph.}~\cite{DechterMP91},
a graph having the same set of nodes as the original one,
but where each arc $u\stackrel{[x,y]}{\longrightarrow}v$ is replaced by two arcs, each labeled with a single real value:
arc $u\stackrel{y}{\longrightarrow}v$ to express the constraint $v - u\leq y$, 
and arc $v\stackrel{-x}{\longrightarrow}u$ to express the constraint $u -v\leq -x$, \ie $x \leq v - u$.
An \STN is said to be \textit{consistent} if it is possible to assign a real value to each timepoint so that all temporal constraints are satisfied.
The consistency property can be verified by searching for negative cycles in the distance graph and it is well known that the consistency check and the determination of the earliest\slash latest value for each timepoint can be done in polynomial time~\cite{DechterMP91}.

However, \STN{s} do not allow the expression of constraints like \textit{``trigger off an event exactly $\delta$ min after the occurrence of the last of its predecessors''}, which are a quite natural constraints to represent synchronization events in a process aware information system plan\slash workflow schema~\cite{wfmc-modello}.
This is because in \STN{s}, and in some of their natural extensions, (1) it is not possible to represent a single constraint involving more than two timepoints and (2) all constraints have to be satisfied in order to have the network consistent. 
On the contrary, the above constraint about a synchronization event can be represented as a set of distance constraints, each involving a different pair of timepoints, that is considered satisfied when at least one of set components is satisfied.
In order to represent and analyze disjunctive constraints like the above one, it is then necessary to consider models like \textit{Disjunctive Temporal Problem} (DTP)~\cite{StergiouK00} where a constraint is a set of disjunctive difference constraints over the timepoints. 
The drawback of such model is that the consistency check problem is \NP-complete~\cite{StergiouK00}.

\subsection{Contribution}
In this article we propose to generalize \STN to \emph{Hyper Temporal Network} (\TN), which allows also the expression of constraints like the above one regarding synchronization events, but where the consistency check is amenable of effective solution algorithms.

Moreover, we show an interesting link between the consistency check of \TN{s} and 
resolution in Mean Payoff Games (\MPG),
a family of perfect information games played on graphs by two opponents~\cite{EhrenfeuchtMycielski:1979}, for which some pseudo-polynomial time algorithms for determining winning strategies are known~\cite{ZwickPaterson:1996,brim2011faster}.

A preliminary version of this article appeared in the proceedings of TIME symposium~\cite{CominPR14}. Here we extend the presentation as follows: 
(1) the definition of \TN has been extended in order to allow the presence of two kinds of hyperarcs; 
(2) the motivating example section has been revised to show how the new kind of hyperarc can be used;
(3) some further issues and pertinent properties about \TN have been introduced and proved;
(4) several proofs have been expanded and clarified;
(5) the experimental analysis of the consistency check algorithm has been improved considering more recent algorithms~\cite{brim2011faster} for finding winning strategies for \MPG{s}. 
This has improved the performances dramatically.

\paragraph{Organization}
The remainder of the article is organized as follows.
In Section~\ref{sect:motivating} we present a motivating example from the domain of the workflow-based process management to bring out \TN{s}.
Section \ref{sect:background} introduces some definitions and well-known results for \STN{s} and introduces some definitions about hypergraphs.
The generalization of \STN{s} into \TN{s} and the definition of consistency problem for \TN{s} are presented in Section~\ref{sect:problemFormulation}.
In Section~\ref{sect:meanpayoffgames} we recall the main facts and results about Mean Payoff Games which are useful for the following sections.
Section~\ref{sect:reductions} presents the investigation into the link between the \TN consistency problem and Mean Payoff Games deriving pseudo-polynomial time algorithms for checking the consistency of \TN{s} and computing feasible schedules whenever they exist.
Some empirical evaluations of the proposed algorithms are reported in Section~\ref{sect:experiments}.
In Section~\ref{sect:relatedwork}, some related works are presented and discussed with respect to our approach.
Section~\ref{sect:conclusions} summarizes the main facts brought to light in this article and presents a possible future development of the work we are currently carrying on.

%% file: sectMotivatingExample.tex
In the introduction we have briefly recalled a kind of constraint that cannot be expressed within \STN{s}.
In this section, we describe in more detail two examples of temporal constraints that cannot be fully described in an \STN in order to introduce and motivate
the new expressive capability of our model.
As a further motivation, at the end of the section we also spotlight how this new capability has
been recently exploited to check the consistency of Conditional Simple Temporal Networks (CSTNs)~\cite{TVP03} in a more efficient way.

Let us consider an example in the domain of the workflow-based process management, a domain concerned with the coordination and control of business processes
using information technology.
A \textit{workflow} is a representation of a business process as the coordinated execution of activities by human or automatic executors (agents).
A \textit{Workflow management system} (WfMS) is a software system that supports the automatic execution of workflows~\cite{wfmc-modello}.
In a WfMS, the management of temporal aspects is a critical component and in the literature there are many proposals on how to extend a workflow in order to
represent and manage temporal constraints of a business
process~\cite{EderPR99,ChinnMadey00,EderGP00,BettiniWJ02,CombiP04,GonzalezR08,CombiP09,CombiGPP12,CombiGMP14}.
In particular, in \cite{EderPR99,ChinnMadey00,EderGP00,CombiP04,BettiniWJ02} authors show how to represent and manage some kinds of temporal constraints using
specific algorithms, while in \cite{GonzalezR08,CombiP09,CombiGPP12,CombiGMP14} authors show how it is possible to represent and manage a wider class of
temporal constraints exploiting models like Time Petri Nets~\cite{MerlinF76} or \STN{s}\slash STNUs~\cite{MorrisMV01}.

In this paper we consider an excerpt of the conceptual temporal model proposed by Combi~\etal~\cite{CombiP09}, where the specification of a temporal workflow is
given by a \textit{workflow schema}, a directed graph (also called workflow graph) where nodes correspond to activities and arcs represent control flows that
define activity dependencies on the order of execution. Both nodes and arcs may be associated to temporal ranges to specify temporal constraints.
There are two different types of activity: tasks and connectors.
Tasks represent elementary work units that will be executed by external agents.
Each task is graphically represented by a box containing a name and a temporal range that specifies the allowed temporal span for its execution.
Connectors represent internal activities executed by the WfMS to achieve a correct and coordinated execution of tasks.
They are graphically represented by diamonds and, as with tasks, each of them has a temporal range that gives the temporal span allowed to the WfMS for
executing it.
Every arc has a temporal property that gives the allowed times that can be spent by the WfMS for possibly delaying the consideration of the next activity after
the end of the previous one.
There are different kinds of connector that allow one to modify a control flow.
\textit{Split} connectors are nodes with one incoming arc and two or more outgoing arcs: after the execution of the predecessor, (possibly) several successors
have to be considered for the execution. The set of nodes that can start their execution is given by the kind of split connector.
A split connector can be: Parallel, Alternative or Conditional. \textit{Join} connectors are nodes with two or more incoming arcs and one outgoing arc only.
The {types} of {activities} considered in \cite{CombiP09} are a subset of the possible activities specified by the Workflow Management
Coalition~\cite{wfmc-modello,vanDerAalst03}.

\figref{FIG:wf} shows a simple workflow schema where the Parallel connector $\sAND_1$ splits the flow into three parallel flows of execution (one for the
sequence of tasks \textsf{$T_1$} and \textsf{$T_2$}, one for task \textsf{$T_3$}, and one for task \textsf{$T_4$} and \textsf{$T_5$}) that have to be joined
(synchronized) by the AND join connector $\sAND_2$ before continuing the execution; all temporal ranges are in minutes.
\begin{figure}[tb] \centering \begin{tikzpicture}[arrows=->,scale=1,node distance=.65 and 1]
    \node[connector] (Tot) {\TOTText{1}{$[1,5]$}}; \node[task,above right=of Tot] (T1) {\TaskText{$T_1$}{$[1,50]$}}; \node[task,right=of T1] (T2)
    {\TaskText{$T_2$}{$[14,20]$}}; \node[task,right=of Tot] (T3) {\TaskText{$T_3$}{$[5,10]$}}; \node[task,below right=of Tot] (T4) {\TaskText{$T_4$}{$[8,40]$}};
    \node[task,right=of T4] (T5) {\TaskText{$T_5$}{$[10,25]$}}; \node[connector,below right=of T2] (And) {\TOTText{2}{$[1,2]$}};
    \draw[dotted] (Tot)++(-1.2,0) -- (Tot); \draw (Tot.north) |- node[timeLabel,above,pos=.5] {$[1,2]$} (T1); \draw (Tot) -- node[timeLabel,above,pos=.5]
    {$[1,2]$} (T3) ; \draw (Tot.south) |- node[timeLabel,below,pos=.5] {$[1,2]$} (T4); \draw (T1) -- node[timeLabel,above,pos=.5] {$[1,2]$}  (T2); \draw (T3) --
    (And) node[timeLabel,above,pos=.5] {$[1,5]$}; \draw (T4) -- node[timeLabel,above,pos=.5] {$[1,2]$}  (T5); \draw (T2) -| (And.north)
    node[timeLabel,above,pos=.5] {$[2,6]$};; \draw (T5) -| (And.south) node[timeLabel,below,pos=.5] {$[5,10]$};
    \draw[dotted] (And) -- ++(1.2,0);
\end{tikzpicture}
\caption{A simple workflow schema excerpt with three parallel flows of execution.}\label{FIG:wf}
\end{figure}
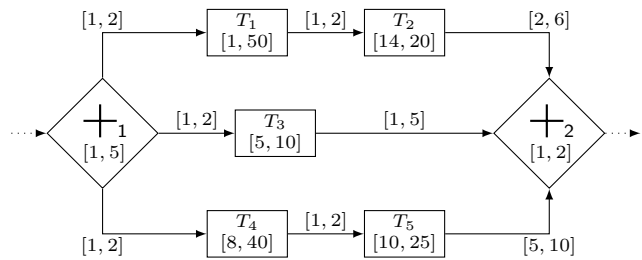

Let us consider the connector $\sAND_2$; according to the recommendations from the Workflow Management Coalition (WfMC)~\cite{wfmc-modello} and the temporal
specification from~\cite{CombiP09}, the execution of this connector requires to wait all incoming flows and, after the last incoming flow, to wait a time
according to the connector temporal range before following the outgoing arc.
In other words, the incoming flows can arrive at different instants but only when the last one arrives, the connector has to be activated in order to continue
with the execution.

Combi~\etal~\cite{CombiP09} proposed a method to translate workflow schemata to \STN{s}\slash STNUs~\cite{MorrisMV01} in order to analyze and validate all
temporal aspects in a rigorous way.
As already noted in~\cite{CombiGMP12} and~\cite{CombiGMP14}, such translation cannot specifically represent the behavior of an AND join connector, because the
kind of constraints in an \STN \slash STNU is limited.
Therefore, in~\cite{CombiGMP14}, the authors proposed an adjustment of the translation of an AND join connector introducing for each incoming arc of the
connector a \textit{buffer} node connected with some determined new arcs and assuming a reasonable but fixed execution algorithm for the \STN.
In more detail, let us consider \figref{FIG:wf-stn} that depicts the representation of workflow of \figref{FIG:wf} by means of an \STN following partially the
method described in~\cite{CombiGMP14} (without loss of generality, here we convert task constraints as \STN arcs instead of STNU contingent ones because we are
interested only in the AND join conversion).
Each activity of the workflow is represented by two \STN nodes, one to represent the begin timepoint, $B_i$, one for the end one, $E_i$, and temporal ranges in
the workflow are represented by \STN arc labels.
Regarding the translation of the AND join node $\sAND_2$, nodes representing the task endings on parallel flows, $E_{T_2}$, $E_{T_3}$, and $E_{T_5}$, are
connected to \textit{buffer} nodes  $b_{1}$, $b_{2}$, and $b_{3}$ that allow the parallel flows to complete their execution following only their temporal
constraints.
Then,  $b_{1}$, $b_{2}$, and $b_{3}$ are connected to node $B_{\sAND_2}$ (which represents the begin instant of the AND join connector) by temporal constraints
$[0,t_1], [0,t_2]$ and $[0,t_3]$, where
the values $t_1, t_2,$ and $t_3$ are determined during the workflow-to-\STN conversion as explained in~\cite{CombiGMP14}.
\begin{figure}[tb]
  \centering
  \begin{tikzpicture}[arrows=->,scale=1,node distance=1 and 1]
	\node[node,smallLabel] (a1b) {$B_{\sAND_1}$};
   	\node[node,below=of a1b,smallLabel] (a1e) {$E_{\sAND_1}$};
    \node[node,above right=of a1e,smallLabel] (t1b) {$B_{T_1}$};
	\node[node,right=of t1b,smallLabel] (t1e) {$E_{T_1}$};
	\node[node,right=of t1e,smallLabel] (t2b) {$B_{T_2}$};
    \node[node,right=of t2b,smallLabel] (t2e) {$E_{T_2}$};
    \node[node,right=4.3 of a1e,smallLabel] (t3b) {$B_{T_3}$};
	\node[node,right=of t3b,smallLabel] (t3e) {$E_{T_3}$};
	\node[node,below right=of a1e,smallLabel] (t4b) {$B_{T_4}$};
	\node[node,right=of t4b,smallLabel] (t4e) {$E_{T_4}$};
    \node[node,right=of t4e,smallLabel] (t5b) {$B_{T_5}$};
    \node[node,right=of t5b,smallLabel] (t5e) {$E_{T_5}$};
	\node[node,right=of t2e,smallLabel] (w2) {$b_{1}$};
	\node[node,right=of t3e,smallLabel] (w3) {$b_{2}$};
	\node[node,right=of t5e,smallLabel] (w4) {$b_{3}$};
    \node[node,right=of w3,smallLabel] (a2b) {$B_{\sAND_2}$};
	\node[node,below=of a2b,smallLabel] (a2e) {$E_{\sAND_2}$};
	\draw[dotted] (a1b)++(-1,0) -- (a1b);
   	\draw[dotted] (a2e) -- ++(.9,0);
	\draw[] (a1b) to [bend left=15] node[timeLabel,right] {$5$} (a1e);%
   	\draw[] (a1e) to [bend left=15] node[timeLabel,left] {$-1$} (a1b);%
   	\draw[] (a2b) to [bend left=15] node[timeLabel,right] {$2$} (a2e);%
   	\draw[] (a2e) to [bend left=15] node[timeLabel,left] {$-1$} (a2b);%
   	\draw[] (a1e) to [bend left=15] node[timeLabel,above,sloped] {$2$} (t1b);%
   	\draw[] (t1b) to [bend left=15] node[timeLabel,below,sloped,pos=.2] {$-1$} (a1e);%
   	\draw[] (a1e) to [bend left=15] node[timeLabel,above,sloped] {$2$} (t3b);%
   	\draw[] (t3b) to [bend left=15] node[timeLabel,below,sloped,pos=.2] {$-1$} (a1e);%
   	\draw[] (a1e) to [bend left=15] node[timeLabel,above,sloped,pos=.7] {$2$} (t4b);%
   	\draw[] (t4b) to [bend left=15] node[timeLabel,below,sloped] {$-1$} (a1e);%
   	\draw[] (t1b) to [bend left=15] node[timeLabel,above,sloped] {$50$} (t1e);%
   	\draw[] (t1e) to [bend left=15] node[timeLabel,below,sloped] {$-1$} (t1b);%
   	\draw[] (t2b) to [bend left=15] node[timeLabel,above,sloped] {$20$} (t2e);%
   	\draw[] (t2e) to [bend left=15] node[timeLabel,below,sloped] {$-14$} (t2b);%
   	\draw[] (t3b) to [bend left=15] node[timeLabel,above,sloped] {$10$} (t3e);%
   	\draw[] (t3e) to [bend left=15] node[timeLabel,below,sloped] {$-5$} (t3b);%
   	\draw[] (t4b) to [bend left=15] node[timeLabel,above,sloped] {$40$} (t4e);%
   	\draw[] (t4e) to [bend left=15] node[timeLabel,below,sloped] {$-8$} (t4b);%
   	\draw[] (t5b) to [bend left=15] node[timeLabel,above,sloped] {$25$} (t5e);%
   	\draw[] (t5e) to [bend left=15] node[timeLabel,below,sloped] {$-10$} (t5b);%
	\draw[] (t1e) to [bend left=15] node[timeLabel,above,sloped] {$2$} (t2b);%
   	\draw[] (t2b) to [bend left=15] node[timeLabel,below,sloped] {$-1$} (t1e);%
	\draw[] (t4e) to [bend left=15] node[timeLabel,above,sloped] {$2$} (t5b);%
   	\draw[] (t5b) to [bend left=15] node[timeLabel,below,sloped] {$-1$} (t4e);%
   	\draw[] (t2e) to [bend left=15] node[timeLabel,above,sloped] {$6$} (w2);%
   	\draw[] (w2) to [bend left=15] node[timeLabel,below,sloped] {$-2$} (t2e);%
   	\draw[] (t3e) to [bend left=15] node[timeLabel,above,sloped] {$5$} (w3);%
   	\draw[] (w3) to [bend left=15] node[timeLabel,below,sloped] {$-1$} (t3e);%
   	\draw[] (t5e) to [bend left=15] node[timeLabel,above,sloped,pos=.4] {$10$} (w4);%
   	\draw[] (w4) to [bend left=15] node[timeLabel,below,sloped] {$-5$} (t5e);%
   	\draw[] (w2) to [bend left=15] node[timeLabel,above] {$t_1$} (a2b);%
   	\draw[] (a2b) to [bend left=15] node[timeLabel,above] {$0$} (w2);%
   	\draw[] (w3) to [bend left=10] node[timeLabel,above,sloped] {$t_2$} (a2b);%
   	\draw[] (a2b) to [bend left=10] node[timeLabel,below,sloped] {$0$} (w3);%
   	\draw[] (w4) to [bend left=15] node[timeLabel,above] {$t_3$} (a2b);%
   	\draw[] (a2b) to [bend left=15] node[timeLabel,below] {$0$} (w4);%
   	
   	\draw[dotted,-,rounded corners=10] (t2e.north west)++(-2ex,2ex) -- ($(w2.north east)+(1ex,2ex)$) -- ($(a2b.east)+(2ex,0)$) 
   		-- ($(w4.south east)+(0ex,-2ex)$) -- ($(t5e.south west)+(-2ex,-2ex)$) -- cycle; 
\end{tikzpicture}
\caption{An \STN representing temporal aspects of the workflow depicted in \figref{FIG:wf}. The dotted region emphasizes, within the workflow excerpt, the connections to an AND join connector.}
\label{FIG:wf-stn}
\end{figure}

Now, let us consider a possible execution scenario.
If $b_1$, $b_2$, and $b_3$ occur all together at instant~20, then, following the proposed temporal semantics~\cite{CombiP09}, the only possible instant value
for $B_{\sAND_2}$ must be~20 while the updated \STN allows any value in the range $[20,20+\min\{t_1,t_2,t_3\}]$.
In~\cite{CombiGMP14}, the authors showed that the right value is always the lower bound of such extended range and, therefore, it is sufficient to adopt an
early execution strategy in order to choose the right value for timepoint $B_{\sAND_2}$.

In other words, the proposed translation has two drawbacks: (1) it requires some preliminary computations for determining $t_1, t_2,$ and $t_3$ values, and (2)
the resulting {\STN} admits some solutions that are not admissible by the semantics of the AND join connector.


To specifically represent the behavior of an AND join connector with respect to its predecessor time points without auxiliary conditions or analysis, it is
necessary to introduce a new kind of constraint based on hyperarcs, as shown in \figref{FIG:wf-hstn}.
In the figure the \textit{multi-tail hyperarc} $A$ consists of three dashed arcs---called \textit{components}---and replaces the arcs from $b_i\; (i=1,2,3)$ to
$B_{\sAND_2}$ of \figref{FIG:wf-stn}.
\begin{figure}[tb] 
\centering 
\begin{tikzpicture}[arrows=->,scale=1,node distance=.95 and 1] 
	\node[node] (t2e) {$E_{T_2}$};
    \node[node,below=of t2e] (t3e) {$E_{T_3}$}; \node[node,below=of t3e] (t4e) {$E_{T_5}$};
	\node[node,left=of t3e] (a1) {$E_{\sAND_1}$};
    \node[node,right=of t2e] (w2) {$b_1$}; \node[node,right=of t3e] (w3) {$b_2$}; \node[node,right=of t4e] (w4) {$b_3$};
	\node[node,right=of w3] (a2b) {$B_{\sAND_2}$};
	
    \draw[] (t2e) to [bend left] node[timeLabel,above,sloped] {$6$} (w2);%
    \draw[] (t3e) to [bend left] node[timeLabel,above,sloped] {$5$} (w3);%
    \draw[] (t4e) to [bend left] node[timeLabel,above,sloped] {$10$} (w4);%
    \draw[] (w2) to [bend left] node[timeLabel,below,sloped] {$-2$} (t2e);%
    \draw[] (w3) to [bend left] node[timeLabel,below,sloped] {$-1$} (t3e);%
    \draw[] (w4) to [bend left] node[timeLabel,below,sloped] {$-5$} (t4e);%
      	
   	\draw[dashed] (w2) to [bend left=30] node[timeLabel,above,sloped] {$A, 0$} (a2b);%
   	\draw[dashed] (w3) to [bend left=15] node[timeLabel,above,sloped] {$A, 0$} (a2b);%
   	\draw[dashed] (w4) to [bend left=-15] node[timeLabel,above,sloped] {$A, 0$} (a2b);%
   	\draw[] (a2b) to [bend left=-15] node[timeLabel,below,sloped] {$0$} (w2);%
    \draw[] (a2b) to [bend left=15] node[timeLabel,below,sloped] {$0$} (w3);%
   	\draw[] (a2b) to [bend left=30] node[timeLabel,below,sloped] {$0$} (w4);%
   	
   	\draw[] (a1) to [bend left=15] node[timeLabel,above,sloped] {$74$} (t2e);%
	\draw[] (t2e) to [bend left=15] node[timeLabel,above,sloped] {$-17$} (a1);%
	\draw[] (a1) to [bend left=10] node[timeLabel,above,sloped] {$12$} (t3e);%
	\draw[] (t3e) to [bend left=10] node[timeLabel,below,sloped] {$-6$} (a1);%
	\draw[] (a1) to [bend left=15] node[timeLabel,above,sloped] {$69$} (t4e);%
	\draw[] (t4e) to [bend left=15] node[timeLabel,above,sloped] {$-20$} (a1);%
	
   	\draw[dotted] (a1)++(-1,.1) to [bend left=15] (a1);%
	\draw[dotted] (a1) to [bend left=15] ($(a1)+(-1,-.1)$);%
	\draw[dotted] (a2b)++(1,-.1) to [bend left=15] (a2b);%
	\draw[dotted] (a2b) to [bend left=15] ($(a2b)+(1,.1)$);%
	
	\draw[dotted,-,rounded corners=10] (t2e.north west)++(-2ex,2ex) -- ($(w2.north east)+(1ex,2ex)$) 
		-- ($(a2b.east)+(3ex,0ex)$) 
   		-- ($(w4.south east)+(1ex,-2ex)$) -- ($(t4e.south west)+(-2ex,-2ex)$) -- cycle; 
\end{tikzpicture}
\caption{An augmented \STN, that we call \TN, where dashed arcs represent components of hyperarcs, a new kind of constraint.
	This \TN improves the representation of the \STN of \figref{FIG:wf-stn}. 
	To emphasize the changes, here we have summed all arcs outside the dotted region.}%
\label{FIG:wf-hstn}
\end{figure}
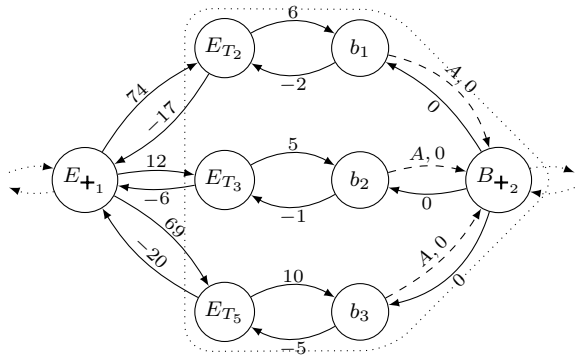
We say that a multi-tail hyperarc is satisfied if at least one of its components is satisfied.
In \figref{FIG:wf-hstn} dashed arcs define the hyperarc $A$ that is satisfied if $B_{\sAND_2}$ is $0$ distant from at least one time point among $b_1, b_2,$ and $b_3$.
Since $B_{\sAND_2}$ is constrained to occur at the same instant or after each time point $b_1, b_2,$ and $b_3$ by the arcs between $B_{\sAND_2}$ and $b_i,\, i=1,2,3$, the result is that to satisfy $A$ it is necessary that $B_{\sAND_2}$ occurs at the same instant of the last time point among $b_1, b_2, $ and $b_3$, as required originally.
In more general, a multi-tail hyperarc is defined as a set of distance constraints (components) between some time points and a common end point. 

The use of hyperarcs allows also the representation of temporal aspect of other advanced connectors as, for example, the Structured
Discriminator~\cite{vanDerAalst03}.
The Structured Discriminator connector provides a means of merging two or more distinct flows in a workflow instance into a single subsequent. In particular, it
triggers the subsequent flow as soon as the the first incoming flow arrives. The arrival of other incoming flows thereafter have no effect on the subsequent flow.
As such, the Structured Discriminator provides a mechanism for progressing the execution of a process once the first of a series of concurrent tasks has
completed and according to the connector temporal range. \figref{fig.wp9a} depicts an excerpt of a workflow schema containing a structured discriminator
connector, \protect\sDIS, that joins three parallel flows.

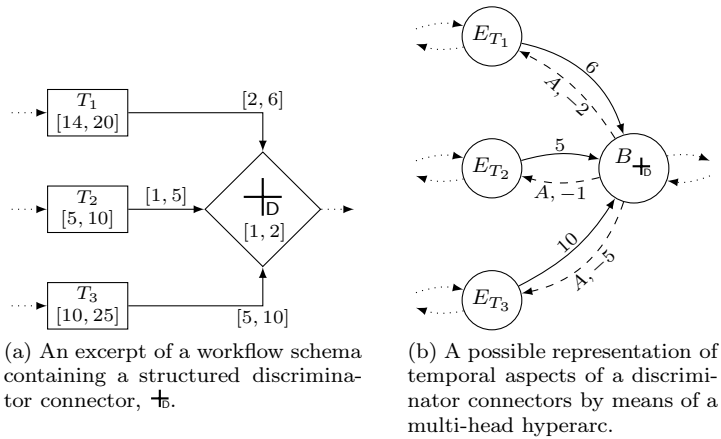
\begin{figure}[tb]
\centering
\subfloat[An excerpt of a workflow schema containing a structured discriminator connector, \protect\sDIS.]{
\begin{tikzpicture}[arrows=->,scale=1,node distance=.65 and 1]
    \node[task] (T1) {\TaskText{$T_1$}{$[14,20]$}};
    \node[task,below=of T1] (T2) {\TaskText{$T_2$}{$[5,10]$}};
    \node[task,below=of T2] (T3) {\TaskText{$T_3$}{$[10,25]$}};
   \node[connector,right=of T2] (DIS) {\DISText{}{$[1,2]$}};

    \draw[dotted] (T1)++(-1,0)--(T1);
    \draw[dotted] (T2)++(-1,0)--(T2);
    \draw[dotted] (T3)++(-1,0)--(T3);
    \draw (T2) -- (DIS) node[timeLabel,above,pos=.5] {$[1,5]$};
    \draw (T1) -| (DIS.north) node[timeLabel,above,pos=.5] {$[2,6]$};;
    \draw (T3) -| (DIS.south) node[timeLabel,below,pos=.5] {$[5,10]$};
    \draw[dotted] (DIS) -- ++(1.2,0);
\end{tikzpicture}\label{fig.wp9a}
}
\quad\quad
\subfloat[A possible representation of temporal aspects of a discriminator connectors by means of a multi-head hyperarc.]{
\begin{tikzpicture}[arrows=->,scale=1,node distance=.95 and 1]
    \node[node] (T1) {$E_{T_1}$};
    \node[node,below=of T1] (T2) {$E_{T_2}$};
    \node[node,below=of T2] (T3) {$E_{T_3}$};
	\node[node,right=of T2] (a2b) {$B_{\sDIS}$};
	
   	\draw[] (T1) to [bend left=30] node[timeLabel,above,sloped] {$6$} (a2b);%
   	\draw[] (T2) to [bend left=15] node[timeLabel,above,sloped] {$5$} (a2b);%
   	\draw[] (T3) to [bend left=-15] node[timeLabel,above,sloped] {$10$} (a2b);%
   	\draw[dashed] (a2b) to [bend left=-15] node[timeLabel,below,sloped] {$A,-2$} (T1);%
    \draw[dashed] (a2b) to [bend left=15] node[timeLabel,below,sloped] {$A,-1$} (T2);%
   	\draw[dashed] (a2b) to [bend left=30] node[timeLabel,below,sloped] {$A,-5$} (T3);%
   	
   	\draw[dotted] (T1)++(-1,.1) to [bend left=15] (T1);%
	\draw[dotted] (T1) to [bend left=15] ($(T1)+(-1,-.1)$);%
   	\draw[dotted] (T2)++(-1,.1) to [bend left=15] (T2);%
	\draw[dotted] (T2) to [bend left=15] ($(T2)+(-1,-.1)$);%
   	\draw[dotted] (T3)++(-1,.1) to [bend left=15] (T3);%
	\draw[dotted] (T3) to [bend left=15] ($(T3)+(-1,-.1)$);%
	\draw[dotted] (a2b)++(1,-.1) to [bend left=15] (a2b);%
	\draw[dotted] (a2b) to [bend left=15] ($(a2b)+(1,.1)$);%
	
\end{tikzpicture}\label{fig.wp9b}
}
\caption{A structured discriminator connector and a possible representation of its temporal aspects.}\label{FIG:wp9}
\end{figure}
 
At the best of our knowledge, currently there are no proposals for the representation of temporal constraints of a discriminator connector in any temporal
workflow model or process-aware information system~\cite{LanzWR12}.
Even exploiting the methodology proposed in \cite{CombiGMP14}, it is
easy to verify that it is not possible to represent such connector as an \STN because in a consistent \STN all constraints have to be satisfied while here it
is necessary to allow the possibility that only one constraint of a set has to be satisfied in order to specifically represent a discriminator connector.
A possible way for specifically managing a discriminator connector consists in following the approach suggested by \cite{CombiGMP14}
for representing activities and considering a variant of the multi-tail hyperarc, called \textit{multi-head hyperarc}, for representing its temporal constraint,
as depicted in \figref{fig.wp9b}.
In the figure there is a multi-head hyperarc $A$ that connects the node representing the beginning instant of the discriminator activity to all nodes
representing the end instant of the activities that precede the considered discriminator and are directly connected to it.
In general a multi-head hyperarc is defined as a set of distance constraints (components) between one time point and some end points.
We say that a multi-head hyperarc is satisfied if at least one of its components is satisfied.
In \figref{fig.wp9b}, dashed arcs define the hyperarc $A$ that is satisfied if $B_{\sDIS}$ is at least 2 distant from $E_{T_1}$ or 1 distant from $E_{T_2}$ or 5
distant from $E_{T_3}$. It is sufficient that one of such previous nodes is executed and that the delay represented in the corresponding connecting arc is
passed to execute $B_{\sDIS}$, as required by the structured discriminator connector semantics.
\smallskip

\TN{s} are not only suitable for better representing temporal constraints originating from temporal workflow, but also for better representing more
general temporal constraint networks like Conditional Simple Temporal Network~\cite{TVP03}.

A Conditional Simple Temporal Network (CSTN) is an enriched graph for representing and reasoning about temporal constraints in domains where some constraints
may apply only in certain condition settings (scenarios).
Each \emph{condition} in a CSTN is represented by a propositional letter whose truth value is \emph{observed} in real time as the outcome of the
execution of an \emph{observation time-point}.
An execution strategy for a CSTN has to determine an execution time for each time-point guaranteeing that all temporal constrains that are significant in the
resulting scenario are satisfied.
An execution strategy can be dynamic in that its execution decisions can react to the information obtained from such observations.
The Conditional Simple Temporal Problem (CSTP) consists in determining whether a given CSTN admits a dynamic execution strategy for any possible combination of
propositional outcomes happens to be observed over time.
If such a strategy exists, the CSTN is said to be dynamically consistent (DC).  

Tsamardinos \etal~\cite{TVP03} solved the CSTP by first encoding it as a meta-level Disjunctive Temporal Problem (DTP), 
then feeding it to an off-the-shelf DTP solver.
Although of theoretical interest, this approach is not practical because the CSTP-to-DTP encoding has exponential size, 
and the DTP solver itself runs in exponential time.
To our knowledge, this approach has never been empirically evaluated~\cite{HunsbergerPC15}.

In~\cite{TIME2015CTNsWithTNs}, 
Comin and Rizzi proposed a novel representation of CSTNs in terms of \TN{s} allowing the determination of the first singly
exponential-time algorithm for checking the dynamic consistency of Conditional Simple Temporal Networks.
More precisely, a CSTN instance is represented as a suitable \TN{} where 
each possible scenario is represented and connected to other scenarios in an
appropriate way and, then, such \TN{} instance is solved in pseudo-polynomial time 
by the algorithms analyzed in the present paper.
\smallskip

In summary, \TN{s} allow the representation of temporal constraints 
that are more general of those represented in STNs~\cite{DechterMP91}, 
because they allow disjunctions involving more than two time points, 
but less general than those represented in 
DTPs~\cite{StergiouK00} because all disjunctions related to a multi-head(tail) hyperarcs have to contain a common variable.
Such kind of STN generalization not only allows the compact representation of some 
common temporal constraints in the domains like the workflow-based process
management but also allows the determination of new interesting algorithm 
for checking dynamic-consistency in richer models like CSTN.

%% file: sectBackground.tex
In this section, we introduce some definitions, notations and well-know results about graphs and conservative graphs; moreover, we recall the relation between
the consistency property of \STN{s} and the conservative property of weighted graphs.

We consider graphs that are directed and weighted on the arcs.
Thus, if $G=(V,A)$ is a graph, then every arc $a\in A$ is a triple $(t_a,h_a,w_a)$ where $t_a \in V$ is the \textit{tail} of $a$, $h_a \in V$ is the
\textit{head} of $a$, and $w_a\in\R$ is the \textit{weight} of $a$. 
Moreover, since we use graphs to represent distance constraints, 
they do not need to have either loops (unary constraints are meaningless)
or parallel arcs (two parallel constraints represent two different distance 
constraints between the same pair of node: only the most restrictive is meaningful).
We also use the notations $h(a)$ for $h_a$, $t(a)$ for $t_a$,
and $w(a)$ or $w(t_a,h_a)$ for $w_a$, when it helps.
The \textit{order} and \textit{size} of a graph $G = (V,A)$
are denoted by $n \triangleq |V|$ and $m \triangleq  |A|$, respectively. 
The size is a good measure for the encoding length of $G$.

A \emph{cycle} of $G$ is a set of arcs $C\subseteq A$ cyclically sequenced 
as $a_0, a_1, \ldots, a_{\ell-1}$ so that $h(a_i) = t(a_j)$ if and only if
$j=(i+1)\mod \ell$; it is called a \emph{negative cycle} if $w(C) \leq 0$, where $w(C)$ stands for $\sum_{e\in C} w_e$.
A graph is called \textit{conservative} when it contains no negative cycle.

A \textit{potential} is a function $p: V \mapsto \R$.
The \textit{reduced weight} of an arc $a = (u,v,w_a)$ with respect to a potential 
$p$ is defined as $w^{p}_a \triangleq w_a - p_v + p_u$.
A potential $p$ of $G = (V,A)$ is called \textit{feasible} if $w^{p}_a\geq 0$ for every $a\in A$.
Notice that, for any cycle $C$, $w^p(C) = w(C)$.
Therefore, the existence of a feasible potential implies that 
the graph is conservative as $w(C) = w^p(C) \geq 0$ for every cycle $C$.

The Bellman-Ford algorithm~\cite{Cormen01} can be used to produce in $O(n m)$ time:
\begin{itemize}
	\item either a proof that $G$ is conservative in the form of a feasible potential function;
	\item or a proof that $G$ is not conservative in the form of a negative cycle $C$ in $G$.
\end{itemize}
When the graph is conservative, the smallest weight of a walk between two nodes is well defined, and, fixed a root node $r$ in $G$,
the potentials returned by the Bellman-Ford algorithm are, for each node $v$, the smallest weight of a walk from $r$ to $v$.
Moreover, if all the arc weights are integral, then these potentials are integral as well. 
Therefore, the Bellman-Ford algorithm provides a proof to the following theorem.
\begin{theorem}[\hspace{-1sp}\cite{Bellman58,Ford,Cormen01}]\label{teo:charConservativeGraphs}
   A graph admits a feasible potential if and only if it is conservative.
   Moreover, when all arc weights are integral, the feasible potential is an integral function.
\end{theorem}

An \STN can be viewed as a weighted graph whose nodes are timepoints that must be placed on the real line and whose arcs express
mutual constraints on the allocations of their end points.
An \STN $G = (V,A)$ is called \textit{consistent} if it admits a \emph{feasible scheduling}, \ie a scheduling $s: V\mapsto \R$ such that
\[
            s(v) \leq s(u) + w(u,v) \quad \forall \text{ arc $(u,v)$ of $G$.} 
\]

\begin{corollary}[\hspace{-1sp}\cite{Bellman58,DechterMP91,Cormen01}]\label{cor.STNcharCons}
   An \STN $G$ is consistent if and only if $G$ is conservative.
\end{corollary}
\begin{proof}
    A feasible scheduling is just a feasible potential. 
    Therefore, this corollary is just a restatement of Theorem~\ref{teo:charConservativeGraphs}.
\end{proof}

In this paper, we also deal with directed weighted hypergraphs.
\begin{definition}[Hypergraph]
A \emph{hypergraph} $\H$ is a pair $(V,\A)$, where $V$ is the set of nodes, and $\A$ is the set of \emph{hyperarcs}.
Each hyperarc $A\in\A$ is either a \emph{multi-head} or a \emph{multi-tail} hyperarc.

A \emph{multi-head} hyperarc $A=(t_A, H_A, w_A)$ has a distinguished node $t_A$, called the \emph{tail} of $A$, and a nonempty set
$H_A\subseteq V\setminus\{t_A\}$ containing the \emph{heads} of $A$; 
to each head $v\in H_A$ is associated a \emph{weight} $w_A(v)\in\R$. 
\figref{fig.multi-head} depicts a possible representation of a multi-head hyperarc: the tail is connected to each head by a dashed arc labeled by the name of the hyperarc and the weight associated to the considered head.

A \textit{multi-tail} hyperarc $A=(T_A, h_A, w_A)$ has a distinguished node $h_A$, called the \emph{head} of $A$, and a nonempty set 
$T_A\subseteq V\setminus\{h_A\}$ containing the \emph{tails} of $A$;
to each tail $v\in T_A$ is associated a \emph{weight} $w_A(v)\in\R$.   
\figref{fig.multi-tail} depicts a possible representation of a multi-tail hyperarc:  the head is connected to each tail by a dotted arc labeled by the name of the hyperarc and the weight associated to the considered tail.
\end{definition}

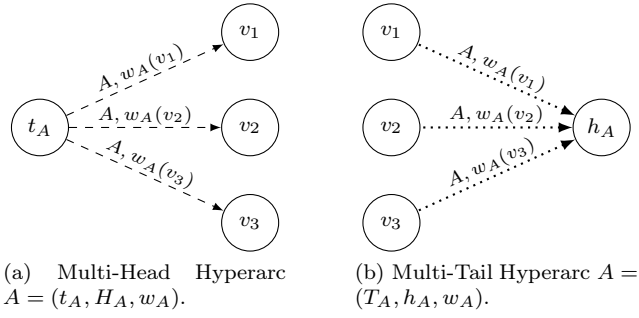
\begin{figure}[tb]
\centering
\subfloat[Multi-Head Hyperarc $A=(t_A, H_A, w_A)$.]{
\begin{tikzpicture}[arrows=->,scale=1,node distance=.5 and 2]
    \node[node] (v1) {$v_1$};
    \node[node,below=of v1] (v2) {$v_2$};
    \node[node,below=of v2] (v3) {$v_3$};
	\node[node,left=of v2] (u) {$t_A$};
%
   	\draw[multiHead] (u) to node[timeLabel,above,sloped] {$A, w_A(v_1)$} (v1);%
   	\draw[multiHead] (u) to node[timeLabel,above,sloped] {$A, w_A(v_2)$} (v2);%
   	\draw[multiHead] (u) to node[timeLabel,above,sloped] {$A, w_A(v_3)$} (v3);%
\end{tikzpicture}
\label{fig.multi-head}
}\quad\quad\quad
\subfloat[Multi-Tail Hyperarc $A=( T_A, h_A, w_A)$.]{
\begin{tikzpicture}[arrows=<-,scale=1,node distance=.5 and 2]
    \node[node] (v1) {$v_1$};
    \node[node,below=of v1] (v2) {$v_2$};
    \node[node,below=of v2] (v3) {$v_3$};
	\node[node,right=of v2] (u) {$h_A$};
   	\draw[multiTail] (u) to node[timeLabel,above,sloped] {$A, w_A(v_1)$} (v1);%
   	\draw[multiTail] (u) to node[timeLabel,above,sloped] {$A, w_A(v_2)$} (v2);%
   	\draw[multiTail] (u) to node[timeLabel,above,sloped] {$A, w_A(v_3)$} (v3);%
\end{tikzpicture}
\label{fig.multi-tail}
}
\caption{A graphical representation of the two kinds of hyperarcs.}
\end{figure}

The \emph{cardinality} of a hyperarc $A\in \A$ is given by  $|A| \triangleq  |H_A\cup \{t_A\}|$ if $A$ is multi-head, and $|A| \triangleq |T_A\cup\{h_A\}|$ if
$A$ is multi-tail; if $|A|=2$, then $A=(u, v, w)$ is a standard arc.
The \textit{order} and \textit{size} of a hypergraph $(V,\A)$ 
are denoted by $n \triangleq |V|$ and $m \triangleq \sum_{A\in \A} |A|$, respectively.

%% file: sectProblemFormulation.tex
We introduce now \textit{Hyper Temporal Networks} (\TN{s}), a strict generalization of \STN{s} to partially overcome the limitation of allowing only
conjunctions of constraints.
Compared to \STN distance graphs, which they naturally extend, \TN{s} allow a greater flexibility in the definition of temporal constraints.

A \TN{} is a directed weighted hypergraph $\H=(V,\A)$ where a node represents a time point variable (timepoint), 
and a multi-head\slash multi-tail hyperarc represents a set of temporal distance constraints  between the tail\slash head and the heads\slash tails, respectively. 

For example, the multi-tail hyperarc $A=(T_A, B_{\sAND_2}, w_A)$ in \figref{FIG:wf-hstn}, where $T_A=\{b_1, b_2, b_3\}$ and $w_A(b_i)=0$ for $i=1,2,3$, stands
for the set of distance constraints $\{B_{\sAND_2} - b_i \leq 0 \mid i=1,2,3\}$. 

In general, we say that a hyperarc is \textit{satisfied} when at least one of its distance constraints is satisfied.
Then, we say that a \TN{} is \textit{consistent} when it is possible to assign a value to each time-point variable so that all of its hyperarcs are satisfied.

More formally, in the \TN framework the consistency problem is defined as the following decision problem.
\newcommand{\savefootnote}[2]{\footnote{\label{#1}#2}}
\newcommand{\repeatfootnote}[1]{\textsuperscript{\ref{#1}}}
\begin{definition}[\GTNC] 
Given a \TN \mbox{$\H=(V,\A)$}, decide whether there exists a scheduling \mbox{$s:V \rightarrow \R$} 
such that, for every hyperarc $A\in\A$, the following holds:
\begin{itemize}
\item if $A=(t,h,w)$ is a standard arc, 
then
\[ 
   s(h)-s(t)\leq w; 
\] 
\item if $A=(t_A, H_A, w_A)$ is a multi-head hyperarc, then 
\[
   s(t_A) \geq \min_{v\in H_A} \{s(v) - w_A(v) \};
\]
\item if $A=(T_A, h_A, w_A)$ is a multi-tail hyperarc, then 
\[
   s(h_A) \leq \max_{v\in T_A} \{s(v) + w_A(v) \}.
\]
\end{itemize}
\end{definition}

Any such scheduling \mbox{$s:V \rightarrow \R$} is called \textit{feasible}. 
A \TN that admits at least one feasible scheduling is called \textit{consistent}.

Comparing the consistency of \TN{s} with the consistency of \STN{s}, the most important aspect of novelty is that, while in a distance graph of a \STN each arc
represents a distance constraint and all such constraints have to be satisfied by a feasible schedule, in a \TN each hyperarc represents a set of one or more
distance constraints and a feasible scheduling has to satisfy at least one such distance constraints for each hyperarc.

Let us show some interesting properties about the consistency problem for \TN{s}.

The first interesting property is that any integral-weighted \TN admits an integral feasible schedule when it is consistent, as proved in the following lemma.
\begin{lemma}\label{Lem:int_sched}
Let $\H=(V,\A)$ be an integral-weighted and consistent $\TN$. 
Then $\H$ admits an integral feasible scheduling $s:V \rightarrow \{-T,-T+1, \ldots, T-1, T\}$,
where $T = \sum_{A\in\A, v\in V} |w_A(v)|$. 
\end{lemma}
\begin{proof} 
Since $\H$ is consistent, then there exists a feasible scheduling $\tilde{s}:V\rightarrow\R$.
The idea in this proof is to project the \TN $\H$ over a conservative graph and then, in that setting, 
to exploit the integrality properties of potentials as stated in Theorem~\ref{teo:charConservativeGraphs}.
However, this projection is asked to resolve the non-determinism contained in the disjunctive nature of the hyperarcs; in order to sort out such non-determinism,
the projection is built considering the given feasible scheduling $\tilde{s}$ as follows.

For each hyperarc $A\in \A$, a weighted directed arc $e_A$ is defined as follows:
\begin{itemize}
	\item if $A=(u,v,w)$ is a standard arc, 
then $e_A\triangleq (u,v,w)$. Note that $\tilde{s}(v) \leq \tilde{s}(u) + w$ follows by the feasibility of $\tilde{s}$; 
\item if $A=(t_A, H_A, w_A)$ is a multi-head hyperarc, then  
	\[e_A\triangleq (t_A, v_A, w_A(v)) \text{ where } v_A=\arg\min_{v\in H_A} \{\tilde{s}(v) - w_A(v) \}. \]
Here, $\tilde{s}(v_A) \leq \tilde{s}(t_A) + w_A(v)$ follows by the feasibility of $\tilde{s}$;
\item if $A=(T_A, h_A, w_A)$ is a multi-tail hyperarc, then 
	\[e_A\triangleq (v_A, h_A, w_A(v)) \text{ where } v_A=\arg\max_{v\in T_A} \{\tilde{s}(v) + w_A(v) \}.\]
Here, $\tilde{s}(h_A) \leq \tilde{s}(v_A) + w_A(v)$ follows by the feasibility of $\tilde{s}$.
\end{itemize}
Now, a weighted directed graph $G=(V,E)$ with $E\triangleq\{e_A\mid A\in\A\}$ is defined. 
$G$ is integral-weighted and conservative graph since it admits $\tilde{s}$ as a potential function.
Therefore, $G$ admits an integral potential function $s:V \rightarrow \{-T,-T+1, \ldots, T-1, T\}$. 
Indeed, such a function $s$ is obtained by applying the Bellman-Ford algorithm on $G$.
To conclude, we observe that $s$ is also an integral feasible scheduling for $\H$.
\end{proof}

\noindent The following theorem states that \GTNC is \NP-complete.
\begin{theorem}\label{Teo:npcompleteness} 
\GTNC is an \NP-complete problem even if input instances $\H=(V, \A)$ are restricted to satisfy $w_A(\cdot) \in\{-1, 0, 1\}$ and $|H_A|, |T_A|\leq 2$ for every $A\in\A$.
\end{theorem}
\begin{proof}
If $\H=(V,\A)$ is integral-weighted and consistent, then it admits an integral feasible scheduling $s:V\rightarrow\{-T, \ldots, T\}$ by Lemma~\ref{Lem:int_sched}.   
Moreover, any such feasible scheduling can be verified in polynomial time \wrt the size of the input; hence, \GTNC is in \NP.

To show that the problem is \NP-hard, we describe a reduction from 3-SAT. 

Let us consider a boolean 3-CNF formula with $n\geq 1$ variables and $m\geq 1$ clauses: 
\[\varphi(x_1, \ldots, x_n) = \bigwedge_{i=1}^m (\alpha_i \vee \beta_i \vee \gamma_i)\]
where $\C_i = (\alpha_i \vee \beta_i \vee \gamma_i)$ is the $i$-th clause of $\varphi$
and  each $\alpha_i,\beta_i,\gamma_i\in \{x_j, \overline{x}_j\mid 1\leq j\leq n\}$ is either a positive or a negative literal. 

We associate to $\varphi$ a \TN $\H_{\varphi}=(V, \A)$, where each boolean variable $x_i$ occurring in $\varphi$ is represented by two nodes, $x_i$ and
$\overline{x}_i$.
$V$ also contains node $z$ that represents the reference initial node for the \TN $\H_{\varphi}$, \ie the first node that has to be executed.
For each pair $x_i$ and $\overline{x}_i$, $\H_{\varphi}$ contains a pair of hyperarc constraints as depicted in \figref{FIG:Var_i}: one with multi-head
$\{x_i,\overline{x}_i\}$ and tail in $z$ and the other multi-tail $\{x_i,\overline{x}_i\}$ and head in $z$. If $\H_{\varphi}$ is consistent, the pair of
hyperarcs associated to $x,\neg x$ assures that $\H_{\varphi}$ admits a feasible scheduling $s$ such that $s(x_i)$ and $s(\overline{x}_i)$ are coherently set
with values in $\{0,1\}$ (see Lemma~\ref{Lem:int_sched}).
In this way, $s$ is forced to encode a truth assignment on the $x_i$'s.
The \TN $\H_{\varphi}$ contains also a node $\C_j$ for each clause $\C_j$ of $\varphi$; each node $\C_j$ is connected by a multi-tail hyperarc with head in
$\C_j$ and tails over the literals occurring in $\C_j$ and by two standard and opposite arcs with node $z$ as displayed in \figref{FIG:Cl_i}.
Such setting of arcs assures that if $\H_{\varphi}$ admits a feasible scheduling $s$, then $s$ assigns value 1 at least to one of the node representing the
literals connected with the hyperarc.

%
More formally, $\H_{\varphi}=(V, \A)$  is defined as follows:
\begin{itemize}
	\item $V=\{z\}\cup \{x_i \mid 1\leq i \leq n\}\cup \{\overline{x}_i\mid 1\leq i \leq n\}\cup\{\C_j \mid 1\leq j\leq m\}$;
	\item $\A = \bigcup_{i=1}^n \text{Var}_i \cup \bigcup_{j=1}^m\text{Cla}_j$, where:
	\begin{itemize}
		\item $\text{Var}_i=\Big\{ (z, x_i, 1), (x_i, z, 0), (z, \overline{x}_i, 1), (\overline{x}_i, z, 0), \\
			 (\{x_i, \overline{x}_i\}, z, [w(x_i), w(\overline{x}_i)] = [-1, -1] ),\\ 
			 (z, \{x_i, \overline{x}_i\}, [w(x_i), w(\overline{x}_i)] = [0, 0] )\Big\}$.\\ 
			This defines the variable gadget for $x_i$ as depicted in \figref{FIG:Var_i}; 
		\item $\text{Cla}_j=\Big\{(z, \C_j, 1), (\C_j, z, -1), \\
					(\{\alpha_j, \beta_j, \gamma_j\}, \C_j, [ w(\alpha_j), w(\beta_j), w(\gamma_j) ] = [0,0,0] )\Big \}$.\\ 
			This defines the clause gadget for clause $\C_j = (\alpha_i \vee \beta_i \vee \gamma_i)$ as depicted in \figref{FIG:Cl_i}.
	\end{itemize}

\end{itemize}

\begin{figure}[tb]

\subfloat[Gadget for a 3-SAT variable $x_i$.]{
\begin{tikzpicture}[arrows=->,scale=1,node distance=2 and 2]
	\node[node, label={below:$[0]$}] (zero) {$z$};
	\node[node, above right=of zero] (nX) {$\overline{x}_i$};
	\node[node, above left=of zero] (X) {$x_i$};
	\draw[] (zero) to [bend left=40] node[below] {$1$} (X); 
	\draw[] (X) to [bend left=40] node[above] {$0$} (zero);
	\draw[dashed, thick] (zero) to [bend right=15] node[above] {$0$} (X);
	\draw[dotted, thick] (X) to [bend right=15] node[below] {$-1$} (zero);
	\draw[] (zero) to [bend left=40] node[above] {$1$} (nX); 
	\draw[] (nX) to [bend left=40] node[below] {$0$} (zero); 	
	\draw[dashed, thick] (zero) to [bend left=15] node[above] {$0$} (nX);
	\draw[dotted, thick] (nX) to [bend left=15] node[below] {$-1$} (zero);
\end{tikzpicture}\label{FIG:Var_i}
}
\quad
\subfloat[Gadget for a 3-SAT clause $\C_j = (\alpha_j \vee \beta_j \vee \gamma_j)$ where each $\alpha_j, \beta_j, \gamma_j$ is a positive or negative literal.]{
\begin{tikzpicture}[arrows=->,scale=.85,node distance=1.5 and 2]
	\clip (-94pt,-155pt) rectangle (94pt, 25pt);
	\node[node,label={above:$[1]$}] (one) {$\C_j$};
	\node[node,below =of one] (beta) {$\beta_j$};
	\node[node,left=of beta] (alpha) {$\alpha_j$};
	\node[node,right=of beta] (gamma) {$\gamma_j$};
 	\node[node,label={below:$[0]$}, below=of beta] (zero) {$z$};
 	\coordinate (fakeL) at ($(alpha.west)+(-1,0)$);
 	\coordinate (fakeR) at ($(gamma.east)+(1,0)$);
 	\draw[] (zero.west) .. controls ($(fakeL)+(0,-10mm)$) and ($(fakeL)+(0,10mm)$) .. node[left,pos=.3] {$+1$} (one);
  	\draw[] (one) .. controls ($(fakeR)+(0,10mm)$) and ($(fakeR)+(0,-10mm)$) .. node[right,pos=.3] {$-1$} (zero.east); 
	\draw[] (alpha) to [bend right=20] node[timeLabel,below] {} (zero); 
	\draw[] (zero) to [bend right=20] (alpha);
	\draw[dotted, thick] (alpha) to [bend right=5] (zero);
	\draw[dashed, thick] (zero) to [bend right=5] (alpha);
	\draw[dotted, thick] (alpha) to [bend right=25] node[timeLabel,below] {$0$} (one.south west);
	\draw[] (zero) to [bend left=25] (beta); 
	\draw[] (beta) to [bend left=25] (zero); 	
	\draw[dotted, thick] (beta) to [bend right=10] (zero);
	\draw[dashed, thick] (zero) to [bend right=10] (beta);
	\draw[dotted, thick] (beta) to [] node[timeLabel,right] {$0$} (one.south);	
	\draw[] (gamma) to [bend left=20] (zero); 
	\draw[] (zero) to [bend left=20] (gamma);
	\draw[dotted, thick] (gamma) to [bend left=5] (zero);
	\draw[dashed, thick] (zero) to [bend left=5] (gamma);
	\draw[dotted, thick] (gamma) to [bend left=25] node[timeLabel,below] {$0$} (one.south east);
\end{tikzpicture}\label{FIG:Cl_i}
}
\caption{Gadgets used in the reduction from 3-SAT to \GTNC.}\label{fig.gadget}
\end{figure}
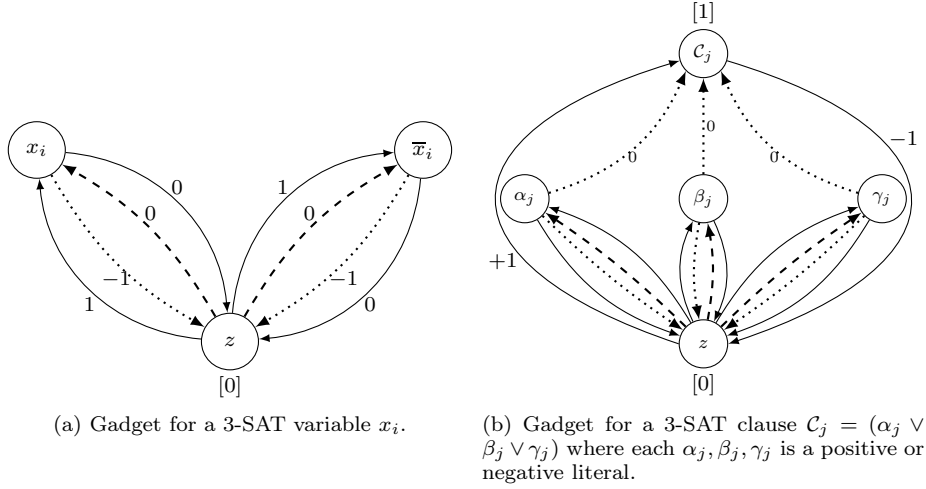

Notice that $|V|=1+2n+m=O(m+n)$ and $m_{\A}=8n+5m=O(m+n)$; therefore, the transformation is linearly bounded in time and space.

We next show that $\varphi$ is satisfiable if and only if $\H_{\varphi}$ is consistent.

Any truth assignment $\nu:\{x_1, \ldots, x_n\}\rightarrow \{\texttt{true}, \texttt{false}\}$ satisfying $\varphi$ can be translated into a feasible scheduling
$s:V\rightarrow \Z$ of $\H_{\varphi}$ as follows.
For node $z$, let $s(z)=0$, and let $s(\C_j)=1$ for each $j=1, \ldots, m$; then, for each $i=1, \ldots, n$, let $s(x_i) = 1$ and $s(\overline{x}_i) = 0$ if the
truth value of $x_i$, $\nu(x_i)$, is \texttt{true}, otherwise let $s(x_i) = 0$ and $s(\overline{x}_i) = 1$.
It is easy to verify that, using this scheduling $s$, all the constraints comprising each single gadget are satisfied and, therefore, the network is
consistent.

Vice versa, assume that $\H_{\varphi}$ is consistent.
Then, it admits an integral feasible scheduling $s$ by Lemma~\ref{Lem:int_sched}.
After the translation $s(v) = s(v) - s(z)$, we can assume that $s(z)=0$.
Hence, $s(\C_j) = 1$ for each $j=1,\ldots, m$, as enforced by the two standard arcs incident at $\C_j$ in the clause gadget, and $\{s(x_i),s(\overline{x}_i)\} =
\{0,1\}$ for each $i=1,\ldots, n$, as enforced by the constraints comprising the variable gadgets.
Therefore, the feasible scheduling $s$ can be translated into a truth assignment $\nu:\{x_1, \ldots, x_n\}\rightarrow \{\texttt{true}, \texttt{false}\}$ defined
by $\nu(x_i)=\texttt{true}$ if $s(x_i)=1$ (and $s(\overline{x}_i)=0$); $\nu(x_i)=\texttt{false}$ if $s(x_i)=0$ (and $s(\overline{x}_i)=1$) for every $i = 1,
\ldots, n$.

To conclude, we observe that any hyperarc $A\in\A$ of $\H_{\varphi}$ has weights $w_A(\cdot)\in\{-1, 0, 1\}$ and size $|A|\leq 3$.
Since any hyperarc with three heads (tails) can be replaced by two hyperarcs each having at most two heads (tails), 
the consistency problem remains \NP-Complete even if  $w_A(\cdot)\in\{-1, 0, 1\}$ and $|A|\leq 2$ for every $A\in \A$.
\end{proof}

Theorem~\ref{Teo:npcompleteness} motivates the study of consistency problems on hypergraphs having either 
only multi-head or only multi-tail hyperarcs. In the former case, the consistency problem is called 
\HTNC, while in the latter it is called \TTNC.
In the following theorem we observe that the two problems are inter-reducible, 
\ie we can solve consistency for any one of the two models in $f(m,n,W)$ time whenever we  
have a $f(m,n,W)$ time procedure for solving consistency for the other one.
 
\begin{theorem}\label{Teo:inter-reducitble-HTNs} 
\HTNC and \TTNC are inter-reducible by means of $\log$-space, linear-time, local-replacement reductions.
\end{theorem}
\begin{proof} 
We show the reduction from multi-tail to multi-head hypergraphs; the other direction is symmetric.
Informally, what we will do is to reverse all the arcs (so that what was multi-tail becomes multi-head), and, contextually, we invert the time-axis (to account 
for the inversion of the direction of all arcs).

Let $\H=(V,\A)$ be a multi-tail hypergraph, we associate to $\H$ a multi-head hypergraph $\H'=(V, \A')$ by reversing all multi-tail hyperarcs. 
Formally, we define
\[
   {\A}' = \{(v, S, w)\mid (S, v, w)\in \A\}.
\]
For example, a standard arc $(t,h,w)\in \A$ is transformed into a reversed standard arc $(h,t,w)$ in ${\A}'$ while a hyperarc with two weighted tails $t_1$ and
$t_2$ becomes a hyperarc having $t_1$ and $t_2$ as its two weighted heads.
 

Now, $\H$ is consistent if and only if $\H'$ is consistent.
To prove it, we note that each scheduling $s$ for $\H$ can be associated, with a flip of the time direction, to the scheduling $s' \triangleq -s$.
Then,  it holds that $s$ is feasible for $\H$ if and only if $s'$ is feasible for $\H'$.
Indeed, $s$ satisfies the constraint represented by an hyperarc $A=(T_A, h_A, w_A)\in\A$, namely 
\[
   s(h_A)\leq \max_{v\in T_A} \{s(v) + w_A(v) \},
\] 
or, equivalently 
\[ 
	-s(h_A)\geq \min_{v\in T_A}\{-s(v) - w_A(v) \}, 
\] 
if and only if $s'$ (that is, $-s$) satisfies the constraint represented by the
reversed hyperarc $A'=(h_A, T_A, w_A)$, namely 
\[
   s'(h_A)\geq \min_{v\in T_A}\{s'(v)-w_{A'}(v)\}.
\]
\end{proof}

In the remainder of this work we shall adopt the multi-head hypergraph as our reference model.
Hence, when considering hypergraphs and \TN{s}, we will be implicitly referring to multi-head hyperarcs.
Notably, we consider the following specialized notion of consistency for \TN{s}.
\begin{definition}[\TNC] 
Given a (multi-head) \TN \mbox{$\H=(V,\A)$}, decide whether there exists a scheduling \mbox{$s:V \rightarrow \R$} such that:
\begin{equation}
   s(t_A) \geq \min_{v\in H_A} \{s(v) - w_A(v) \}\quad \forall A\in\A.
\end{equation}
\end{definition}

\begin{remark}\label{Rem:PLconvex}
Notice that this notion of consistency for \TN{s} is a strict generalization of \STN one.
In general, the feasible schedules of an \STN are the solutions of a linear system and, therefore, they form a convex polytope.  
Since an \STN may be viewed as a \TN, the space of feasible schedules of an \STN can always be described as the space of feasible schedules of a \TN.
The converse is not true because feasible schedules for a \TN do not form a convex polytope.
Let us consider, for example, a \TN of just three nodes $x_1$, $x_2$, $x_3$ and one single hyperarc with heads $\{x_1,x_2\}$ and tail $x_3$ expressing the 
constraint $x_3\geq \min\{x_1,x_2\}$; 
$(0,2,2)$ and $(-2,0,2)$ are both admissible schedules, but $(1,1,0) = \frac{1}{2}(0,2,2) - \frac{1}{2}(-2,0,2)$ is not an admissible schedule. 
In conclusion, the \STN model is a special case of the Linear Programming paradigm, whereas the \TN model is not.
\end{remark}

In the rest of this section, we extend the characterization of \STN consistency recalled in Section~\ref{sect:background} to \TN{s}.

\begin{definition}[Reduced Slack Value $w^{p}_A(v)$]
With reference to a potential $p:V\rightarrow \R$, we define, for every arc $A\in \A$ and every $v\in H_A$, 
the \textit{reduced slack value} $w^{p}_A(v)$ as $w_A(v) + p(t_A) - p(v)$
and the \textit{reduced slack} $w^{p}_A$ as
\[
    w^{p}_A \triangleq \max\{w^{p}_A(v) \mid v\in H_A\}.
\]
A potential $p$ is said to be \emph{feasible} if and only if $ w^{p}_A \geq 0$ for every $A\in \A$.
\end{definition}

Again, as it was the case for \STN{s}, a mapping $f:V \rightarrow \R$ is a feasible potential if and only if it is a feasible schedule.
In order to better characterize feasible schedules, we introduce a notion of \emph{negative cycle}.
\begin{definition}[Negative Cycle]
Given a multi-head \TN $\H=(V, \A)$, a \textit{cycle} is a pair $(S,\C)$ with $S\subseteq V$ and $\C \subseteq \A$ such that:
\begin{enumerate}
  \item $S = \bigcup_{A\in \C} (H_A\cup \{t_A\})$ and $S\neq\emptyset$;
  \item $\forall v\in S$ there exists an unique $A\in \C$ such that $t_A = v$.
\end{enumerate}
Moreover, we let $a(v)$ denote the 
unique arc $A\in \C$ with $t_A = v$ as required in previous item~2.
Every infinite path in a cycle $(S,\C)$ contains, at least, one \textit{finite cyclic sequence} $v_i, v_{i+1}, \ldots, v_{i+p}$,
where $v_{i+p} = v_i$ is the only repeated node in the sequence.
A cycle $(S,\C)$ is \textit{negative} if for any finite cyclic sequence $v_1, v_{2}, \ldots, v_{p}$, 
it holds that
\[
   \sum_{t=1}^{p-1} w_{a(v_t)}(v_{t+1}) < 0.
\] 
\end{definition}

There are two results about negative cycles as stated in the following lemmas.
\begin{lemma}\label{lem:nc}
   A \TN with a negative cycle admits no feasible schedule.
\end{lemma}
\begin{proof}
	By contraposition.
	Let $\H$ be a consistent \TN and let $p$ be a feasible potential for $\H$.
	Also, let $(S,\C)$ be any cycle of $\H$;
	we will show that $(S,\C)$ is not negative.
	For every $A\in \C$, let $h_A$ be the head of $A$ with maximum reduced slack value:
	\[ 
	h_A \triangleq \arg \max_{v\in H_{A}} \{w^p_A(v) \}.
	\]

	Let us consider the infinite path in $(S,\C)$ built choosing, at each node $v_t$, $h_{a(v_t)}$ as the following node.
	As already seen, such a path contains at least one finite cyclic sequence $v_h, v_{h+1}, \ldots, v_k$ with $v_k=v_h$.
	The sum of weights of the finite cyclic sequence is given by 
	\[
            \sum_{t=h}^{k-1} w_{a(v_t)}(v_{t+1}) = \sum_{t=h}^{k-1} w^p_{a(v_t)}(v_{t+1})
        \]
	for every potential $p$;
	since $p$ is feasible, all terms of the last sum are non-negative.
	It follows that $(S,\C)$ is not negative.
%
\end{proof}

At first sight, it may appear that checking whether $(S,\C)$ is a negative cycle might take exponential time since one should check a possibly exponential
number of cyclic sequences.
The next lemma shows instead that it is possible to check the presence of negative cycle in polynomial time.

\begin{lemma}
Let $(S, \C)$ be a cycle in a \TN.
Then checking whether $(S,\C)$ is a negative cycle can be done in polynomial time.
\end{lemma}
\begin{proof}
Consider the weighted graph $G=(S,\cup_{t\in S} A_t)$ where each hyperarc $a(t)$,
  for every $t\in S$, is transformed into a set of standard arcs as follows:
    \[a(t)\leadsto A_t \triangleq \{(t,v,-w_{a(t)}(v)) \mid v\in H_{a(t)})\},\; \forall\, t\in S.\]

Notice that $G$ is thus an \STN. Checking whether $(S,\C)$ is a negative cycle amounts to check whether all cycles in $G$ have strictly positive weight.
To do this, firstly, a potential function $\pi$ for $G$ is determined by Bellman-Ford algorithm.
If the algorithm returns a negative cycle instead of $\pi$, then there is \emph{no} negative cycle in $(S,\C)$ and the check ends.

Otherwise, since $w(C)=w^\pi(C) \geq 0$ for every cycle $C$ of $G$, it is necessary to verify that no cycle in $G$ has $w^\pi(C) = 0$.
This check can be done by verifying the acyclicity of the subgraph of $G$ comprising only arcs $a$ of $G$ with $w^{\pi}(a) = 0$.
The check that a graph is acyclic can be done in linear time by a depth first visit~\cite{Cormen01}.
\end{proof}

A hypergraph $\H$ is called \emph{conservative} when it contains no negative cycle.
In the next sections we will provide a pseudo-polynomial time algorithm that always returns either a feasible scheduling or a negative cycle, thus extending the
validity of the classical good-characterization of \STN consistency to general \TN consistency.
Here, we anticipate the statement of the main result in order to complete this general introduction of \TN{s}.
\begin{theorem}\label{teo:charCons}
   A \TN $\H$ is \textit{consistent} if and only if it is \textit{conservative}.
   Moreover, when all weights are integral, then $\H$ admits an integral scheduling if and only if it is conservative.
\end{theorem}
\begin{proof}
If $\H$ is consistent, then it is conservative by Lemma~\ref{lem:nc}.
If $\H$ is not consistent, then there is a negative cycle as shown in Theorem~\ref{Teo:MainAlgorithms}-(3).
The existence of an integral scheduling when all weights are integral is guaranteed by Lemma~\ref{Lem:int_sched}. 
\end{proof}

%% file: sectMPG.tex
In this section, we propose an introduction to Mean Payoff Games (\MPG{s}) tailored to the needs of the present work.
MPGs represent a well-studied model for representing some kinds of two-player dynamics and we will show in Section 6 that there is a substantial equivalence
between the MPG and the HyTN model, which will allow us to exploit some important algorithmic and structural results.

An \MPG is a weighted directed graph $G=(V_0\cupdot V_1, E)$ whose node set $V$ is partitioned into two disjoint sets $V_0$ and $V_1$, where, for $p=0,1$, the
nodes in $V_p$ are those under control of Player~$p$.
Even with these graphs we have no loops and no parallel arcs.
It is also assumed that every node has at least one outgoing arc.
Notice that, in general, $(V_0,V_1)$ does not need to be a bipartiton of $G$, \ie $E$ may contain arcs with both endpoints in $V_0$, or with both endpoints in
$V_1$.

Each play starts with a pebble placed at some node $v_0\in V_0\cupdot V_1$ and consists in a sequence of moves.
Move~$t$ begins with the pebble placed in node $v_{t-1}$ and is played by the Player~$p$ such that $v_{t-1}\in V_p$:
the player chooses any arc $e\in E$ with tail $t_e = v_{t-1}$ and moves the pebble along $e$; at the end of the move the pebble is in node $v_t = h_e$.
The game ends as soon as $v_t = v_{t'}$ for some $t>t'$, \ie when the pebble comes back to an already visited node $v_{t'}$.
At this point, the pebble has traversed a cyclic sequence of arcs $e_{t'+1}, \ldots, e_t$ and Player~$0$ ``pays'' to Player~$1$ the average weight of the
visited cycle:
\[ 
	\frac{1}{t-t'}\sum_{i=t'+1}^t w(e_i).
\] 
If this amount is negative, then Player~$0$ \emph{wins} the game, otherwise the winner is Player~$1$.

A \textit{strategy} for Player~$p$ is a mapping that, given all the previous visited nodes and the current node, returns which node has to be visited in the
next move;
a strategy is said to be  \textit{positional} (or \emph{memoryless}) if it depends only on the current position $v_t$ and does not take into account all the
previous history.
If $s\in V_0\cup V_1$ and Player~$p$ has a strategy leading him to win any possible play starting at $v_0 = s$, then we say that $s$ is a \textit{winning start
position} for Player~$p$.
We denote by $W_p$ the set of winning start positions for Player~$p$.
A \emph{winning strategy} for Player~$p$ leads Player~$p$ to win every play started from any node in $W_p$.
Since these finite games are zero-sum, \ie what won by a player is what lost by the other one, then they admit a \textit{game value} $\nu$:
for each start position $s\in V$ of the game, there exists a $\nu_s\in\R$ such that Player~$0$ has a strategy ensuring payoff at most $\nu_s$, while Player~$1$
has a strategy ensuring payoff at least $\nu_s$.

It is worthwhile to consider an infinite variant of the model, in which the game does not stop, and continues for an infinite number of steps.
In this model, Player~$1$ wants to maximize the limit inferior of the average weight: \[\liminf_{n\rightarrow\infty}\frac{1}{n}\sum_{t=1}^n w(v_{t-1},v_t)\] 
Symmetrically, Player~$0$ wants to minimize the limit superior of the same average weight: \[\limsup_{n\rightarrow\infty}\frac{1}{n}\sum_{t=1}^n w(v_{t-1},v_t)\]

In their \emph{Determinacy Theorem}, Ehrenfeucht and Mycielski~\cite{EhrenfeuchtMycielski:1979} proved that any infinite game admits a value $\nu^{\infty}$, and
that this value equals the one of the finite counterpart game on every start position, \ie $\nu^{\infty}_s = \nu_s$ for every $s\in V_0\cup V_1$.
Moreover, they proved the existence of positional strategies which are \emph{optimal} for both variants of the model:
when Player~$p$ limits himself to an optimal strategy $\pi_p$, \ie when, for every $v\in V_p$, he disregards all arcs with tail in $v$ except the one with head
in $\pi_p(v)$, then he will secure himself the optimal payoff $\nu$ in every play, finite or infinite, however the adversary plays.
The graph $G_{\pi_p}$ obtained from $G$ by dropping all arcs with tail in $V_p$ not prescribed by $\pi_p$ is called the \emph{projection} of the game $G$ on
$\pi_p$, and is a solitaire game whose value can be easily computed by means of a simple variant of Bellman-Ford algorithm.
Therefore, the  Ehrenfeucht and Mycielski's results are already sufficient for determining a simple exponential time algorithm computing the node values
and the two optimal positional strategies in an \MPG: the algorithm consists in evaluating each possible strategy for one of the two players as a solitaire
game for determining the optimal one.
In the literature there are many local search algorithms that explore this space in a more efficient way~\cite{BV07,Sch08,BrimCha12,ScheweTV15} and some of
them have been proven to be practically efficient in many settings by experiments~\cite{Sch08,BrimCha12}.
Moreover, the global optimization problem of computing the best strategies for one player, according to a given metric, has been shown to have the property that
every local optimum is also a global one for many complete metrics~\cite{BV07}.

\sloppy
As another line of research, Zwick and Paterson~\cite{ZwickPaterson:1996} proposed pseudo-polynomial time algorithms for
computing values of games, as well as positional optimal strategies. 
In particular, they considered the following four algorithmic problems:
\begin{enumerate}
	\item \MPGD{$(\nu,s)$}: given a real number $\nu$ and a start position $s$, decide whether $\nu_s \geq \nu$;
	\item \MPGT{($T$)}: given a real number $T$, determine for which nodes $s\in V$ it holds that $\nu_s \geq T$;
	\item \MPGV: compute the optimal values $\nu_s$ for all $s\in V$.
	\item \MPGS: assuming $\nu_s \geq 0$ ($\nu_s < 0$) for every $s\in V$, \emph{synthesize} a positional winning strategy for Player~$1$ (Player~$0$);
\end{enumerate}
and they proved the following theorem:
\fussy

\begin{theorem}[\hspace{-1sp}\cite{ZwickPaterson:1996}]\label{Teo:ZPSummary}
Let $G=(V, E)$ be a mean payoff game.
Assume all weights are integers and let $W=\max_{e\in E}{|w(e)|}$. 
Then the following hold:
\begin{enumerate}
	\item \MPGT{{$(T)$}} can be solved in time $\Ord(|V|^2 |E| \, W)$ when $T\in \Z$, whereas it can be solved in time $\Ord(|V|^3 |E| \, W)$ when $T\in \R$;
	\item \MPGV can be solved in time $\Ord(|V|^3 |E| \, W)$;
	\item \MPGS can be solved in time $\Ord(|V|^4 |E| \log(|E|/|V|) \, W)$.
\end{enumerate}
\end{theorem}
Then, they observed that \MPGD is the basic decision problem for \MPG{s} in the sense that several natural questions for \MPG{s},
like evaluating the value $\nu_s$ for every node $s$ or constructing the optimal positional strategies, may all be Turing-reduced to it.
They also pointed out that the existential results of Ehrenfeucht and Mycielski~\cite{EhrenfeuchtMycielski:1979} 
already implies that $\MPGD\in \NP\cap \coNP$ and asked whether there might exist a strongly polynomial time decision procedure.
Proving the existence of such algorithm is an open problem~\cite{brim2011faster}.
Finally, they showed how to reduce mean payoff games to other important families of games on graphs, like discounted payoff games and simple stochastic games.

The complexity status of \MPGD has been since updated by proving that it lays in $\UP\cap \coUP$ by Jurdzi\'nski in \cite{Jur98}.

In recent years, some other interesting results have been proven.
Notably, in 2007 Lifshits, Pavlov~\cite{LifshitsPavlov:2007} proposed a \emph{potential theory} for \MPG{s} and in 2011 Brim et al.~\cite{brim2011faster}
obtained faster algorithms exploiting results obtained in the the fields of \emph{Energy Games} and \emph{energy progress measures}, which are intimately
related to the potentials studied in~\cite{LifshitsPavlov:2007}.


Their algorithmic results are summarized in the following theorem.

\begin{theorem}[\hspace{-1sp}\cite{brim2011faster}]\label{Teo:FasterAlgo}
	For \MPG{s} in which all weights are integers {and for $T\in\Z$}, the Value Iteration Algorithm~\cite{brim2011faster}
	solves \MPGT{$(T)$} and \MPGS in time $O(|V|\,|E|\,W)$, where $W=\max_{e\in E}{|w(e)|}$. 
\end{theorem}

We remark that both the algorithm of Paterson and Zwick~\cite{ZwickPaterson:1996}
and the Value Iteration Algorithm~\cite{brim2011faster}
prescribe well defined procedures even if the weights on the arcs are real values. What is lost in running these algorithms on real weights is only the pseudo-polynomial upper bound on their running time. 

For our purposes, the family of pseudo-polynomial algorithms for \MPG{s} is the best option. Indeed, in most of temporal workflow graphs
all weights are expressed by integers of relatively small magnitude with respect to the intrinsic temporal granularity of the considered workflow.
For example, in a workflow containing temporal distance constraints of days, the commonly adopted temporal granularity
is the ``minute'' (m) and, therefore, all weights can be assumed to be less than $10^4$ as order of magnitude.
In such circumstances, Brim's algorithm offers the guarantee to terminate within short computation times.
For these reasons we opted for integrating the procedures of Zwick and Paterson, as well as the faster procedures of Brim et al.\cite{brim2011faster}, in order
to efficiently solve instances of \TNC and compute feasible schedules.

Furthermore, as will be discovered in the experimental section, if these algorithms are suitably adapted---so as to allow them to terminate earlier as soon as 
certain evidences of inconsistency have been collected---then their observed behavior outperforms by orders of magnitude what predicted by their
theoretical pseudo-polynomial bounds even on input instances containing very large integer values. 

Based on these findings, we think that these pseudo-polynomial algorithms are to be considered (and probably adopted) even for solving \TN{} instances where
weights are floating point values whose magnitudes may differ in a significant way.
In case the running time results to be unacceptable for a real application, one could then consider the possibility to round the weights to integer values.
This rounding would clearly require special care: a very accurate approximation might lead to very high computation times 
while a gross approximation might not represent the original instance in a correct way.


%% file: sectReductions.tex
This section presents the direct connection and the computational equivalence between \MPGT and \TNC.
The equivalence is formally proven by offering one reduction in each direction.

The reduction of \TNC to \MPGT allows to apply, in the context of \TN{s}, 
any of the algorithms known for \MPG{s}, included the exponential and subexponential ones.

Vice versa, in consideration of the fact that the $\MPGD \stackrel{?}{\in} \p$ question is an open
problem~\cite{ZwickPaterson:1996,Jur98,Sch08,BV07,brim2011faster}, 
the reduction of \MPGD to \TNC confirms that \TNC offers an algorithmically more ambitious
and mathematically steeper generalization of \STN-Consistency  (see also Remark~\ref{Rem:PLconvex}).
Moreover, the reduction gives a further evidence that, within \STN{s}, 
a new algorithmic approach is necessary in order to manage temporal aspects of event like
the synchronization one presented in the Introduction.

Let us start considering the first reduction.
\begin{theorem}\label{Teo:TN2MPG}
   There exists a $\log$-space\footnote{A strong and basic-form 
	of reduction introduced by Papadimitriou in~\cite{Papadimitriou:1994}.},
   linear-time, local-replacement\footnote{A restricted kind of Karp reduction introduced in ~\cite{GareyJohnson:1979}.}
   reduction from \TNC to \MPGT.
\end{theorem}

Since this reduction plays a main role in the algorithmic solutions proposed in this paper, 
we firstly describe how it works and, secondly, we prove its correctness by means of two lemmas, 
Lemma~\ref{lem:mainReductionHcons} and Lemma~\ref{lem:mainReductionGwin}.

The reduction goes as follows.

Let $\H=(V, \A)$ be a \TN. We assume that every $v\in V$ is the tail of some arc $A\in \A$.
This assumption is not a restriction since, if $\H$ contains a \textit{sink node} $v$, 
\ie a node $v$ with no arc $A\in \A$ having tail in it, then $\H$ is
consistent if and only if so is $\H_v$, the \TN obtained from $\H$ by removing node $v$ and every hyperarc having $v$ as an head.
Indeed, any feasible scheduling $s:V\mapsto \R$ for $\H$, once projected onto $V\setminus \{v\}$, 
gives a feasible scheduling for $\H_v$ since every constraint
involving $v$ has been dropped and no constraint has been added; 
conversely, any feasible scheduling $s$ for $\H_v$ can be easily extended to a feasible
scheduling $s$ for $\H$ by exploiting the property of $v$ being a sink node: 
it is sufficient to set $s(v) \triangleq \min\{ s(t_A) - w_A(v) \mid A\in \A, v\in H_A \}$.

Now, let us consider a mean payoff game $G_{\H}=(V_0\cupdot V_1, E)$ where: 
(1) $V_0 = V$, $V_1=\A$, nodes in $V_0$ are colored by \emph{black} while nodes in $V_1$ are colored by \emph{white}, and (2) 
for each $A\in \A$, the following weighted arcs are added to $E$:
\begin{itemize}
	\item an arc of weight $0$ from the black node $t_a$ to the white node $A$, \ie arc $(t_A, A, 0)$;
	\item for each head node $h\in H_A$, an arc of weight $w_A(h)$ from the white node $A$ to the black node $h$, \ie arc $(A,h,w_A(h))$.
\end{itemize}
In short, $G_{\H}=(V_0\cupdot V_1, \A)$, with 
$V_0 = V$,
$V_1 = \A$,
$E = \{ (t_A,A,0) \mid A\in \A \}  \cup \{ (A,h,w_A(h)) \mid A\in \A, h\in H_A \}$. 
\figref{FIG:htn-mpg1} depicts how a hyperarc is transformed into a MPG subnetwork while \figref{FIG:PseudocodeReduction-htn-mpg} reports a pseudocode for the whole construction process, \ie Algorithm~\ref{ALGO:PseudocodeReduction-htn-mpg}.
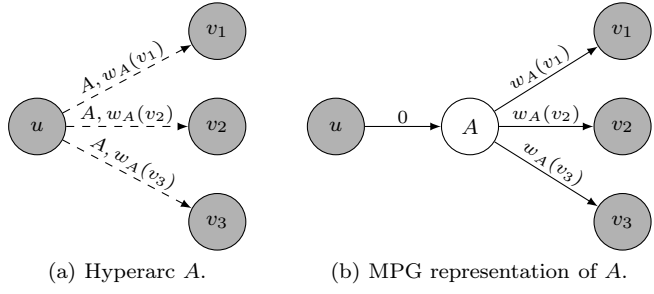
\begin{figure}[tb]
\centering
\subfloat[Hyperarc $A$.]{
\begin{tikzpicture}[arrows=->,scale=1,node distance=.5 and 1]
    \node[node,blackNode] (v1) {$v_1$};
    \node[node,blackNode,below=of v1] (v2) {$v_2$};
    \node[node,blackNode,below=of v2] (v3) {$v_3$};
	\node[node,blackNode,left=of v2, xshift=-5ex] (u) {$u$};
%
   	\draw[dashed] (u) to node[timeLabel,above,sloped] {$A, w_A(v_1)$} (v1);%
   	\draw[dashed] (u) to node[timeLabel,above,sloped] {$A, w_A(v_2)$} (v2);%
   	\draw[dashed] (u) to node[timeLabel,above,sloped] {$A, w_A(v_3)$} (v3);%
\end{tikzpicture}
}
\qquad
\subfloat[MPG representation of $A$.]{
\begin{tikzpicture}[arrows=->,scale=1,node distance=.5 and 1]
    \node[node,blackNode] (v1) {$v_1$};
    \node[node,blackNode,below=of v1] (v2) {$v_2$};
    \node[node,blackNode,below=of v2] (v3) {$v_3$};
	\node[node,left=of v2, xshift=-2ex] (a) {$A$};
	\node[node,blackNode,left=of a] (u) {$u$};
%
   	\draw[] (a) -- node[timeLabel,above,sloped] {$w_A(v_1)$} (v1);%
   	\draw[] (a) -- node[timeLabel,above,sloped] {$w_A(v_2)$} (v2);%
   	\draw[] (a) -- node[timeLabel,above,sloped] {$w_A(v_3)$} (v3);%
   	\draw[] (u) -- node[timeLabel,above,sloped] {$0$} (a);%
\end{tikzpicture}
}
\caption{The conversion of a hyperarc into a white MPG node and its incident arcs.}
\label{FIG:htn-mpg1}
\end{figure}

\begin{figure}[tb]
\removelatexerror
\begin{algorithm}[H]\label{ALGO:PseudocodeReduction-htn-mpg}
\caption{\texttt{makeACorrespondingGame}$(\H)$}
\tcp{a \TN $\H = (V,\A)$}
$V_0 \leftarrow  V$\;
$V_1 \leftarrow  \A$\;
$E \leftarrow \emptyset$\;
\ForEach{ $A\in \A $} {
   $E \leftarrow E\cup (t_A,A,0)$\;
   \ForEach{$h\in H_A$} {
      $E\leftarrow  E\cup (A,h,w_A(h))$\;
   }
}
\KwOut{The \MPG $G_{\H} = (V_0\cupdot V_1, E)$}
\end{algorithm}
\caption{The algorithm implementing the reduction from a \TN to the corresponding \MPG.}
\label{FIG:PseudocodeReduction-htn-mpg}
\end{figure}

$G_{\H}$ has $|V|+|\A|$ nodes and $O(m)$ arcs and can be constructed in linear time.
Moreover, $G_{\H}$ is a bipartite graph with bipartition $(V_0, V_1)$ and it has been obtained from $\H$ by a simple local replacement rule:
replace every hyperarc $A\in \A$ by a claw subgraph as depicted in \figref{FIG:htn-mpg1}.
For each single object, it is necessary only to manage a constant number of indexes, each of them having a polynomial size;
thus the reduction is log-space.
\figref{FIG:htn-mpg} depicts an MPG obtained applying 
the reduction to the motivating example \TN depicted in \figref{FIG:wf-hstn};
we remark that the MPG depicted in~\figref{FIG:htn-mpg} has been obtained by considering the (equivalent) multi-head \TN transformation of the multi-tail \TN shown in \figref{FIG:wf-hstn},
indeed, Theorem~\ref{Teo:inter-reducitble-HTNs} allows us to consider both the multi-head
or the (equivalent) multi-tail HyTN without concerns \wrt the consistency checking problem.
\begin{figure}[tb]
\centering
\begin{tikzpicture}[arrows=->,scale=1,node distance=2 and 1]
	\node[node,blackNode] (t2e) {$E_{T_2}$};
    \node[node,blackNode,below=of t2e] (t3e) {$E_{T_3}$};
    \node[node,blackNode,below=of t3e] (t4e) {$E_{T_5}$};
	\node[node,blackNode,left=2 of t3e] (a1e) {$E_{\sAND_1}$};
    \node[node,blackNode,right=2 of t2e] (w2) {$b_1$};
    \node[node,blackNode,right=2 of t3e] (w3) {$b_2$};
    \node[node,blackNode,right=2 of t4e] (w4) {$b_3$};
	\node[node,blackNode,right=2 of w3] (a2b) {$B_{\sAND_2}$};
 	\node[node,above right=2.5 and .25 of a1e] (e1u) {$e_1^u$};
	\node[node,below right=0.25 of e1u, yshift=-2ex] (e1l) {$e_1^l$};
 	\node[node,right=.7 of a1e,yshift=6ex] (e2u) {$e_2^u$};
	\node[node,right=.7 of a1e,yshift=-6ex] (e2l) {$e_2^l$};
 	\node[node,below right=2.5 and .25 of a1e] (e3u) {$e_3^u$};
	\node[node,above right=0.25 of e3u, yshift=2ex] (e3l) {$e_3^l$};
 	\node[node,right=.7 of t2e,yshift=4ex] (e4u) {$e_4^u$};
	\node[node,right=.7 of t2e,yshift=-6ex] (e4l) {$e_4^l$};
 	\node[node,right=.7 of t3e,yshift=5ex] (e5u) {$e_5^u$};
	\node[node,right=.7 of t3e,yshift=-5ex] (e5l) {$e_5^l$};
 	\node[node,right=.7 of t4e,yshift=6ex] (e6u) {$e_6^u$};
	\node[node,right=.7 of t4e,yshift=-4ex] (e6l) {$e_6^l$};
	\node[node,right=.7 of w2,yshift=3ex] (e7) {$e_7$};
	\node[node,right=.7 of w3,yshift=-4ex] (e8) {$e_8$};
	\node[node,right=.7 of w4,yshift=-3ex] (e9) {$e_9$};
	\node[node,right=.7 of w3,yshift=6ex] (h1) {$h_1$};

	\draw[] (a1e) -- node[timeLabel,above,sloped] {$0$} (e1l);%
	\draw[] (a1e) -- node[timeLabel,above,sloped] {$0$} (e2l);%
	\draw[] (a1e) -- node[timeLabel,above,sloped] {$0$} (e3l);%
   	\draw[thick] (e1u) -- node[timeLabel,above,sloped] {$74$} (a1e);%
	\draw[thick] (e2u) -- node[timeLabel,above,sloped] {$12$} (a1e);%
	\draw[thick] (e3u) -- node[timeLabel,above,sloped] {$69$} (a1e);%
	\draw[] (t2e) -- node[timeLabel,above,sloped] {$0$} (e1u);%
	\draw[] (t2e) -- node[timeLabel,above,sloped] {$0$} (e4l);%
	\draw[thick] (e1l) -- node[timeLabel,above,sloped] {$-17$} (t2e);%
    	\draw[thick] (e4u) -- node[timeLabel,above,sloped] {$6$} (t2e);%
	\draw[] (t3e) -- node[timeLabel,above,sloped] {$0$} (e2u);%
	\draw[] (t3e) -- node[timeLabel,above,sloped] {$0$} (e5l);%
	\draw[thick] (e2l) -- node[timeLabel,above,sloped] {$-6$} (t3e);%
    	\draw[thick] (e5u) -- node[timeLabel,above,sloped] {$5$} (t3e);%
	\draw[] (t4e) -- node[timeLabel,above,sloped] {$0$} (e3u);%
	\draw[] (t4e) -- node[timeLabel,above,sloped] {$0$} (e6l);%
	\draw[thick] (e3l) -- node[timeLabel,above,sloped] {$-20$} (t4e);%
    	\draw[thick] (e6u) -- node[timeLabel,above,sloped] {$10$} (t4e);%

	\draw[] (w2) -- node[timeLabel,above,sloped] {$0$} (e4u);%
    	\draw[] (w2) -- node[timeLabel,above,sloped] {$0$} (e7);%
	\draw[thick] (h1) -- node[timeLabel,above,sloped] {$0$} (w2);%
	\draw[thick] (e4l) -- node[timeLabel,above,sloped] {$-2$} (w2);%
	\draw[] (w3) -- node[timeLabel,above,sloped] {$0$} (e5u);%
   	 \draw[] (w3) -- node[timeLabel,above,sloped, xshift=-1ex] {$0$} (e8);%
	\draw[dashed] (h1) -- node[timeLabel,above,sloped] {$0$} (w3);%
	\draw[thick] (e5l) -- node[timeLabel,above,sloped] {$-1$} (w3);%
	\draw[] (w4) -- node[timeLabel,above,sloped] {$0$} (e6u);%
   	 \draw[] (w4) -- node[timeLabel,above,sloped] {$0$} (e9);%
	\draw[dashed] (h1) to [bend right=15] node[timeLabel,above,sloped] {$0$} (w4);%
	\draw[thick] (e6l) -- node[timeLabel,above,sloped] {$-5$} (w4);%
	\draw[] (a2b) -- node[timeLabel,above,sloped] {$0$} (h1);%
	\draw[thick] (e7) -- node[timeLabel,above,sloped] {$0$} (a2b);%
	\draw[thick] (e8) -- node[timeLabel,above,sloped] {$0$} (a2b);%
	\draw[thick] (e9) -- node[timeLabel,above,sloped] {$0$} (a2b);%
\end{tikzpicture}
\caption{The MPG equivalent to the \TN depicted in \figref{FIG:wf-hstn},
  obtained by considering the (equivalent) multi-head \TN transformation of the multi-tail \TN shown in \figref{FIG:wf-hstn}.
A winning positional strategy $\pi_1$ for Player~$1$ is highlighted by thick arcs.
The dashed arcs are those not prescribed by strategy $\pi_1$, \ie they are removed when projecting the MPG on $\pi_1$.}
\label{FIG:htn-mpg}
\end{figure}
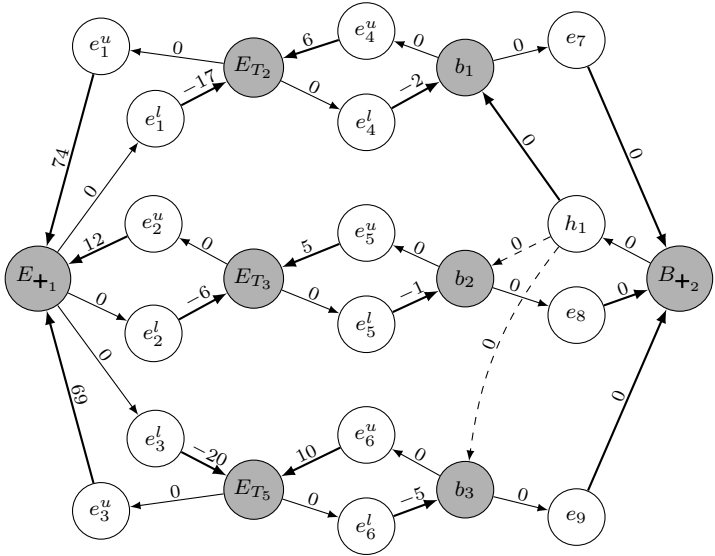

Now, let us introduce the formal proof of Theorem~\ref{Teo:TN2MPG} by the following two lemmas.
\begin{lemma}\label{lem:mainReductionHcons}
If $\H$ is consistent then every node of $G_{\H}$ is a winning start position for Player~$1$.
\end{lemma}
\begin{proof} 
Since $\H$ is consistent, there exists a feasible scheduling \mbox{$s:V\rightarrow \R$} such that, for each hyperarc 
$A\in \A$, the reduced slack weight is non-negative $w_A^{s}\geq 0$.
Consider the following positional strategy $\pi_1$ for Player~$1$: 
for each $A\in V_1$, 
\[
\pi_1(A)=\arg\min_{h\in H_A} \{ s(h) - w_A(h) \}.
\]
We claim that $\pi_1$ ensures Player~$1$ the win, wherever node the game starts from and however Player~$0$ moves.
In order to show this, we prove that the projection graph $G_{\pi_1}$ is conservative exhibiting a feasible potential $p$.
Let \mbox{$p:V_0\cup V_1\rightarrow \R$} be defined as follows: 
\begin{equation}\label{eq:def:potential}
 p(v) \triangleq
\begin{cases}
        s(v)   & \text{if $v\in V_0$,} \\
      s(t(v))   & \text{if $v\in V_1$.} \\
\end{cases}
\end{equation}
Now, let $a=(u,v,w)$ be any arc of $G_{\pi_1}$:

\noindent{\emph Case~1:} 
if $v\in V_1$, then $v$ is a hyperarc of $\H$ with $t(v)=u$ and $w=0$; 
therefore, $p(v) = s(t(v)) = s(u) = p(u)$ since $u\in V_0$.
Then $w^p(u,v) = w - p(v) + p(u) = 0 \geq 0$ follows;

\noindent{\emph Case~2:}
if $v\in V_0$, then $u\in V_1$ and $w=w_u(v)$.
Moreover, $v=\pi_1(u)$,
which implies that $v = \arg\min_{h\in H_u} \{s(h) - w_u(h)\}$.
Therefore, recalling that $w_u^s\geq 0$, \ie $s(t(u)) \geq \min_{h\in H_u} \{s(h) - w_u(h)\}$:
\[
   p(u) = s(t(u)) \geq \min_{h\in H_u} \{ s(h) - w_u(h) \} = s(v)-w_u(v) = p(v)-w. \\
\]
Hence, $w^p(u,v) = w - p(v) + p(u) \geq 0$.

\noindent
In conclusion, $G_{\pi_1}$ is conservative.
Therefore, the positional strategy $\pi_1$ certifies that any node of $G$ is a winning start position for Player~$1$. 
\end{proof}

\begin{lemma}\label{lem:mainReductionGwin}
  If every node of $G_{\H}$ is a winning start position for Player~$1$
  then $\H$ is a consistent \TN.
\end{lemma}
\begin{proof} 
If every node is a winning start position for Player~$1$, then there exists a positional strategy $\pi_1$ which is everywhere winning for Player~$1$.
Notice that $G_{\pi_1}$ must be conservative
since Player~$0$ can clearly win any play starting from a node located on a negative cycle.
Let $p:V_0\cup V_1\rightarrow \R$ be a feasible potential for $G_{\pi_1}$.
We claim that the restriction of $p$ onto $V_0$ is a feasible scheduling for $\H$. 
Indeed, for any hyperarc $A$ of $\H$,
$(t_A,A,0)$ is an arc of $G_{\pi_1}$, 
whence $p(A)\leq p(t_A)$.
Moreover,
$(A,\pi_1(A), w_A(\pi_1(A)))$ is also an arc of $G_{\pi_1}$,
whence $p(\pi_1(A))\leq p(A) + w_A(\pi_1(A))$. 
Since $\pi_1(A)\in H_A$, then the following holds:
\begin{align*}
  p(t_A) \geq p(A) &\geq p(\pi_1(A)) - w_A(\pi_1(A)) \\
                  &\geq \min_{h\in H_A} \{ s(h) - w_A(h) \}.
\end{align*}
Hence, the restriction of $p$ onto $V_0$ is a feasible scheduling for $\H$. Thus, $\H$ is consistent.
\end{proof}

In \figref{FIG:mpg-with-t} the values under the nodes represent a feasible potential for the projection of the MPG depicted in \figref{FIG:htn-mpg}.
By Lemma~\ref{lem:mainReductionGwin}, the restriction of such a feasible potential on the black nodes is also a feasible scheduling for the corresponding \TN depicted in \figref{FIG:wf-hstn}.
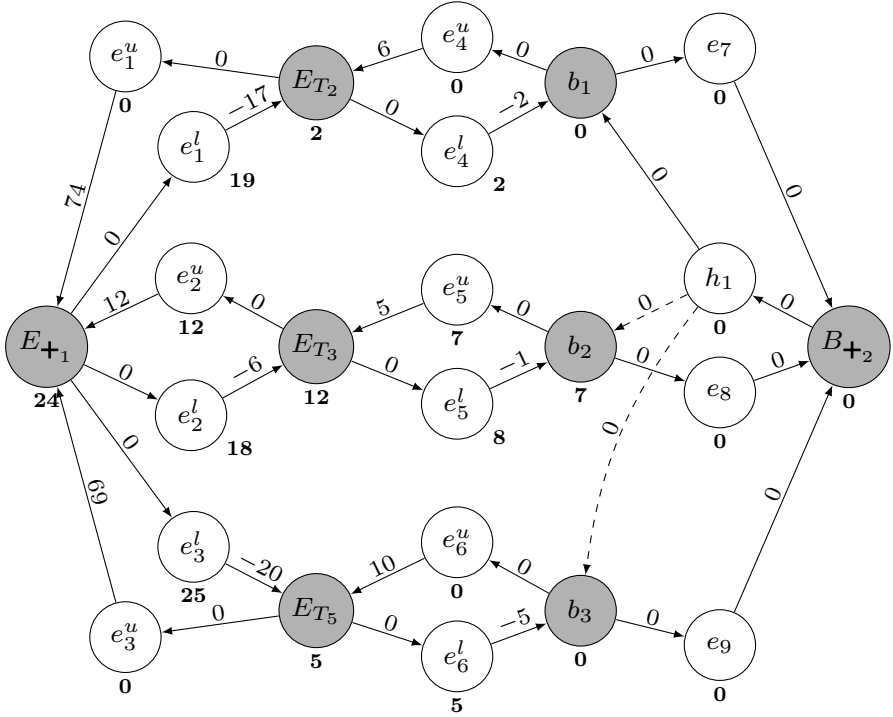
\begin{figure}[tb]
\centering
\begin{tikzpicture}[arrows=->,scale=1.25,node distance=2 and 1]
	\node[node,blackNode,label=-90:$\textbf{2}$] (t2e) {$E_{T_2}$};
    \node[node,blackNode,below=of t2e,label=-90:$\textbf{12}$] (t3e) {$E_{T_3}$};
    \node[node,blackNode,below=of t3e,label=-90:$\textbf{5}$] (t4e) {$E_{T_5}$};
	\node[node,blackNode,left=2 of t3e,label=-90:$\textbf{24}$] (a1e) {$E_{\sAND_1}$};
    \node[node,blackNode,right=2 of t2e,label=-90:$\textbf{0}$] (w2) {$b_1$};
    \node[node,blackNode,right=2 of t3e,label=-90:$\textbf{7}$] (w3) {$b_2$};
    \node[node,blackNode,right=2 of t4e,label=-90:$\textbf{0}$] (w4) {$b_3$};
	\node[node,blackNode,right=2 of w3,label=-90:$\textbf{0}$] (a2b) {$B_{\sAND_2}$};
 	\node[node,above right=2.5 and .25 of a1e, label=-90:$\textbf{0}$] (e1u) {$e_1^u$};
	\node[node,below right=0.25 of e1u, yshift=-2ex, label=-35:$\textbf{19}$] (e1l) {$e_1^l$};
 	\node[node,right=.7 of a1e,yshift=6ex, label=-90:$\textbf{12}$] (e2u) {$e_2^u$};
	\node[node,right=.7 of a1e,yshift=-6ex, label=-35:$\textbf{18}$] (e2l) {$e_2^l$};
 	\node[node,below right=2.5 and .25 of a1e, label=-90:$\textbf{0}$] (e3u) {$e_3^u$};
	\node[node,above right=0.25 of e3u, yshift=2ex, label=-90:$\textbf{25}$] (e3l) {$e_3^l$};
 	\node[node,right=.7 of t2e,yshift=4ex, label=-90:$\textbf{0}$] (e4u) {$e_4^u$};
	\node[node,right=.7 of t2e,yshift=-6ex, label=-35:$\textbf{2}$] (e4l) {$e_4^l$};
 	\node[node,right=.7 of t3e,yshift=5ex, label=-90:$\textbf{7}$] (e5u) {$e_5^u$};
	\node[node,right=.7 of t3e,yshift=-5ex, label=-35:$\textbf{8}$] (e5l) {$e_5^l$};
 	\node[node,right=.7 of t4e,yshift=6ex, label=-90:$\textbf{0}$] (e6u) {$e_6^u$};
	\node[node,right=.7 of t4e,yshift=-4ex, label=-90:$\textbf{5}$] (e6l) {$e_6^l$};
	\node[node,right=.7 of w2,yshift=3ex, label=-90:$\textbf{0}$] (e7) {$e_7$};
	\node[node,right=.7 of w3,yshift=-4ex, label=-90:$\textbf{0}$] (e8) {$e_8$};
	\node[node,right=.7 of w4,yshift=-3ex, label=-90:$\textbf{0}$] (e9) {$e_9$};
	\node[node,right=.7 of w3,yshift=6ex, label=-90:$\textbf{0}$] (h1) {$h_1$};

	\draw[] (a1e) -- node[timeLabel,above,sloped] {$0$} (e1l);%
	\draw[] (a1e) -- node[timeLabel,above,sloped] {$0$} (e2l);%
	\draw[] (a1e) -- node[timeLabel,above,sloped] {$0$} (e3l);%
   	\draw[] (e1u) -- node[timeLabel,above,sloped] {$74$} (a1e);%
	\draw[] (e2u) -- node[timeLabel,above,sloped] {$12$} (a1e);%
	\draw[] (e3u) -- node[timeLabel,above,sloped] {$69$} (a1e);%
	\draw[] (t2e) -- node[timeLabel,above,sloped] {$0$} (e1u);%
	\draw[] (t2e) -- node[timeLabel,above,sloped] {$0$} (e4l);%
	\draw[] (e1l) -- node[timeLabel,above,sloped] {$-17$} (t2e);%
    \draw[] (e4u) -- node[timeLabel,above,sloped] {$6$} (t2e);%
	\draw[] (t3e) -- node[timeLabel,above,sloped] {$0$} (e2u);%
	\draw[] (t3e) -- node[timeLabel,above,sloped] {$0$} (e5l);%
	\draw[] (e2l) -- node[timeLabel,above,sloped] {$-6$} (t3e);%
    \draw[] (e5u) -- node[timeLabel,above,sloped] {$5$} (t3e);%
	\draw[] (t4e) -- node[timeLabel,above,sloped] {$0$} (e3u);%
	\draw[] (t4e) -- node[timeLabel,above,sloped] {$0$} (e6l);%
	\draw[] (e3l) -- node[timeLabel,above,sloped] {$-20$} (t4e);%
    \draw[] (e6u) -- node[timeLabel,above,sloped] {$10$} (t4e);%

	\draw[] (w2) -- node[timeLabel,above,sloped] {$0$} (e4u);%
    \draw[] (w2) -- node[timeLabel,above,sloped] {$0$} (e7);%
	\draw[] (h1) -- node[timeLabel,above,sloped] {$0$} (w2);%
	\draw[] (e4l) -- node[timeLabel,above,sloped] {$-2$} (w2);%
	\draw[] (w3) -- node[timeLabel,above,sloped] {$0$} (e5u);%
    \draw[] (w3) -- node[timeLabel,above,sloped, xshift=-1ex] {$0$} (e8);%
	\draw[dashed] (h1) -- node[timeLabel,above,sloped] {$0$} (w3);%
	\draw[] (e5l) -- node[timeLabel,above,sloped] {$-1$} (w3);%
	\draw[] (w4) -- node[timeLabel,above,sloped] {$0$} (e6u);%
    \draw[] (w4) -- node[timeLabel,above,sloped] {$0$} (e9);%
	\draw[dashed] (h1) to [bend right=15] node[timeLabel,above,sloped] {$0$} (w4);%
	\draw[] (e6l) -- node[timeLabel,above,sloped] {$-5$} (w4);%
	\draw[] (a2b) -- node[timeLabel,above,sloped] {$0$} (h1);%
	\draw[] (e7) -- node[timeLabel,above,sloped] {$0$} (a2b);%
	\draw[] (e8) -- node[timeLabel,above,sloped] {$0$} (a2b);%
	\draw[] (e9) -- node[timeLabel,above,sloped] {$0$} (a2b);%
\end{tikzpicture}
\caption{The integer labels under the nodes are a feasible potential
for the {projection} on $\pi_1$ of the MPG depicted in \figref{FIG:htn-mpg}.
The restriction of this potential on the black nodes (those in $V_0$)
is a feasible scheduling for the \TN depicted in \figref{FIG:wf-hstn}
as explained in the proof of Lemma~\ref{lem:mainReductionGwin}.}%
\label{FIG:mpg-with-t}
\end{figure}
Now, 
we have all the necessary results to prove the following theorem.

\begin{theorem}\label{Teo:MainAlgorithms}
	Let $\H=(V, \A)$ be an integral-weighted \TN, $m= \sum_{A\in\A}|A|$, 
	and $W=\max_{A\in\A}\{ \max_{h\in A} |w_A(h)|\}$ the maximal weight value present in $\H$.  
	The following propositions hold: 
	\begin{enumerate}
		\item There exists an $O((|V|+|\A|) m W)$ pseudo-polynomial time algorithm deciding \TNC for $\H$;
		\item\label{Cor:PseudoPolyScheduling} There exists an $O((|V|+|\A|) m W)$ pseudo-polynomial 
		time algorithm such that, given on input any consistent \TN $\H$, 
		it returns as output a feasible scheduling $s:V_{\H}\rightarrow \Z$ of $\H$;
		\item There exists an $O((|V|+|\A|) m W)$ pseudo-polynomial time algorithm such that, given on input any not-consistent \TN $\H$, it returns as output a
		negative cycle $(S,\C)$ of $\H$.
	\end{enumerate}
\end{theorem}
\begin{proof}
\begin{enumerate}
\item The decision algorithm is sketched in \figref{FIG:PseudocodeDecisionTNC}.
It takes in input a \TN $\H=(V,\A)$ and, in line~1, constructs the corresponding \MPG $G_{\H}$ as described in Theorem~\ref{Teo:TN2MPG}. 
This first step takes $O(m)$ time and yields a graph with $|V|+|\A|$ nodes and $O(m)$ arcs. 
Then, in line~2, the instance of \MPGT with $T=0$ on graph $G_{\H}$ is solved in $O((|V|+|\A|) m W)$ time by the Value Iteration Algorithm (see Theorem~\ref{Teo:FasterAlgo}).
The output consists in a partition of $G_{\H}$ nodes into two sets: $W_1 = \{v\in V\cup \A \mid \nu_v\geq 0\}$ and $W_0 = \{v\in V\cup \A \mid \nu_v< 0\}$.
If $W_0$ is empty, then $\H$ is consistent by Lemma~\ref{lem:mainReductionGwin}, otherwise it is not consistent by Lemma~\ref{lem:mainReductionHcons}.

\begin{figure}[tb]
\removelatexerror
\begin{algorithm}[H]
\caption{\texttt{isConsistent}$(\H)$}\label{ALG:PseudocodeDecisionTNC}
	\tcp{a \TN $\H = (V,\A)$ of unknown consistency state}
	$G_{\H} \leftarrow \texttt{makeACorrespondingGame}(\H)$; \tcp{See Algorithm~\ref{ALGO:PseudocodeReduction-htn-mpg}}
	$(W_0, W_1)\leftarrow \texttt{solve{\MPGT}}(G_{\H},0)$; \tcp{Brim's algorithm, see Theorem~\ref{Teo:FasterAlgo}}
	\lIf{$(W_0=\emptyset)$}{\textbf{Output}: \texttt{YES}}
	\lElse{\textbf{Output}: \texttt{NO}}
\end{algorithm}
\caption{Pseudocode of the algorithm for deciding \TNC.}
\label{FIG:PseudocodeDecisionTNC}
\end{figure}

\item In case $W_0$ is empty, a feasible scheduling is obtained as shown in Algorithm~\ref{ALG:PseudocodeFeasibleSchedules}. 
First, in line~2, the algorithm computes a positional winning strategy $\pi_1$ for Player~$1$.
This takes $O((|V| + |\A|) m W)$ time by Theorem~\ref{Teo:FasterAlgo}.  
Next, in line~3, it builds the graph $G_{\pi_1}$ which is conservative
since $\pi_1$ is a positional winning strategy for Player~$1$.
Then, in lines~4-5, it adds a new node $s$ to $V_1$ and a new arc $e_v=(s,v,0)$ for each node $v\in V_0$ in $G_{\pi_1}$. 
Let $G'_{\pi_1}=(V_0\cupdot (V_1\cup \{s\}), E')$ the graph thus obtained.
Observe that every node of $G'_{\pi_1}$ is reachable from $s$.
Indeed, every node $A\in V_1 =\A$ can be reached by traversing two arcs:
from $s$ to $t_A$ along the arc $e_{t_A}=(s,t_A,0)$, which belongs to $G'_{\pi_1}$ as $t_A\in V_0$,
then from $t_A$ to $A$ along the arc $(t_A,A,0)$, which belongs to $G_{\pi_1}$ (and hence to $G'_{\pi_1}$) since $t_A\in V_0$.\\
Since the added node $s$ is a source, then $G'_{\pi_1}$ is conservative too.
Therefore, in $G'_{\pi_1}$,
the set of distances from node $s$,
computed calling the Bellman-Ford algorithm in line~6,
forms a feasible potential $p:V_0\cup V_1\cup \{s\}\rightarrow \Z$
and the restriction of $p$ onto $V_0 = V$ is a feasible scheduling for $\H$.

\begin{figure}[tb]
\begin{algorithm}[H]
	\caption{\texttt{computeAFeasibleSchedule}$(\H)$}\label{ALG:PseudocodeFeasibleSchedules}
	\tcp{a consistent \TN $\H = (V,\A)$}%
	$G\leftarrow \texttt{makeACorrespondingGame}(\H)$; \tcp{See Algorithm~\ref{ALGO:PseudocodeReduction-htn-mpg}}
	$\pi_1 \leftarrow$ \MPGS($G$); \tcp{Compute a positional winning strategy for Player~$1$; see Theorem~\ref{Teo:FasterAlgo}}
	$G_{\pi_1}\leftarrow$ compute the subgraph of $G$ induced by $\pi_1$\;
	\tcp{Recall $G_{\pi_1}=(V_0 \cupdot V_1,E)$, where $V_0=V$ and $V_1=\A$.}
	$s \leftarrow$ a new node; \tcp{$s\not\in V_0 \cup V_1$}
	Add $s$ to $V_1$ and add an arc $(s, v, 0)$ for each $v\in V_0$\;
	$p\leftarrow \texttt{Bellman-Ford}(G_{\pi_1}, s)$; \tcp{compute a potential function $p$}
	\KwOut{the restriction of $p$ onto $V$}
\end{algorithm}
\caption{Pseudocode of the algorithm for computing a feasible schedule.}%
\label{FIG:PseudocodeFeasibleSchedules}
\end{figure}

\item In case $W_0$ is not empty, a negative cycle is determined by Algorithm~\ref{ALG:PseudocodeNegativeCycles}.
   Let $G[W_0]$ be the subgraph of $G$ induced by $W_0$, \ie the graph obtained from $G$ by removing all nodes not in $W_0$ and all the arcs incident into them.
   Notice that every node $v\in W_0$ is a winning start position for Player~$0$ in game $G[W_0]$ because $v$ is a winning start position for Player~$0$ in game $G$, and no winning strategy for Player~$0$ in $G$ can prescribe a move from a node in $W_0$ to a node in $W_1$; therefore, that same winning strategy remains valid on $G[W_0]$.
   This implies that, for every $u\in W_0$, there exists at least one arc $(u,v)$ with $v\in W_0$.
   In particular, since $(V_0,V_1)$ is a bipartition of $G$,
   then $\overline{W}_0 \triangleq W_0 \cap V_0 \neq \emptyset$.
   In line~3, a positional winning strategy $\pi_0$ for Player~$0$ on $G[W_0]$ is determined. 
   By Theorem~\ref{Teo:FasterAlgo}, this computation takes time $O((|V|+|\A|) m W)$.
   Consider the set of hyperarcs $\C=\{\pi_0(v)\}_{v\in \overline{W}_0}$; the pair $(\overline{W}_0, \C)$ returned by the algorithm is a negative cycle.
   Indeed, for any $v\in \overline{W}_0$, $\pi_0(v)\in V_1$ is a hyperarc of $\H$. Thus the head set $H_{\pi_0(v)}\subseteq V_0$.
   Also, $H_{\pi_0(v)}\subseteq W_0$, 
   since $v$ is a winning start position for Player~$0$
   and $\pi_0$ is a winning strategy for Player~$0$.
   Combining, $H_{\pi_0(v)}\subseteq \overline{W}_0$ determining that $(\overline{W}_0, \C)$ is a negative cycle.
\end{enumerate}
\end{proof}

\begin{figure}[tb]
	\begin{algorithm}[H]
		\caption{\texttt{computeANegativeCycle}$(\H, W_0)$}\label{ALG:PseudocodeNegativeCycles}
		\tcp{a \TN $\H = (V,\A) = (V_0\cupdot V_1, \A)$ which is not consistent}
		\tcp{the non-empty set $W_0 = \{v\in V \mid \nu_v< 0\}$}
		$G\leftarrow \texttt{makeACorrespondingGame}(\H)$; \tcp{See Algorithm~\ref{ALGO:PseudocodeReduction-htn-mpg}}
		$G[W_0]\leftarrow$ compute the subraph of $G$ induced by $W_0$\; 
		$\pi_0\leftarrow$ \MPGS($G[W_0]$); \tcp{Compute a positional winning strategy for Player~$0$; see Theorem~\ref{Teo:FasterAlgo}}
		$\overline{W}_0 \leftarrow W_0 \cap V_0$\;
		$\C\leftarrow \{\pi_0(v)\}_{v\in \overline{W}_0}$\;
		\KwOut{$(\overline{W}_0, \C)$}
	\end{algorithm}
	\caption{Pseudocode of the algorithm for computing a negative cycle.}
	\label{FIG:PseudocodeNegativeCycles}
\end{figure}

\begin{remark}\label{Rem:FeasibleSchedAlgo}
	In Theorem~\ref{Teo:MainAlgorithms}~Item~2), a set of feasible potentials may be obtained without executing the Bellman-Ford algorithm. 
	Actually, if the partition $(W_0,W_1)$ is computed by the Value Iteration Algorithm~\cite{brim2011faster},
	then a feasible scheduling for $\H$ can be directly derived from the \emph{progress measure} computed within the algorithm.
	In more detail, let $G=(V_0\cupdot V_1, E)$ be an \MPG weighted by $w:E\rightarrow\Z$.
	An \emph{energy progress measure} is a function $f:V_0\cup V_1\rightarrow \N\cup\{+\infty\}$ such that:
	if $v\in V_0$, then for every $(v,v', w)\in E$ it holds $f(v)\geq f(v') - w$; 
	otherwise, $v\in V_1$ and there exists $(v,v', w)\in E$ such that $f(v)\geq f(v') - w$.
	An energy progress measure $f:V_0\cup V_1\rightarrow \N\cup\{+\infty\}$ such that $0\leq f(v)<+\infty$ for every $v\in V_0\cup V_1$ is provided by 
	the resolution algorithm of Theorem~\ref{Teo:FasterAlgo} in time $O((|V|+|\A|) m W)$.
	
	The progress measure $f$ is already a feasible scheduling for $\H$: in fact, for every hyperarc $A\in\A$, it holds $(t_A, A, 0)\in E$ and $(A,v, w_A(v))\in E$,
	for every $v\in H_A$; combining these two last facts, it follows that:
	\[ f(t_A)\geq f(A)\geq \min_{v\in H_A} \{ f(v)-w_A(v) \}, \] \ie $f$ is a scheduling satisfying all constrains $A\in\A$.
	This allow us to employ the algorithm depicted in \figref{FIG:PseudocodeFeasibleSchedules-Remark} instead of the one depicted in
	\figref{FIG:PseudocodeFeasibleSchedules} in the case that $W_1 = V$.
\end{remark}

\begin{figure}[tb]
	\begin{algorithm}[H]
		\caption{\texttt{computeAFeasibleSchedule-Remark\ref{Rem:FeasibleSchedAlgo}}$(H)$}
		\tcp{a consistent \TN $\H = (V,\A) = (V_0\cupdot V_1, \A)$}
		\tcp{ref. Remark~\ref{Rem:FeasibleSchedAlgo} and Theorem~\ref{Teo:FasterAlgo}\cite{brim2011faster}}
		$G\leftarrow \texttt{makeACorrespondingGame}(\H)$; \tcp{ref. Algorithm~\ref{ALGO:PseudocodeReduction-htn-mpg}}
		$f\leftarrow \texttt{Value-Iteration}(G)$; \tcp{compute an energy progress measure for $G$ as in Theorem~\ref{Teo:FasterAlgo}} 
		\KwOut{$f$}
	\end{algorithm}
	\caption{Pseudocode of the algorithm of Remark~\ref{Rem:FeasibleSchedAlgo} for computing a feasible schedule.}%
	\label{FIG:PseudocodeFeasibleSchedules-Remark}
\end{figure}

The computational equivalence between \MPGD problem and \TNC can be now determined by showing that also \MPGD can be reduced to \TNC.
\begin{theorem}\label{Teo:MPG2TN}
   There exists a $\log$-space, linear-time, local-replacement reduction from \MPGD to \TNC.
\end{theorem}
\begin{proof}
Let $G=(V_0\cupdot V_1, E)$ be an \MPG. 
For each node $u\in V_0\cup V_1$, let $N_G(u)$ denote the outgoing neighborhood of $u$ in $G$, \ie 
$N_G(u)\triangleq\{v\in V_0\cup V_1\mid (u,v)\in E\}$. 

A corresponding \TN $\H=(V, \A)$, where $V = V_0\cupdot V_1$, is constructed from $G$ as follows.
For every $u\in V_1$, a hyperarc $A_u\in \A$ is added to $\H$, where: 
\[A_u\triangleq (u, N_G(u), w_{A_u}),\] with weight 
$w_{A_u}(v)\triangleq w(u,v)$ for every $v\in N_G(u)$.
Moreover, for every $u\in V_0$ and every $v\in N_G(u)$, a hyperarc $A_{uv}\in\A$ is added to $\H$, 
where: \[A_{uv}\triangleq (u,v,w(u,v)).\]
This construction requires a log-space and linear-time computation.

Now, we firstly prove that if $\H$ is consistent then every node of $G$ is a winning start position for Player~$1$.

Indeed, let $s:V\rightarrow \R$ be a feasible scheduling for $\H$.
Thus, $w_A^{s}\geq 0$ for every hyperarc $A\in \A$.
Notice that, by construction, for each $u\in V_1$ there exists a unique hyperarc $A_u\in \A$ with tail $t_{A_u}=u$; 
moreover, it holds that $H_{A_u} \triangleq N_G(u)$.
Hence, for each $u\in V_1$, we can define a positional strategy $\pi_1$ for Player~$1$ as follows: 
\[\pi_1(u) \triangleq \arg\min_{h\in H_{A_u}} \{ s(h)-w_{A_u}(h) \} .\] 
Now, consider the potential function $p:V\rightarrow\R$ defined as: $p(u)\triangleq s(u)$ for every $u\in V$. 
We argue that $p$ is a feasible potential for $G_{\pi_1}$. 

In fact, let $a=(u,v,w)\in E$ be any arc of $G$:
\begin{description}[Case~2:]
\item[Case~1:] Assume that $u\in V_0$. Then, by construction, $a=A_{uv}$. 
Hence, from $p\triangleq s$ and the feasibility of $s$, we have: 
\[
	\begin{array}{lll}
	p(u)=s(u)&\geq& \min_{h\in H_{A_{uv}}} \{ s(h)-w_{A_{uv}}(h) \} \\ 
		 &=& s(v)-w_{A_{uv}}(v) \\
	         &=& p(v)-w
	\end{array}
\]
Hence, $w^p(u,v)=w-p(v)+p(u)\geq 0$;
\item[Case~2:] Assume that $u\in V_1$. Then, by construction, $A=(u, N_G(u), w_A)$, where $w_A(z)=w(u,z)$ for every $z\in N_G(u)$;
moreover, notice that $v=\pi_1(u)\in N_G(u)=H_A$. 
Hence, from $p\triangleq s$, the feasibility of $s$, and the definition of $\pi_1$, we have: 
\[
	\begin{array}{lll}
	p(u)=s(u)&\geq& \min_{h\in H_{u}} \{ s(h)-w_{A_u}(h) \} \\
	         &=& s(\pi_1(u)) - w_{A_u}(\pi_1(u)) \\
	         &=& s(v) - w_{A_u}(v) \\
	         &=& p(v) - w \\
	\end{array}
\]
Hence, $w^p(u,v)=w-p(v)-p(u)\geq 0$.
\end{description}
Thus, $G_{\pi_1}$ is conservative. This implies that every node of $G$ is a winning start for Player~$1$.  

Secondly, we prove that if every node of $G$ is a winning start position for Player~$1$, then $H$ is consistent.  

Let $\pi_1$ be a positional winning strategy for Player~$1$.
It follows that $G_{\pi_1}$ is conservative and, therefore, it admits a feasible potential $p:V\rightarrow \R$.
Now, consider the scheduling function $s:V\rightarrow\R$ for $\H$ 
defined as: $s(u)\triangleq p(u)$ for every $u\in V$. 
We argue that $s$ is a feasible scheduling of $\H$.

In fact, let $A=(t_A, H_A, w_A)\in\A$ be any hyperarc of $\H$:
\begin{description}[Case~2:]
\item[Case 1:] assume $t_A\in V_0$. 
Then, by construction, $A=(u,v,w)$ for some $v\in N_G(u), w\in\R$ and $u=t_A$. 
Hence, from $s\triangleq p$ and the feasibility of $p$, we have:
\[
\begin{array}{lll}
s(t_A) = p(u) & \geq& p(v) - w \\
              &=& s(v) - w_A(v) \\
              &=& \min_{h\in H_A} \{ s(h) - w_A(h) \} \\
\end{array}
\]
Hence, $s$ satisfies $A$, \ie $w_A^s\geq 0$ ;
\item[Case 2:] assume $t_A\in V_1$. 
Then, by construction, $A=(u, N_G(u), w_A)$ for $u=t_A$ and $w_A(v)=w(u,v)\in\R$ for every $v\in N_G(u)$;
moreover, if $v\triangleq \pi_1(u)$, then $v\in N_G(u)=H_A$. 
Hence, from $s\triangleq p$ and the feasibility of $p$, we have:
\[
\begin{array}{lll}
s(t_A) = p(u) &\geq& p(v) - w \\
	      &=& s(v) - w_A(v) \\
              &\geq& \min_{h\in H_A} \{ s(h) - w_A(h) \} \\
\end{array}
\]
Hence, $s$ satisfies $A$, \ie $w_A^s\geq 0$.
\end{description}
This proves that $s$ satisfies every hyperarc $A\in\A$. 
Then $s$ is a feasible scheduling of $\H$, which is thus consistent.
\end{proof}

%% file: sectExperiments.tex
This section describes our empirical evaluation of the proposed consistency checking algorithms 
to evaluate the performances and the general applicability of the proposed \TN model.
Both Algorithm~\ref{ALG:PseudocodeFeasibleSchedules} and Algorithm~\ref{ALG:PseudocodeNegativeCycles} consist of one single call to 
Algorithm~\ref{ALG:PseudocodeDecisionTNC}, plus some extra computation of lower asymptotic complexity.
Since the cost of these further computations was confirmed to be practically negligible in some preliminary experiments, 
we report on the results of our experimental investigations only for 
Algorithm~\ref{ALG:PseudocodeDecisionTNC}. 

All algorithms and procedures employed in this empirical evaluation have been implemented in C/C++ 
and executed on a Linux machine having the following characteristics:
\begin{itemize}
\item 2 CPU AMD Opteron 4334;
\item 64GB RAM; 
\item Ubuntu server 14.04.1 Operating System.
\end{itemize}
The source code and all \TN{s} used in the experiments are freely available~\cite{Comin15}.

The main goal of this empirical evaluation was to determine the average computation time 
of Algorithm~\ref{ALG:PseudocodeDecisionTNC}, with respect to
randomly-generated \TN{s} following different criteria, in order to give an idea of the practical behavior of the algorithm.
According to Theorem~\ref{Teo:MainAlgorithms}, 
the worst-case time complexity of Algorithm~\ref{ALG:PseudocodeDecisionTNC} is $O(( n + m' ) m W)$,
where $n=|V|$, $m'=|\A|$, $m= \sum_{A\in\A}|A|$, and $W=\max_{A\in\A}\{ \max_{h\in A} |w_A(h)|\}$.
Hence, we implemented different experiments with respect to the parameters $n, m', m,$ and $W$. 
Here we propose a summary of the obtained results presenting a brief report about four tests, 
Test~1, Test~2, Test~3 and Test~4.

In Test~1 the average computation time was determined for different \TN{} 
orders $n$ to emphasize the practical computation time dependency on $n$.
In Test~2 the average computation time was determined for different \TN{} 
maximal edge-weights $W$ to understand how much the practical computation time is dependent on  $W$.
In Test~3 we investigated how some execution times affect the value of the standard deviation, 
with the goal to determine how many instances require a significant greater computation time 
with respect to the average time of a data set.
Finally, in Test~4 the average computation time was determined with respect to different values of the number of possible strategies of Player~1
$\prod_{A\in\A} |H_A|$ in order to give an idea about the possible practical relation between execution time and number of possible strategies.

The generation of random \TN instances was carried out exploiting two generators.
The first generator was the random workflow schema generator provided by ATAPIS toolset~\cite{Lanz14b}: 
it produces random workflow graphs according to different 
input parameters that allow to control the minimal and maximal number of activities, probability for having
parallel branches, the minimal and maximal probability of inter-task temporal constraints, etc. on the generated graphs.
We verified that this tool allows the determination of graphs that are not only a closer approximation to real-world examples,
but also more difficult to check than those generated at random without particular criteria.

We generated benchmarks as follows:
\begin{enumerate} 
	\item First, temporal workflow graphs were generated by fixing the probability for parallel branches to 10\% and
	maximal value for each activity duration or delay between activities to a value $W$, where $W$ was chosen accordingly to the test type;
	\item Then, each workflow graph was translated into an equivalent \TN $\H$ by the simple transformation algorithm exemplified in Section~\ref{sect:motivating}.
\end{enumerate}
It is worth noting that different random workflow graphs all having the same number of activities may translate into \TN{s}
having different orders $n$ because the original workflow graphs may have different number of connector nodes.
Considering the transformation algorithm exemplified in Section~\ref{sect:motivating}, it is easy to verify that a workflow with $N$ activities can translate
into a CSTN having between $2N+2$ nodes (when the workflow is a simple sequence)
and $5N+2$ nodes (when the workflow is a sequence of groups of two parallel activities).

ATAPIS toolset has been designed to generate graphs 
with strongly connected components (Andreas Lanz, personal communication, October 6, 2015). 
In particular, it has been optimized for small graphs with up to $50$ activities.
This design choice was motivated by the widely accepted seven 
process modeling guidelines~\cite{Mendling10} which suggests to always ``decompose a (workflow) model
with more than $50$ elements (activities)''.
Therefore, we used the tool for generating random workflow 
graphs with $100$ activities at maximum and, consequently, obtaining \TN{s} having $502$ nodes at most.

In Table~\ref{Table:WorkflowTNsize} we report the orders of the smallest and 
the largest \TN{} determined from each set of random generated
workflow graphs having $N$ activities for $N\in\{10,20,\ldots,100\}$.
\begin{table}[tb]
	\centering
	\caption{Orders of the smallest and biggest \TN{} determined 
		for each set of random generated workflows having $N$ activities.}%
		\label{Table:WorkflowTNsize}
	    \begin{tabular}[b]{| r | r | r |}
		\hline
		 \multicolumn{1}{|c|}{$N$}  & \multicolumn{1}{c|}{Order of smallest \TN{}} & \multicolumn{1}{c|}{Order of biggest \TN{}} \\
		\hline
	      10 & 26   & 50 \\
	      20 & 48   & 94 \\
	      30 & 78   & 138 \\  
	      40 & 104  & 196 \\
	      50 & 136  & 236 \\
	      60 & 164  & 268 \\
	      70 & 196  & 306 \\
	      80 & 222  & 350 \\
	      90 & 262  & 394 \\
	     100 & 292  & 410 \\
		\hline
		\end{tabular}
\end{table}

In order to study the scalability of the algorithm with respect to the number of nodes, we had to rely on a second generator of random \TN{} graphs.
Our choice has been to use the \texttt{randomgame} procedure of \texttt{pgsolver} suite~\cite{pgsolver}, 
that can produce parity games instances for any given number of nodes. 
In particular, we exploited \texttt{randomgame} in the following way:
\begin{enumerate} 
	\item First, \texttt{randomgame} was used to generate random directed graphs with out-degree fixed to 3;
	\item Then, the resulting graphs were translated into \MPG{s} by weighting each arc with an integer randomly chosen in the interval $[-W,W]$, where $W$ was
	chosen accordingly to the test type;
	\item Finally, each \MPG $G$ was translated into a \TN $\H_G$ by the reduction algorithm of Theorem~\ref{Teo:MPG2TN}.  
	During the translation from \MPG to \TN, only 10\% of the hyperarcs were maintained having multiple heads, 
	while 90\% of hyperarcs were transformed into standard arcs.
	This 10\%-rule stems from the fact that we are considering 
	workflow based applications where the percentage of (multi-headed) hyperarcs is less than 10\%
	compared to standard arcs in general.
\end{enumerate}
With such settings, the resulting \TN{s} are characterized by $m, m'\in\Theta(n)$.
\smallskip


\begin{table}[t]
	\caption{Comparison between different kinds of queue implementation in the Value-Iteration procedure. All values are in seconds.}\label{Table:Value-Iteration-Queues}
	\centering
	\begin{tabular}{| c | r | r | r | r |}
		\cline{2-5}
		\multicolumn{1}{c|}{} & FIFO Queue & LIFO Queue &\vtop{\hbox{\strut LIFO Queue}\hbox{\strut + Stopping-Criterion}} & Max-Priority Queue\\
		\hline
		$\mu$  		&  90.55    &    11.77    &    6.98    &    184.69   \\ 
		$\sigma$	&  487.69   &    64.10    &   34.61    &    653.26  \\
		\hline
	\end{tabular}
\end{table}

Before presenting the summary of results, it is worthwhile to present some implementation choices about Algorithm~\ref{ALG:PseudocodeDecisionTNC} that we had
to adopt.
The core of the algorithm consists of calls to algorithms \texttt{makeACorrespondingGame(\H)}, that transforms the given \TN \H  into a MPG $G_{\H}$, and
\texttt{solveMPG-Threshold($G_{\H},0$)} (Value Iteration algorithm), that determines  for which game nodes $s$ it holds that $v_s \geq 0$.
The \texttt{makeACorrespondingGame()} implementation didn't require significant choices thanks to the simple structure of the algorithm. 
On the contrary, in the implementation of \texttt{solveMPG-Threshold()} we introduced some further ideas in order to speed-up the  algorithm and avoid
unnecessary computations.
In particular, it is not necessary for \texttt{solveMPG-Threshold()} to continue to determine other potential value $v_{s'}$ as soon as it determines a value
$v_s < 0$: at this point we can already conclude that the network is not consistent and, with a lower computational cost, we can yield 
a generalized negative circuit assessing this fact (Lemmas~\ref{lem:mainReductionHcons} and \ref{lem:mainReductionGwin}).

Moreover, we verified that there is an important data structure in the original Value Iteration algorithm, a queue, that is not further specified by the
authors and that different implementations of it affect the performance of the algorithm.
Therefore, we decided to verify whether \texttt{solveMPG-Threshold()} performance could be appreciably improved adding a suitable stopping criterion and a
proper queue implementation. 
Table~\ref{Table:Value-Iteration-Queues} reports the obtained results, mean execution time $\mu$ and its standard deviations $\sigma$, determined running the
following different versions of \texttt{solveMPG-Threshold()} on the same data set of $10^3$ not consistent \TN{s}\footnote{We
considered not consistent \TN{s} because they practically required more time to be solved.} having $|V|=10^6$ and $W=10^3$:
\begin{enumerate}
	\item \textbf{FIFO Queue}: the original queue is implemented as a FIFO queue; 
	\item \textbf{LIFO Queue:} the original queue is implemented as a a LIFO queue (stack);
	\item \textbf{LIFO Queue+Stopping Criterion:} the queue is implemented as stack and the computation is halted either when all potential values are
	stable or when any of them is negative;
	\item \textbf{Max-Priority Queue:} the original queue is implemented as a Fibonacci's heap.
\end{enumerate}

The results show that, in general, \texttt{solveMPG-Threshold()} performance becomes better if the original queue is implemented as a stack and, in particular,
a further improvement can be obtained if the stopping criterion is also considered. 
Nevertheless, such improvements can only partially reduce the statistics
variability of the running time, as it is shown in the following experimental results.
\smallskip

As mentioned above, the goal of Test~1 was to determine the average computation time of Algorithm~\ref{ALG:PseudocodeDecisionTNC} implementation for different
values of $n$ to study the practical computation time dependency on such parameter.

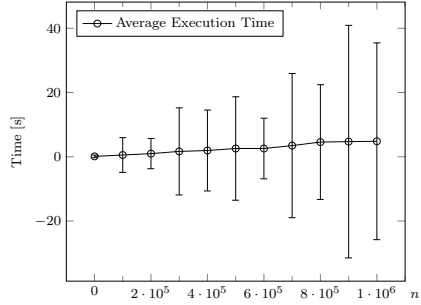
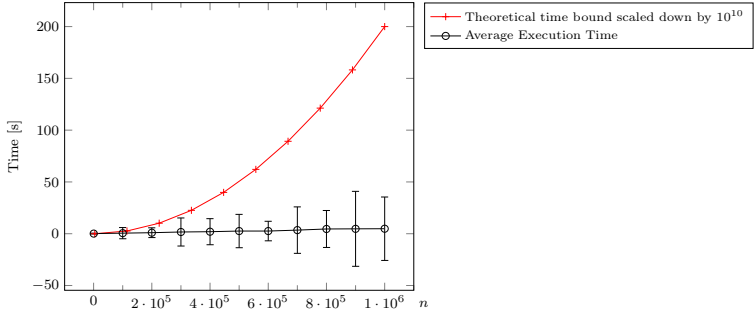
\begin{figure}[tb]
	\centering
	\subfloat[Test~1 results.]{\label{Table:Test1}
		\begin{tabular}[b]{| r | r | r |}
			\hline
		 	\multicolumn{1}{|c|}{$n$}  & \multicolumn{1}{c|}{$\mu$ (sec)} & \multicolumn{1}{c|}{$\sigma$} \\
			\hline
			   $<4\cdot 10^2$ &  0.13 & 0.42\\
			   $1\cdot 10^5$ &   0.55 & 5.41 \\
			   $2\cdot 10^5$ &   0.99 & 4.71 \\
			   $3\cdot 10^5$ &   1.67 & 13.55 \\
			   $4\cdot 10^5$ &   1.95 & 12.59 \\
			   $5\cdot 10^5$ &   2.58 & 16.10 \\
			   $6\cdot 10^5$ &   2.58 & 9.43 \\
			   $7\cdot 10^5$ &   3.48 & 22.43 \\
			   $8\cdot 10^5$ &   4.58 & 17.85 \\
			   $9\cdot 10^5$ &  4.72  & 36.19 \\
			   $10\cdot 10^5$ &  4.83 & 30.62 \\
			\hline
		\end{tabular}
	}
	\qquad
	\subfloat[Interpolation of average execution times of Test~1.]{\label{SubFig:Test1}
 	    \begin{tikzpicture}[scale=0.65]
			\begin{axis}[legend pos=north west, 
				xlabel={$n$},
				scaled x ticks=false,
				minor x tick num=1,
				ylabel={Time},
				y unit=s,
 				/pgfplots/ylabel near ticks,
 				/pgfplots/xlabel near ticks,
				xlabel style={ at={(ticklabel cs:1)}, anchor=south west}
				]
				\addplot[mark=o, 
						error bars/.cd,
						y dir=both, y explicit, 
					]
				    	table[row sep=crcr, x=x,y=y,y error=yerr] {
				        x       y       yerr  \\
				         400  0.13     0.42\\
				      100000  0.55   5.41   \\
				      200000  0.99   4.71   \\
				      300000  1.67   13.55  \\
				      400000  1.95   12.59  \\
				      500000  2.58   16.10  \\ 
				      600000  2.58   9.43    \\
				      700000  3.48   22.43  \\
				      800000  4.58   17.85   \\
				      900000  4.72   36.19   \\
				     1000000  4.83   30.62  \\
				}; 
			\addlegendentry{Average Execution Time}
			\end{axis}
		\end{tikzpicture}
	}
	\qquad
	\subfloat[Comparison between theoretical computation times and experimental ones.\label{SubFig:Test1Comparison}]{
	    \begin{tikzpicture}[scale=0.67]
			\begin{axis}[
				legend cell align=left,
    			legend pos=outer north east, 
				xlabel={$n$},
				scaled x ticks=false,
				minor x tick num=1,
				ylabel={Time},
				y unit=s,
				/pgfplots/ylabel near ticks,
				/pgfplots/xlabel near ticks,
				xlabel style={ at={(ticklabel cs:1)}, anchor=south west}
				]
				\addplot[red,mark=+,domain=4e3:10e5,samples=10] {((2*x)*x/1e10)};
				\addlegendentry{\scriptsize Theoretical time bound scaled down by $10^{10}$}
				\addplot[mark=o, 
						error bars/.cd,
						y dir=both, y explicit, 
					]
				    	table[row sep=crcr, x=x,y=y,y error=yerr] {
				        x       y       yerr  \\
				        400  0.13     0.42\\
				      100000  0.55   5.41   \\
				      200000  0.99   4.71   \\
				      300000  1.67   13.55  \\
				      400000  1.95   12.59   \\
				      500000  2.58   16.10   \\ 
				      600000  2.58   9.43    \\
				      700000  3.48   22.43   \\
				      800000  4.58   17.85   \\
				      900000  4.72   36.19   \\
				     1000000  4.83   30.62   \\
				}; 
			\addlegendentry{\scriptsize Average Execution Time}
			\end{axis}
		\end{tikzpicture}
	}
	\caption{Results of Test~1: average execution times ($\mu$) and relative standard deviations ($\sigma$) 
over a range of different \TN{} orders $n$. Times are in seconds. Each data set comprised of $1600$ \TN instances of unknown consistency
state.}\label{FIG:Test1}
\end{figure}

The instances in Test~1 come from the \texttt{randomgame} generator, 
except those for the first row of the table in \figref{Table:Test1} which have been built by the ATAPIS workflow random generator.
In particular, for each $n\in \{1\cdot10^5, 2\cdot 10^5, \ldots, 10\cdot 10^5\}$,
$1600$ \TN instances with maximum weight $W:=1000$ and unknown consistency state were generated by \texttt{randomgame},
whereas $1600$ \TN{s} of unknown consistency state and order $n$ around $400$ were generated by ATAPIS.
The results of the test are summarized in \figref{FIG:Test1}, where each execution mean time is depicted as a point with a vertical
bar representing its confidence interval determined according to its standard deviation.

The depicted function interpolating the mean values shows that the practical performance of the algorithms is definitely better than the theoretical worst-case
bound of $O(( n + m' ) m W)$; in our case this last is $O(n^2)$ since in the generated data sets
$W$ is constant and $m,m' \in \Theta(n)$.
\figref{SubFig:Test1Comparison} depicts the interpolating function of experimental execution times and, 
in red, the function $n^2/10^{10}$ as a reasonable surrogate for the worst-case execution time.
The comparison shows that the algorithm performs very well in real case executions.

However, since the standard deviation observed in the experiment is not negligible, we further investigated the behavior of the algorithm and we discovered that
there is a correlation between the execution time of the algorithm and the consistency state of the input \TN.
Therefore, $\mu$ and $\sigma$ were recalculated considering two kind of \TN sets: one having all consistent \TN{s}, and the other having all not consistent
\TN{s}.

\begin{figure}[t!b]
\centering
\captionsetup[subfigure]{width=.42\textwidth}
	\subfloat[Average execution times for consistent \TN{s} obtained from workflow graphs with $N$
	activities.\label{Table:Test1WFConsistent}]{\small
		\begin{tabular}[b]{| r | r | r |}
			\hline
			 \multicolumn{1}{|c|}{$N$}  & \multicolumn{1}{c|}{$\mu$ (sec)} & \multicolumn{1}{c|}{$\sigma$} \\
			\hline
			   10 & $6.42\cdot 10^{-5}$ & $1.22\cdot 10^{-5}$ \\
			   20 & $1.05\cdot 10^{-4}$ & $4.85\cdot 10^{-5}$ \\
			   30 & $1.50\cdot 10^{-4}$ & $5.7\cdot 10^{-5}$ \\
			   40 & $2.43\cdot 10^{-4}$ & $1.04\cdot 10^{-4}$ \\
			   50 & $3.20\cdot 10^{-4}$ & $1.78\cdot 10^{-4}$ \\
			   60 & $3.77\cdot 10^{-4}$ & $1.38\cdot 10^{-4}$ \\
			   70 & $4.77\cdot 10^{-4}$ & $1.28\cdot 10^{-4}$ \\
			   80 & $5.73\cdot 10^{-4}$ & $1.80\cdot 10^{-4}$ \\
			   90 & $6.82\cdot 10^{-4}$ & $2.79\cdot 10^{-4}$ \\
			  100 & $8.89\cdot 10^{-4}$ & $4.10\cdot 10^{-4}$  \\
			\hline
			\end{tabular}
	}
	\qquad   
 	\subfloat[Interpolation of average execution times of Table~\ref{Table:Test1WFConsistent}.\label{Fig:Test1WFclassifiedConsistent}]{\small
		\begin{tikzpicture}[scale=0.62, domain=0:4]
		\begin{axis}[
				legend pos=north west
				, ymin=0, ymax=0.00138
				, xlabel={$N$}
				, ylabel={Time}
				, y unit=s
				, scaled x ticks=false
				, /pgfplots/ylabel near ticks,
				, /pgfplots/xlabel near ticks,
				, xlabel style={ at={(ticklabel cs:1)}, anchor=south west}
			]
			\addplot[mark=o, 
			error bars/.cd, 
			y dir=both, 
			y explicit,
			]
		      table[row sep=crcr,x=x,y=y,y error=yerr] {
		        x       y       yerr  
		      10  0.000064    0.00001 \\
		      20  0.000104    0.00004 \\
		      30  0.000149    0.00005 \\
		      40  0.000243    0.00010 \\
		      50  0.000319    0.00017 \\
		      60  0.000376    0.00013 \\
		      70  0.000476    0.00012 \\
		      80  0.000573    0.00017 \\
		      90  0.000682    0.00027 \\
		     100  0.000888    0.00040 \\
		}; 
		\end{axis}
	\end{tikzpicture}
	}
	\\
	\subfloat[Average execution times for not consistent \TN{s} obtained from random workflow graphs with $N$
	activities.\label{Table:Test1WFInconsistent}]{\small
	    \begin{tabular}[b]{| r | r | r |}
		\hline
		 \multicolumn{1}{|c|}{$N$}  & \multicolumn{1}{c|}{$\mu$ (sec)} & \multicolumn{1}{c|}{$\sigma$} \\
		\hline
	      10 & $4.45\cdot 10^{-4}$  & $1.38\cdot 10^{-3}$ \\
	      20 & $1.50\cdot 10^{-3}$  & $5.10\cdot 10^{-3}$ \\
	      30 & $4.04\cdot 10^{-3}$  & $1.48\cdot 10^{-2}$ \\  
	      40 & $1.10\cdot 10^{-2}$  & $3.62\cdot 10^{-2}$ \\
	      50 & $1.64\cdot 10^{-2}$  & $8.42\cdot 10^{-2}$ \\
	      60 & $4.36\cdot 10^{-2}$  & $1.20\cdot 10^{-1}$ \\
	      70 & $8.08\cdot 10^{-2}$  & $2.71\cdot 10^{-1}$ \\
	      80 & $1.31\cdot 10^{-1}$  & $4.20\cdot 10^{-1}$ \\
	      90 & $1.59\cdot 10^{-1}$  & $5.22\cdot 10^{-1}$ \\
	     100 & $2.59\cdot 10^{-1}$  & $8.46\cdot 10^{-1}$ \\
		\hline
		\end{tabular}
	}
	\qquad
	\subfloat[Interpolation of average execution times of Table~\ref{Table:Test1WFInconsistent}.\label{Fig:Test1WFclassifiedInconsistent}]{\small
		\begin{tikzpicture}[scale=0.62, domain=0:4]
		\begin{axis}[
				legend pos=north west
				, ymin=-0.67, ymax=1.2
				, xlabel={$N$}
				, ylabel={Time}
				, y unit=s
				, scaled x ticks=false
				, /pgfplots/ylabel near ticks,
				, /pgfplots/xlabel near ticks,
				, xlabel style={ at={(ticklabel cs:1)}, anchor=south west}
			]
		\addplot[mark=o,error bars/.cd,y dir=both,y explicit]
		      table[row sep=crcr, x=x,y=y,y error=yerr] {
		        x       y       yerr  \\
		      10 0.000445       0.001380583434   \\
		      20 0.001502992113   0.005095620104  \\
		      30 0.004048780273  0.01486809066  \\  
		      40 0.01103171587   0.03621212691  \\
		      50 0.01644273384   0.08424989611  \\
		      60 0.04368451081   0.1193160592  \\
		      70 0.08077450854   0.271221609  \\
		      80 0.1318617099   0.4200639859  \\
		      90 0.1583977008   0.5228110804  \\
		     100 0.2595240291   0.8460078121  \\
		}; 
		\end{axis}
	\end{tikzpicture}
	}
	\caption{Average execution times obtained in Test~1 calculated considering samples of either all consistent or all not consistent
	\TN{s} obtained from workflow graphs.}\label{Fig:Test1-WFclassified}
\end{figure}
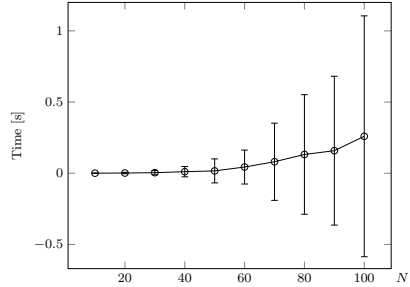

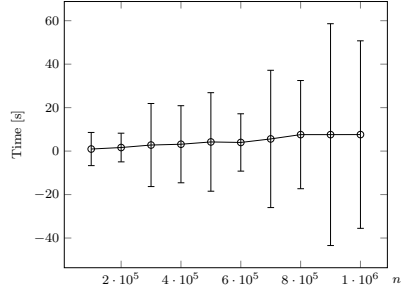
\begin{figure}[t!b]
\centering
	\begingroup
	\centering
	\captionsetup[subfigure]{width=.4\textwidth}
	\subfloat[Average execution times for consistent \TN{s} obtained from 
		MPGs.]{\label{Table:Test1MPGConsistent}\small
		\begin{tabular}[b]{| r | r | r |}
			\hline
			 \multicolumn{1}{|c|}{$n$}  & \multicolumn{1}{c|}{$\mu$ (sec)} & \multicolumn{1}{c|}{$\sigma$} \\
			\hline
			  $1\cdot 10^5$ & 0.16 & 0.04   \\
			  $2\cdot 10^5$ & 0.35 & 0.07   \\
			  $3\cdot 10^5$ & 0.56 & 0.01   \\
			  $4\cdot 10^5$ & 0.75 & 0.02   \\
			  $5\cdot 10^5$ & 0.96 & 0.02   \\
			  $6\cdot 10^5$ & 1.18 & 0.03   \\
			  $7\cdot 10^5$ & 1.38 & 0.03   \\
			  $8\cdot 10^5$ & 1.59 & 0.04   \\
			  $9\cdot 10^5$ & 1.86 & 0.06   \\
			  $10\cdot 10^5$ & 2.07 &  0.08 \\
			\hline
			\end{tabular}
	}
	\qquad\qquad
 	\subfloat[Interpolation of average execution times of Table~\ref{Table:Test1MPGConsistent}.\label{Fig:Test1MPGclassifiedConsistent}]{\small
	 	\begin{tikzpicture}[scale=0.62, domain=0:4]
		\begin{axis}[
				legend pos=north west
				, xlabel={$n$}
				, ylabel={Time}
				, y unit=s
				, scaled x ticks=false
				, /pgfplots/ylabel near ticks,
				, /pgfplots/xlabel near ticks,
				, xlabel style={ at={(ticklabel cs:1)}, anchor=south west}
			]
			\addplot[mark=o, 
			error bars/.cd, 
			y dir=both, 
			y explicit,
			]
		      table[row sep=crcr,x=x,y=y,y error=yerr] {
		        x       y       yerr  
		      100000  0.156  0.041   \\
		      200000  0.349   0.069  \\
		      300000  0.557  0.0141  \\
		      400000  0.752  0.0169  \\
		      500000  0.957  0.0184  \\
		      600000  1.185   0.0326 \\
		      700000  1.379  0.0293  \\
		      800000  1.586   0.0405 \\
		      900000  1.860   0.0603 \\
		     1000000  2.068  0.0851  \\
		}; 
		\end{axis}
	\end{tikzpicture}
	}
	\\
	\subfloat[Average execution times for not consistent \TN{s} obtained from MPGs.\label{Table:Test1MPGInconsistent}]{\small
	    \begin{tabular}[b]{| r | r | r |}
		\hline
		 \multicolumn{1}{|c|}{$n$}  & \multicolumn{1}{c|}{$\mu$ (sec)} & \multicolumn{1}{c|}{$\sigma$} \\
		\hline
	      $1\cdot 10^5$ & 0.95  & 7.63  \\
	      $2\cdot 10^5$ & 1.64  & 6.60  \\
	      $3\cdot 10^5$ & 2.79  & 19.11  \\  
	      $4\cdot 10^5$ & 3.15  & 17.73  \\
	      $5\cdot 10^5$ & 4.21  & 22.67  \\
	      $6\cdot 10^5$ & 3.98  & 13.19  \\
	      $7\cdot 10^5$ & 5.60  & 31.58  \\
	      $8\cdot 10^5$ & 7.58  & 24.89  \\
	      $9\cdot 10^5$ & 7.58  & 51.03  \\
	     $10\cdot 10^5$ & 7.60  & 43.14  \\
		\hline
		\end{tabular}
	}
	\qquad
	\subfloat[Interpolation of average execution times of Table~\ref{Table:Test1MPGInconsistent}.\label{Fig:Test1MPGclassifiedInconsistent}]{\small
		\begin{tikzpicture}[scale=0.62, domain=0:4]
		\begin{axis}[
				legend pos=north west
				, xlabel={$n$}
				, ylabel={Time}
				, y unit=s
				, scaled x ticks=false
				, /pgfplots/ylabel near ticks,
				, /pgfplots/xlabel near ticks,
				, xlabel style={ at={(ticklabel cs:1)}, anchor=south west}
			]
		\addplot[mark=o,error bars/.cd,y dir=both,y explicit]
		      table[row sep=crcr, x=x,y=y,y error=yerr] {
		        x       y       yerr  \\
		      100000  0.950   7.632 \\
		      200000  1.643   6.598 \\
		      300000  2.791   19.107\\  
		      400000  3.154   17.734\\
		      500000  4.211   22.667\\
		      600000  3.977   13.193\\
		      700000  5.600   31.585\\
		      800000  7.581   24.890\\
		      900000  7.580   51.033\\
		     1000000  7.601   43.136\\
		}; 
		\end{axis}
	\end{tikzpicture}
	}
	\endgroup
	\caption{Average execution times obtained in Test~1 calculated 
	for samples of either all consistent or all not consistent
	\TN{s} obtained from MPG graphs.}%
	\label{Fig:Test1-classified}
\end{figure}

\figref{Fig:Test1-WFclassified} depicts average execution times obtained 
in Test~1 calculated considering samples of either all consistent or all not consistent
\TN{s} obtained from workflow graphs. 
\figref{Fig:Test1-classified} offers the same view but for \TN{s} obtained from MPG graphs. 
In general, the mean execution times for consistent \TN{s} are 
smaller than the corresponding ones for not consistent \TN{s}; 
furthermore, they also exhibit a negligible standard deviation. 
However, for samples of consistent \TN{s} obtained from workflows, the standard deviation is not 
negligible even for samples with size $N=20$.
Part of the reasons for this behavior is given by the structure of the data sets: 
in each data set \TN{s} can differ a lot with respect to their order and, therefore,
they may require very different execution times. 
For example, the data set relating to workflow graphs with 20 activities contains \TN{s} with order in
range $[48, 94]$. Since the number of activities is usually considered as main parameter in workflow community, 
we wanted to maintain such structure of data set and experiment results to emphasize 
the dependency of execution time with respect to such number.

On the other side, for consistent \TN{s} determined from MPGs, 
the observed standard deviation $\sigma$ is always less that the 10\% of the average
execution time $\mu$ with 99\% level of confidence, while for not consistent \TN{s} 
it has not been possible to determine any confidence level because the
observed standard deviation $\sigma$ resulted to be always high due to some hard instances.

Even though procedure \texttt{solveMPG-Threshold()}
could require up to $\Theta(W)$ updates according to the theoretical worst-case bound,
our experiments suggest that, in practice, some kind of dependency of the running time on $W$
is appreciable only for a few \MPG games, all associated to not consistent \TN instances.
\medskip

The goal of Test~2 was to determine the average computation time of 
Algorithm~\ref{ALG:PseudocodeDecisionTNC} for different 
values of $W$, in order to understand how much the practical computation time is dependent on $W$.
Therefore, we considered three possible edge weight ranges, 
$[10^2, 10^3]$,  $[10^5, 10^6]$, and $[10^8, 10^9]$, 
and for each of them two data sets have been built using the \texttt{randomgame} generator, 
one comprising only consistent \TN{s}, and the other only not consistent ones. 
Each data set comprised of $800$ \TN{s} instances having $|V|=10^5$ nodes, $m,m'\in \Theta(n)$ and 
edge weights in the corresponding weight range.
\figref{Fig:Test2} depicts the results on these six data sets.
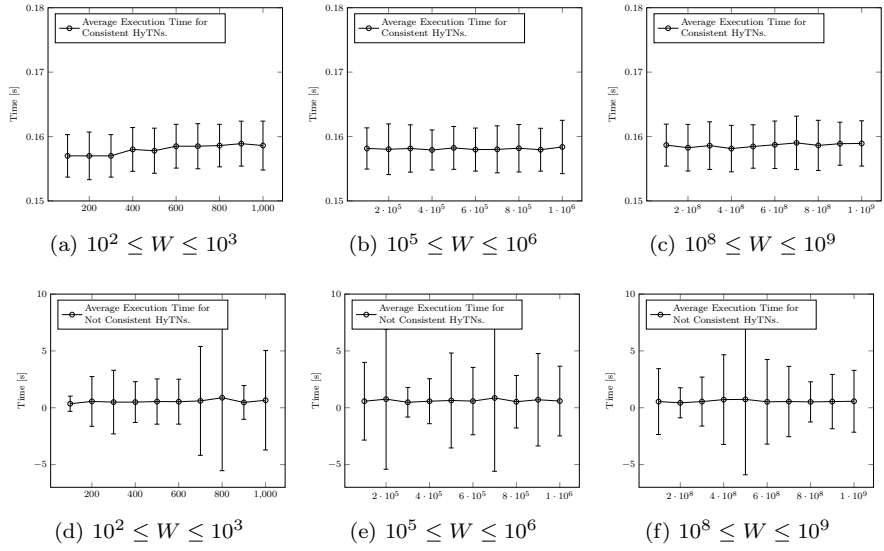
\begin{figure}[tb]
	\centering
	\subfloat[$10^2\leq W\leq 10^3$]{
		\begin{tikzpicture}[scale=0.45]
			\begin{axis}[
				ymin=0.15, ymax=0.18 , 
				 legend pos=north west 
				, xlabel={}
				, ylabel={Time}
				, y unit=s  
				, scaled x ticks=false
				, ylabel near ticks
				, xlabel near ticks
			]
			\addplot[
				mark=o 
				, error bars/.cd 
				, y dir=both, 
				, y explicit
				]
			      table[row sep=crcr, x=x,y=y,y error=yerr] {
			        x          y     yerr  \\   
			        100   0.157    0.0033  \\
			        200   0.157    0.0037  \\    
			        300   0.157    0.0033 \\ 
			        400   0.158    0.0034  \\   
			        500   0.1578   0.0035  \\      
			        600   0.1585   0.0034  \\    
			        700   0.1585   0.0035  \\   
			        800   0.1586   0.0033  \\    
			        900   0.1589   0.00349  \\     
			       1000   0.1586   0.00380 \\     
				}; 
			\addlegendentry[text width=4cm,text depth=]{Average Execution Time for Consistent \TN{s}.}
			\end{axis}
		\end{tikzpicture}
	}
	\ 
	\subfloat[$10^5\leq W\leq 10^6$]{
		\begin{tikzpicture}[scale=0.45]
			\begin{axis}[
				ymin=0.15, ymax=0.18, 
				legend pos=north west 
				, xlabel={}
				, ylabel={Time}
				, y unit=s
				, scaled x ticks=false
				, ylabel near ticks
				, xlabel near ticks
				]
			\addplot[mark=o 
				, error bars/.cd 
				, y dir=both, 
				, y explicit
				]
			      table[row sep=crcr, x=x,y=y,y error=yerr] {
			        x          y       yerr                 \\  
			       100000     0.15814654      0.0032        \\
			       200000     0.1580240287    0.003936462   \\ 
			       300000     0.158143965     0.003672156   \\    
			       400000     0.157910473     0.003111354   \\       
			       500000     0.1582324944    0.003327139   \\    
			       600000     0.1579715081    0.003340000   \\      
			       700000     0.1580064657    0.003649407   \\       
			       800000     0.1581810774    0.003689283   \\       
			       900000     0.1579512434    0.003327979   \\        
			      1000000     0.1583781773    0.004140331   \\        
			}; 
			\addlegendentry[text width=4cm,text depth=]{Average Execution Time for Consistent \TN{s}.}
			\end{axis}
		\end{tikzpicture}
	}
	\ 
	\subfloat[$10^8\leq W\leq 10^9$]{
		\begin{tikzpicture}[scale=0.45]
			\begin{axis}[
				ymin=0.15, ymax=0.18, 
				 legend pos=north west 
				, xlabel={} 
				, ylabel={Time}
				, y unit=s 
				, scaled x ticks=false
				, ylabel near ticks
				, xlabel near ticks
				]
				\addplot[mark=o 
					, error bars/.cd 
					, y dir=both, 
					, y explicit
					]
				      table[row sep=crcr, x=x,y=y,y error=yerr] {
				        x          y      yerr     \\  
				      100000000   0.1586706438  0.003260788918  \\
				      200000000   0.1582644341  0.003622459372  \\ 
				      300000000   0.1585910405  0.003695868666  \\
				      400000000   0.1581284895  0.003611277689  \\
				      500000000   0.1584529272  0.003372738464  \\
				      600000000   0.1587233398  0.003681266929  \\ 
				      700000000   0.1590139164  0.004149306268  \\
				      800000000   0.1586251216  0.003890676756  \\
				      900000000   0.1588834019  0.003342266923  \\
				     1000000000   0.1589280664  0.003514581468  \\
				}; 
				\addlegendentry[text width=4cm,text depth=]{Average Execution Time for Consistent \TN{s}.}
			\end{axis}
		\end{tikzpicture}
	}
	\\
	\subfloat[$10^2\leq W\leq 10^3$]{
		\begin{tikzpicture}[scale=0.45]
			\begin{axis}[
				ymin=-7, ymax=10,
				 legend pos=north west 
				, xlabel={} 
				, ylabel={Time}
				, y unit=s 
				, scaled x ticks=false
				, ylabel near ticks
				, xlabel near ticks
				]
				\addplot[mark=o 
					, error bars/.cd 
					, y dir=both, 
					, y explicit
					]
				      table[row sep=crcr, x=x,y=y,y error=yerr] {
				        x        y        yerr   \\   
				      100      0.3548427975  0.6685780348   \\ 
				      200      0.5580607241  2.188622344   \\  
				      300      0.4968850115  2.798099369   \\
				      400      0.4995424185  1.792345753   \\
				      500      0.5453294182  1.991119645   \\
				      600      0.5333053331  1.976063337   \\
				      700      0.6035481124  4.790754101   \\
				      800      0.8821451701  6.41887248   \\
				      900      0.4669560195  1.484196662   \\
				     1000      0.6597385124  4.371938417   \\
					}; 
				\addlegendentry[text width=4cm,text depth=]{Average Execution Time for Not Consistent \TN{s}.}
			\end{axis}
		\end{tikzpicture} 
	}
	\ 
	\subfloat[$10^5\leq W\leq 10^6$]{
		\begin{tikzpicture}[scale=0.45]
			\begin{axis}[
				 ymin=-7, ymax=10,
				legend pos=north west 
				, xlabel={} 
				, ylabel={Time}
				, y unit=s 
				, scaled x ticks=false
				, ylabel near ticks
				, xlabel near ticks
				]
				\addplot[mark=o 
					, error bars/.cd 
					, y dir=both, 
					, y explicit
					]
				      table[row sep=crcr, x=x,y=y,y error=yerr] {
				        x          y      yerr    \\  
				      100000    0.5716545053   3.41846512   \\
				      200000    0.7536045993   6.163441069  \\ 
				      300000    0.4801901184   1.303251235  \\
				      400000    0.5732116381   1.977771118  \\
				      500000    0.6378956768   4.175609409  \\
				      600000    0.584690086    2.962003938  \\ 
				      700000    0.8604524406   6.451991559  \\
				      800000    0.5295288557   2.310899484  \\
				      900000    0.7028723913   4.067331303  \\
				     1000000    0.5893957813   3.064948818  \\
				}; 
				\addlegendentry[text width=4cm,text depth=]{Average Execution Time for Not Consistent \TN{s}.}
			\end{axis}
		\end{tikzpicture}
	}
	\ 
	\subfloat[$10^8\leq W\leq 10^9$]{
		\begin{tikzpicture}[scale=0.45, domain=0:4]
			\begin{axis}[
				 ymin=-7, ymax=10,
				 legend pos=north west 
				, xlabel={} 
				, ylabel={Time}
				, y unit=s 
				, scaled x ticks=false
				, ylabel near ticks
				, xlabel near ticks
				]
				\addplot[mark=o 
					, error bars/.cd 
					, y dir=both, 
					, y explicit
					]
				      table[row sep=crcr, x=x,y=y,y error=yerr] {
				        x          y      yerr        \\  
				      100000000  0.5421311039 2.898779502     \\
				      200000000  0.4333790206 1.324697828     \\ 
				      300000000  0.5409006207 2.154921304     \\
				      400000000  0.7153354511 3.944879421    \\
				      500000000  0.7347573893 6.636128285     \\
				      600000000  0.5208934125 3.721867528     \\ 
				      700000000  0.545149139  3.094114389     \\
				      800000000  0.5131736499 1.767497354     \\
				      900000000  0.5417591056 2.391126746     \\
				     1000000000  0.5674673492 2.72155958     \\
				}; 
				\addlegendentry[text width=4cm,text depth=]{Average Execution Time for Not Consistent \TN{s}.}
			\end{axis}
		\end{tikzpicture}
	}
	\caption{Average execution times in Test~1 calculated considering samples of either all consistent or all not consistent \TN{s}.}\label{Fig:Test2}
\end{figure}
Applying the worst-case analysis for these data sets, 
it results that the time complexity should be $O(W)$ since $n$, $m$ and $m'$ are constants.
On the contrary, the determined interpolation functions representing the experimental 
execution times do not show any clear dependence on $W$. This result suggests that, in practice, 
uniform random weighted instances are decided very quickly with respect to the magnitude of their weights and that the
algorithm does not seem to exhibit the worst-case pseudo-polynomial behavior predicted in the asymptotic analysis.
Moreover, the average execution times for each data set comprising only consistent \TN{s} 
are less than those for the corresponding data of only not consistent \TN{s}. 
Only for consistent \TN{s} data sets the standard deviation was below 7\% than the average execution time with a confidence of 99\%.
\medskip

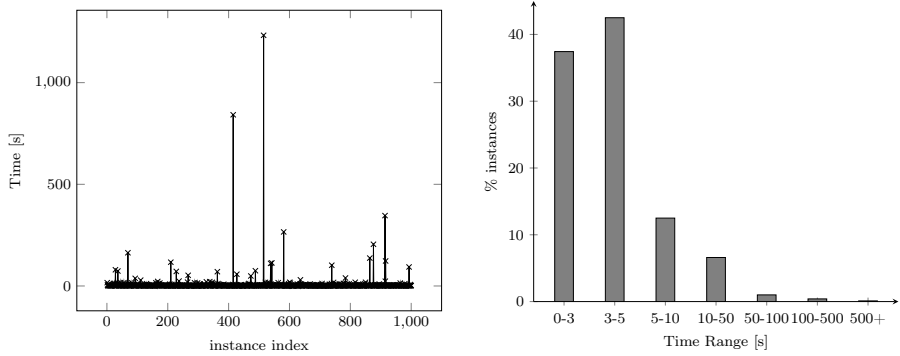
\begin{figure}[tb]
	\subfloat[Execution times of $10^3$ not consistent \TN{s} instances with $n=10^6$ and $W=10^3$.]{\label{Fig:BadInstanceDistribution}
		\begin{tikzpicture}[scale=0.70]
			\begin{axis}[legend pos=north west
				, 
				, xlabel={instance index}
				, scaled x ticks=false 
				, xlabel near ticks
				, ylabel={Time}
				, y unit=s
				, ylabel near ticks
				]
				\addplot[mark=x] table[x=x,y=y] {sun_distro.txt}; 
			\end{axis}
		\end{tikzpicture}
	}
	\quad
	\subfloat[Instance classification with respect to \texttt{solveMPG-Threshold()} execution time.]{\label{Fig:BadInstanceHistogram}
	\begin{tikzpicture}[scale=0.70]
	    \begin{axis}[
	      ymin=0,
	      ymax=45,
	      ybar=10pt,
	      ylabel={\% instances},
	      ylabel near ticks,
	      axis y line=left,
	      xlabel={Time Range},
	      x unit=s,
	      xtick=data,
	      symbolic x coords={0-3,3-5,5-10,10-50,50-100,100-500,500+},
	      axis x line=bottom,
	      xlabel near ticks,
	      xticklabel style={anchor=base,yshift=-\baselineskip},
	      bar width=10pt,
	      enlarge x limits=0.1
	    ]
	      \addplot[fill=gray] coordinates {
	        (0-3,37.43)
			(3-5,42.5)
			(5-10,12.5)
	        (10-50,6.6)
	        (50-100,1)
	        (100-500,0.4)
		(500+,0.1)
	      };
	    \end{axis}
	\end{tikzpicture}
	}
	\caption{\texttt{solveMPG-Threshold()} execution times obtained in Test~3
	determined considering samples of all not consistent \TN instances.}\label{Fig:Test3}
\end{figure}

In order to better understand how some execution times affect the value of the standard deviation, 
we conducted a third experiment, Test~3, with the goal to visualize the distribution of the instances with  
computation times significantly above the average.
Procedure \texttt{solveMPG-Threshold()} has been executed on 
$10^3$ randomly generated not consistent \TN{s}, each having order $n=10^6$ and $W\approx 10^3$.
The determined running times are depicted in \figref{Fig:BadInstanceDistribution}: 
most of the instances are decided very quickly, \ie in a time between $0$ and $10$ seconds, 
while a smaller portion of the \TN{s} required a time between $10$ and $500$ seconds. 
In more details, in repeated tests we verified that, approximately, $1\%$ of the \TN instances required 
a time between $50$ and $100$ seconds to be decided, $0.4\%$ required a time
between $100$ and $500$ seconds, and, finally, only $0.1\%$ required more than $500$ seconds. 
These results are shown in \figref{Fig:BadInstanceHistogram}.

Such behavior has been confirmed in other tests with different graph orders and maximum edge weight values.
In several experiments we conducted, we observed that the maximum execution time of the algorithm keeps growing as we enlarge the size of the dataset. 
This explains why the standard deviation can't be reduced. If we could characterize 
such hard instances in general, we would be making a major progress in understanding the computational complexity of \MPG{s}.
We didn't find any pattern or property that characterizes the found hard instance.
Here we can only show a simple family of \TN{s} instances in which the execution time grows linearly with $W$. 
The family is given by just one single \TN graph where only $W$ changes, as depicted in \figref{ex:HyTN_slow}.
The corresponding \MPG is shown in \figref{ex:MPG_slow} and provides a clear example where Brim's Value Iteration algorithm\cite{brim2011faster} 
performs poorly. It is worth noting that in the context of \MPG{s} this example can be reduced down to $6$ nodes.
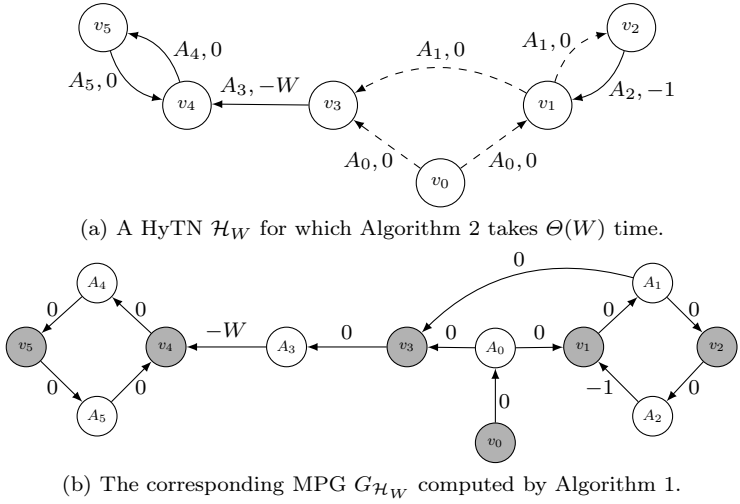
\begin{figure}[tb]
\centering
\subfloat[A \TN{} $\H_W$ for which Algorithm~\ref{ALG:PseudocodeDecisionTNC} takes $\Theta(W)$ time.]{\label{ex:HyTN_slow}
\begin{tikzpicture}[arrows=->, scale=0.85, node distance=2 and 1.5]
	\node[node] (A) {$v_5$};
	\node[node, right=of A, yshift=-10ex, xshift=-8ex] (B) {$v_4$};
	\node[node, right=of B] (C) {$v_3$};
	\node[node, right=of C, xshift=-5ex, yshift=-10ex] (D) {$v_0$};
	\node[node, right=of D, xshift=-5ex, yshift=10ex] (E) {$v_1$};
	\node[node, right=of E, yshift=10ex, xshift=-8ex] (F) {$v_2$};

	\draw[dashed] (D) to [bend right=0] node[below, xshift=2ex] {$A_0, 0$} (E);
	\draw[dashed] (D) to [bend right=0] node[below, xshift=-2ex] {$A_0, 0$} (C);
	\draw[dashed] (E) to [bend left=30] node[above, xshift=-3ex, yshift=-1ex] {$A_1, 0$} (F);
	\draw[dashed] (E) to [bend right=30] node[above] {$A_1, 0$} (C);
	\draw[] (F) to [bend left=30, xshift=2ex, yshift=1ex] node[below, xshift=4ex, yshift=1ex] {$A_2, -1$} (E);
	\draw[] (C) to [xshift=2ex, yshift=1ex] node[above] {$A_3,-W$} (B);
	\draw[] (B) to [bend right=30, xshift=0ex, yshift=0ex] node[above, xshift=4ex, yshift=-2ex] {$A_4, 0$} (A);
	\draw[] (A) to [bend right=30, xshift=-1ex, yshift=1ex] node[below, xshift=-4ex, yshift=2ex] {$A_5, 0$} (B);

\end{tikzpicture}\label{FIG:Slow_HyTN}
}\qquad
\subfloat[The corresponding \MPG $G_{\H_W}$ computed by Algorithm~\ref{ALGO:PseudocodeReduction-htn-mpg}.\label{ex:MPG_slow}]{
\begin{tikzpicture}[arrows=->, scale=0.7, node distance=2 and 1.5]
	\node[node] (A5) {$A_4$};
	\node[node, blackNode, right=of A5, yshift=-10ex, xshift=-8ex] (B) {$v_4$};
	\node[node, blackNode, left=of B, yshift=0ex, xshift=-3ex] (G) {$v_5$};
	\node[node, below=of A5, yshift=2ex, xshift=0ex] (A6) {$A_5$};
	\node[node, right=of B] (A4) {$A_3$};
	\node[node, blackNode, right=of A4] (C) {$v_3$};
	\node[node, blackNode, right=of C, xshift=-5ex, yshift=-15ex] (D) {$v_0$};
	\node[node, above=of D, yshift=-8ex] (A1) {$A_0$};
	\node[node, blackNode, right=of D, xshift=-5ex, yshift=15ex] (E) {$v_1$};
	\node[node, right=of E, yshift=10ex, xshift=-8ex] (A2) {$A_1$};
	\node[node, blackNode, right=of E, yshift=0ex, xshift=2ex] (F) {$v_2$};
	\node[node, below=of A2, yshift=2ex, xshift=0ex] (A3) {$A_2$};
	\draw[] (D) to [bend right=0] node[below, xshift=1ex, yshift=1ex] {$0$} (A1);
	\draw[] (A1) to [bend right=0] node[above] {$0$} (E);
	\draw[] (A1) to [bend right=0] node[above] {$0$} (C);
	\draw[] (E) to [bend right=0] node[above, xshift=-1ex, yshift=-1ex] {$0$} (A2);
	\draw[] (A2) to [bend right=0] node[above, xshift=1ex, yshift=-1ex] {$0$} (F);
	\draw[] (A2) to [bend right=30] node[above] {$0$} (C);
	\draw[] (F) to [bend right=0] node[below, yshift=1ex, xshift=1ex] {$0$} (A3);	
	\draw[] (A3) to [bend right=0] node[below, yshift=1ex, xshift=-2ex] {$-1$} (E);
	\draw[] (C) to [bend right=0] node[above] {$0$} (A4);	
	\draw[] (A4) to [bend right=0] node[above] {$-W$} (B);	
	\draw[] (B) to [bend right=0] node[above, xshift=1ex, yshift=-1ex] {$0$} (A5);	
	\draw[] (A5) to [bend right=0] node[above, xshift=-1ex, yshift=-1ex] {$0$} (G);	
	\draw[] (G) to [bend right=0] node[below, xshift=-1ex, yshift=1ex] {$0$} (A6);	
	\draw[] (A6) to [bend right=0] node[below, xshift=1ex, yshift=1ex] {$0$} (B);	
\end{tikzpicture}\label{FIG:Slow_MPG}
}
\caption{A \TN which requires 
$\Theta(W)$ computation time by Algorithm~\ref{ALG:PseudocodeDecisionTNC}.}\label{FIG:SlowHyTN}
\end{figure}
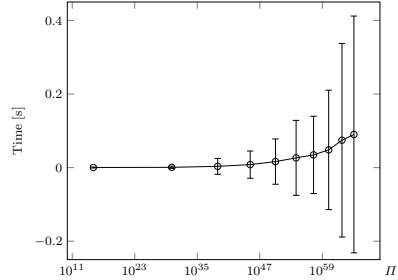
\begin{figure}[t!b]
\centering
	\subfloat[Average execution times obtained for different values of the product of the heads of hyperarcs $\Pi$.]{\label{Table:Test1Consistent}
		\begin{tabular}[b]{| r | r | r |}
			\hline
			 \multicolumn{1}{|c|}{$\Pi$}  & \multicolumn{1}{c|}{$\mu$ (sec)} & \multicolumn{1}{c|}{$\sigma$} \\
			\hline
			  $1.13\cdot 10^{15}$ &  $3.02\cdot 10^{-4}$  &  $1.36\cdot 10^{-3}$  \\
			  $1.27\cdot 10^{30}$ &  $7.78\cdot 10^{-4}$  &  $3.34\cdot 10^{-3}$  \\
			  $8.08\cdot 10^{38}$ &  $3.45\cdot 10^{-3}$  &  $2.15\cdot 10^{-2}$ \\
			  $1.43\cdot 10^{45}$ &  $8.13\cdot 10^{-3}$  &  $3.70\cdot 10^{-2}$  \\
			  $1.00\cdot 10^{50}$ &  $1.63 \cdot 10^{-2}$  &  $6.15\cdot 10^{-2}$  \\
                          $9.10\cdot 10^{53}$ &  $2.65\cdot 10^{-2}$  &  $1.02\cdot 10^{-1}$   \\
                          $2.02\cdot 10^{57}$ &  $3.46 \cdot 10^{-2}$  &  $1.05\cdot 10^{-1}$  \\ 
                          $1.61\cdot 10^{60}$ &  $4.82 \cdot 10^{-2}$  &  $1.62\cdot 10^{-1}$  \\
                          $5.80\cdot 10^{62}$ &  $7.44\cdot 10^{-2}$  &  $2.63\cdot 10^{-1}$   \\
                          $1.13\cdot 10^{65}$ &  $9.01\cdot 10^{-2}$  &  $3.21\cdot 10^{-1}$   \\
			\hline
			\end{tabular}
	}
	\qquad  
 	\subfloat[Interpolation of average execution times in \tabref{Table:Test1Consistent}.]{\label{Fig:Test1-classifiedConsistent}
		\begin{tikzpicture}[scale=0.6, domain=0:4]
		\begin{semilogxaxis}[
				legend pos=north west
				, ymin=-0.25, ymax=0.45
				, xlabel={$\Pi$}
				, xmode=log
				, ylabel={Time}
				, y unit=s
				, scaled x ticks=false
				, ylabel near ticks,
				, xlabel near ticks,
				, xlabel style={ at={(ticklabel cs:1)}, anchor=south west}
			]
			\addplot[mark=o, 
			error bars/.cd, 
			y dir=both, 
			y explicit,
			]
		      table[row sep=crcr, x=x, y=y, y error=yerr] {
		        x                y          yerr              \\ 
		      1.13E+15     0.000302         0.001365082805  \\
		      1.27E+30     0.000778         0.003338263814  \\
		      8.08E+38     0.003448915276   0.02158598964   \\
		      1.43E+45     0.008128432462   0.03702355521   \\
		      1.00E+50     0.01633870517    0.06150968792   \\
		      9.10E+53     0.02649881917    0.1018356509    \\
		      2.02E+57     0.03455488466    0.1050459339    \\
                      1.61E+60     0.04815697745    0.1622394937    \\
		      5.80E+62     0.07442923076    0.263059134     \\
                      1.13E+65     0.09007910558    0.321730001     \\ 
		}; 
		\end{semilogxaxis}
	\end{tikzpicture}
	}
\caption{Average Execution Times obtained in Test~4.}\label{Fig:TestProduct}
\end{figure}

Finally, in order to show how much the running time is 
dependent on the number of different positional strategies of one player, 
in Test~4 the average computation time has been calculated with respect to different values of the product of the heads
of hyperarcs (\ie $\prod_{A\in\A} |H_A|$) in a \TN. In particular, for each 
$\Pi\in \{10^{15}, 10^{30}, \ldots, 10^{65}\}$, $2500$ \TN{s} instances $(V,\A)$ each having $|V|=50$ nodes, $\prod_{A\in\A} |H_A|\approx \Pi$, and
$W=10^3$ have been generated by means of \texttt{randomgame} generator.
The results of the evaluation are depicted in \figref{Fig:TestProduct}, where $\Pi$ values are drawn in logarithmic scale. 
Analyzing the diagram in the figure it is possible to say that, experimentally, the average execution time increases only logarithmically
with respect to the number of different positional strategies of one player.
This results is quite interesting because, considering the \TN{} in \figref{ex:HyTN_slow}, it is evident that the time for checking a \TN{} is more dependent 
on the edge weight magnitude than on the number of different positional strategies of one player.

%

%% file: sectRelatedWork.tex
In the literature there are some extension proposals of the \STN model to augment the capability to represent temporal constraints. 


In the \STN seminal paper~\cite{DechterMP91}, Dechter \etal firstly proposed to consider the Temporal Constraint Satisfaction Problem (TCSP). A binary
constraint in a TCSP is represented using a set of intervals rather than a single interval as in an \STN. In particular, a binary constraint $C_{ij}
=\{[a_1,b_1], [a_2,b_2],\ldots, [a_l,b_l]\}$ between time points $x_i$ and $x_j$ represents the disjunction $a_1 \leq x_j - x_i \leq b_1 \vee a_2 \leq x_j - x_i
\leq b_2 \vee a_l \leq x_j - x_i \leq b_l$.
The problem of verifying consistency of a TCSP is \NP-complete as the 
same authors showed in the paper; hence, they finally propose to consider \STN{s} as a
tractable simplified model.

A similar kind of generalization considering disjunction of temporal distance constraints was proposed by Stergiou and Koubarakis~\cite{StergiouK00} defining
the Disjunctive Temporal Problem (DTP).
A DTP consists of a set of variables $X=\{x_1, x_2,\ldots,x_n\}$ having continuous domains and representing time points and a set of disjunctive difference
constraints between the time points in the form: $a_1 \leq x_{i_1} -x_{j_1} \leq b_1 \vee  a_2 \leq x_{i_2} -x_{j_2} \leq b_2 \vee \ldots \vee a_{k} \leq
x_{i_k} -x_{j_k} \leq b_k$; where $x_{i_1},x_{j_1},\ldots, x_{i_k},x_{j_k}$ are time points from $X$ and $a_1,b_1,\ldots, a_k,b_k$ are real numbers.
A DTP is consistent if there exists an instantiation of variables X to real numbers satisfying all the constraints. Since DTPs are a generalization of TCSPs,
also for DTPs the consistency check problem is \NP-complete. 
In~\cite{StergiouK00} the authors presented some of the theoretical results characterizing the possible
backtracking algorithms that solve the consistency problem in terms of search nodes visited and consistency checks performed.

In 2005, Kumar proposed to consider a restricted class of DTP in order to maintain some of the expressive power of DTPs but, at the same time, allowing an
efficient consistency check. In particular, in~\cite{Kumar05}, RDTPs (restricted DTPs) is defined as a disjunctive temporal problem where a constraint is one of
the following types: (Type 1) $(l \leq x_i-x_j \leq u)$, (Type 2) $(l_1\leq x_i \leq u_1)\vee(l_2 \leq x_i \leq u_2)\ldots(l_{j} \leq x_i \leq  u_{j})$, (Type3)
$(l_1 \leq x_i \leq u_1)\vee(l_2 \leq x_j \leq u_2)$, where $x_i$ and $x_j$ represent a timepoint variable, and $l_i, u_i$ real values. An RDTP instance can be
solved in strongly polynomial-time deterministic algorithm transforming it into a binary Constraint Satisfiability Problem (CSP) over meta variables representing
constraints of Type 2 or Type 3 and, then, showing that such binary constraints are also \textit{connected row-convex (CRC)} constraints, and, then, exploiting
the properties of CRC constraints. An instantiation of a consistency check algorithm for RDTPs that further exploits the structure of CRC constraints has a
running time complexity of $O((|TP_2| + |TP_3|)^3d^2_{\max} + (|TP_2| + |TP_3|)^2(N M+ d^2_{\max}))$, where $TP_2$ is the set of Type 2 constraints, $TP_3$ is
the set of Type 3 ones, $d_{\max}$ is the maximum number of disjuncts in any constraint, and $N\slash M$ is the number of the nodes\slash arcs of the instance,
respectively.
In the same paper, Kumar presented also a simpler and faster, but \textit{randomized}, algorithm for the same class RDTP.

An attempt to model some aspects of \STN{s} similar to those addressed by \TN{s} was lead in~\cite{conf/flairs/BartakC07}, where fun-in and fun-out subgraphs
much resembling our multi-tail and multi-head hyperarcs were considered.
However, since the problem $1$-in-$3$-SAT is \NP-complete even when all the literals comprising the clauses are positive, it readily follows that their models
lead to \NP-complete problems even when fun-out subgraphs (or fun-in subgraphs) are banned. As such, the opportunity for tractability spotlighted in this paper
is missed in those models.

Another approach to extend \STN is represented by the proposal of Khatib \etal~\cite{KhatibMMR01,KhatibMMRSV07}.
They introduced the characterization of \textit{hard} and \textit{soft} constraints.
\STN{s} are able to model just hard temporal constraints, \ie they can represent instances where all constraints have to be satisfied, and that the solutions of
a constraint are all equally satisfying. However, such assumption can be too much restrictive in some real-life scenarios.
For example, it may be that some solutions are preferred with respect to others and, hence, the main problem is to find a way to satisfy them optimally,
according to the preferences specified.
To address these kind of problems, in~\cite{KhatibMMR01} the authors introduced a framework in which each temporal constraint is associated with a preference
function specifying the preference for each distance or duration; a \textit{soft} simple temporal constraint is a 4-tuple $\langle (X, Y), I, A, f\rangle$
consisting of (1) an ordered pair of variables $(X, Y)$ over the integers, called the scope of the constraint; (2) an interval $I =[a,b]$, where $a$ and $b$ are
integers such that $a \leq b$; (3) a set of preferences $A$; (4) a preference function $f$, where $f: [a,b] \mapsto A$ is a mapping of the elements belonging to
interval $I$ into preference values, taken from set $A$.
An assignment $v_x$ and $v_y$ to the variables $X$ and $Y$ is said to satisfy the constraint $\langle (X,Y), I, A, f\rangle$ if and only if $a\leq v_y -v_x$.
In such a case, the preference associated to the assignment by the constraint is $f(v_y-v_x)$. Using soft simple temporal constraint, a new model of temporal
constraint network has been introduce: the \textit{Simple Temporal Problem with Preferences (STPP)}. In general, each solution of a STPP has a global preference
value, obtained by combining in a suitable way  the preference levels at which the solution satisfies the constraints. The optimal solutions of an STPP are
those solutions which are not dominated by any other solution in terms of global preference.  It was shown in~\cite{KhatibMMR01} that, in general, STPPs belongs
to the class of \NP-hard problems. When the preference functions are semi-convex and some other side conditions are observed, then the problem to find an optimal
solutions of an STPP is tractable~\cite{KhatibMMRSV07}.

Finally, another kind of possible extension is represented by the use of $\neq$ operator instead of $\leq$ in the binary constrains of \STN{s}.
Koubarakis\cite{Koubarakis97} showed that if in a \STN temporal constraints are used together with disequations in the form $x - y \neq r$, where r is a real
constant, then the problem of deciding consistency is still tractable.
This extension does not allow the specification of alternative constraints but it is interesting because it allows to exclude some solutions maintaining the
consistency problem tractable.

%% file: STNMaxMPG.bbl
\begin{thebibliography}{10}
\providecommand{\url}[1]{{#1}}
\providecommand{\urlprefix}{URL }
\expandafter\ifx\csname urlstyle\endcsname\relax
  \providecommand{\doi}[1]{DOI~\discretionary{}{}{}#1}\else
  \providecommand{\doi}{DOI~\discretionary{}{}{}\begingroup
  \urlstyle{rm}\Url}\fi

\bibitem{vanDerAalst03}
van~der Aalst, W., ter Hofstede, A., Kiepuszewski, B., Barros, A. (2003):
  Workflow patterns.
\newblock Distributed and Parallel Databases \textbf{14}(1), 5--51.
\newblock \doi{10.1023/A:1022883727209}

\bibitem{conf/flairs/BartakC07}
Bart{\'{a}}k, R., {\v{C}}epek, O.: Temporal networks with alternatives:
  Complexity and model.
\newblock In: D.~Wilson, G.~Sutcliffe (eds.) Proceedings of the Twentieth
  International Florida Artificial Intelligence Research Society Conference,
  May 7-9, 2007, Key West, Florida, {USA.}, pp. 641--646. {AAAI} Press (2007)

\bibitem{Bellman58}
Bellman, R. (1958): On a routing problem.
\newblock Quarterly of Applied Mathematics \textbf{16}(1), 87--90

\bibitem{BettiniWJ02}
Bettini, C., Wang, X.S., Jajodia, S. (2002): Temporal reasoning in workflow
  systems.
\newblock Dist. \& Paral. Data. \textbf{11}(3), 269--306.
\newblock \doi{10.1023/A:1014048800604}

\bibitem{BV07}
Björklund, H., Vorobyov, S. (2007): A combinatorial strongly subexponential
  strategy improvement algorithm for mean payoff games.
\newblock Discrete Applied Mathematics \textbf{155}(2), 210 -- 229.
\newblock \doi{10.1016/j.dam.2006.04.029}

\bibitem{BrimCha12}
Brim, L., Chaloupka, J. (2012): Using strategy improvement to stay alive.
\newblock Int. J. Found. Comput. Sci. \textbf{23}(3), 585--608.
\newblock \doi{10.1142/S0129054112400291}

\bibitem{brim2011faster}
Brim, L., Chaloupka, J., Doyen, L., Gentilini, R., Raskin, J.F. (2011): Faster
  algorithms for mean-payoff games.
\newblock Formal Methods in System Design \textbf{38}(2), 97--118.
\newblock \doi{10.1007/s10703-010-0105-x}

\bibitem{ChinnMadey00}
Chinn, S.J., Madey, G.R. (2000): Temporal representation and reasoning for
  workflow in engineering design change review.
\newblock IEEE Transactions on Engineering Management \textbf{47}(4), 485--492.
\newblock \doi{10.1109/17.895343}

\bibitem{CombiGMP12}
Combi, C., Gambini, M., Migliorini, S., Posenato, R.: Modelling temporal,
  data-centric medical processes.
\newblock In: Proc. of the 2nd ACM SIGHIT Int. Health Informatics Symp.,
  IHI~'12, pp. 141--150. ACM, New York, NY, USA (2012).
\newblock \doi{10.1145/2110363.2110382}

\bibitem{CombiGMP14}
Combi, C., Gambini, M., Migliorini, S., Posenato, R. (2014): Representing
  business processes through a temporal data-centric workflow modeling
  language: An application to the management of clinical pathways.
\newblock Systems, Man, and Cybernetics: Systems, IEEE Transactions on
  \textbf{44}(9), 1182--1203.
\newblock \doi{10.1109/TSMC.2014.2300055}

\bibitem{CombiGPP12}
Combi, C., Gozzi, M., Posenato, R., Pozzi, G. (2012): Conceptual modeling of
  flexible temporal workflows.
\newblock ACM Trans. Auton. Adapt. Syst. \textbf{7}(2), 19:1--19:29.
\newblock \doi{10.1145/2240166.2240169}

\bibitem{CombiP09}
Combi, C., Posenato, R.: Controllability in temporal conceptual workflow
  schemata.
\newblock In: BPM 2009 - Proc. of the 7th Business Process Management
  Conference, pp. 64--79 (2009).
\newblock \doi{10.1007/978-3-642-03848-8_6}

\bibitem{CombiP04}
Combi, C., Pozzi, G.: Architectures for a temporal workflow management system.
\newblock In: Proc. of the 2004 ACM Symp. on Applied Computing, SAC '04, pp.
  659--666. ACM, New York, NY, USA (2004).
\newblock \doi{10.1145/967900.968040}

\bibitem{Comin15}
Comin, C.: A {HyTN} {C}onsistency {C}heck {A}lgorithm {I}mplementation in
  {C++}.
\newblock
  \url{http://profs.scienze.univr.it/~posenato/software/hytn/2015_v1_Code.tgz}
  (2015)

\bibitem{CominPR14}
Comin, C., Posenato, R., Rizzi, R.: A tractable generalization of simple
  temporal networks and its relation to mean payoff games.
\newblock In: 21st International Symposium on Temporal Representation and
  Reasoning (TIME 2014), pp. 7--16. IEEE CPS (2014).
\newblock \doi{10.1109/TIME.2014.19}

\bibitem{TIME2015CTNsWithTNs}
Comin, C., Rizzi, R.: Dynamic consistency of conditional simple temporal
  networks via mean payoff games: a singly-exponential time {DC-C}hecking.
\newblock In: 22nd International Symposium on Temporal Representation and
  Reasoning (TIME 2015), pp. 19--28. IEEE CPS (2015).
\newblock \doi{10.1109/TIME.2015.18}

\bibitem{Cormen01}
Cormen, T.H., Leiserson, C.E., Rivest, R.L., Stein, C.: Introduction to
  Algorithms.
\newblock The MIT Press (2001)

\bibitem{DechterMP91}
Dechter, R., Meiri, I., Pearl, J. (1991): Temporal constraint networks.
\newblock Artificial Intelligence \textbf{49}(1--3), 61--95.
\newblock \doi{10.1016/0004-3702(91)90006-6}

\bibitem{EderGP00}
Eder, J., Gruber, W., Panagos, E.: Temporal modeling of workflows with
  conditional execution paths.
\newblock In: M.~Ibrahim, J.~K{\"u}ng, N.~Revell (eds.) Database and Expert
  Systems Applications (DEXA 2000), \emph{LNCS}, vol. 1873, pp. 243--253.
  Springer Berlin Heidelberg (2000).
\newblock \doi{10.1007/3-540-44469-6_23}

\bibitem{EderPR99}
Eder, J., Panagos, E., Rabinovich, M.: Time constraints in workflow systems.
\newblock In: M.~Jarke, A.~Oberweis (eds.) Advanced Information Systems
  Engineering, \emph{LNCS}, vol. 1626, pp. 286--300. Springer Berlin Heidelberg
  (1999).
\newblock \doi{10.1007/3-540-48738-7_22}

\bibitem{EhrenfeuchtMycielski:1979}
Ehrenfeucht, A., Mycielski, J. (1979): Positional strategies for mean payoff
  games.
\newblock Int. Journal of Game Theory \textbf{8}(2), 109--113.
\newblock \doi{10.1007/BF01768705}

\bibitem{Ford}
Ford~Jr., L.R., Fulkerson, D.R.: Flows in networks, vol.~3.
\newblock Princeton University Press (1962)

\bibitem{GonzalezR08}
Gonzalez~del Foyo, P.M., Reinaldo~Silva, J.: Using time {P}etri {N}ets for
  modeling and verification of timed constrained workflow systems.
\newblock In: ABCM Symposium Series in Mechatronics, pp. 471--478. Dept. Of
  Mechatronics, University of S{\~a}o Paulo, Brazil (2008)

\bibitem{GareyJohnson:1979}
Garey, M.R., Johnson, D.S.: Computers and Intractability: A Guide to the Theory
  of NP-Completeness.
\newblock W. H. Freeman \& Co., New York, NY, USA (1979)

\bibitem{wfmc-modello}
Hollingsworth, D.: The workflow reference model.
\newblock http://www.wfmc.org/standards/model.htm (1995)

\bibitem{HunsbergerPC15}
Hunsberger, L., Posenato, R., Combi, C.: A sound-and-complete propagation-based
  algorithm for checking the dynamic consistency of conditional simple temporal
  networks.
\newblock In: F.~Grandi, M.~Lange, A.~Lomuscio (eds.) 22st International
  Symposium on Temporal Representation and Reasoning (TIME 2015), pp. 4--18.
  IEEE CPS (2015).
\newblock \doi{10.1109/TIME.2015.26}

\bibitem{Jur98}
Jurdziński, M. (1998): Deciding the winner in parity games is in
  {UP}~$\cap$~{co-UP}.
\newblock Information Processing Letters \textbf{68}(3), 119--124.
\newblock \doi{10.1016/S0020-0190(98)00150-1}

\bibitem{KhatibMMR01}
Khatib, L., Morris, P., Morris, R., Rossi, F.: Temporal constraint reasoning
  with preferences.
\newblock In: Proceedings of the 17th International Joint Conference on
  Artificial Intelligence - Volume 1, IJCAI'01, pp. 322--327. Morgan Kaufmann
  Publishers Inc., San Francisco, CA, USA (2001)

\bibitem{KhatibMMRSV07}
Khatib, L., Morris, P., Morris, R., Rossi, F., Sperduti, A., Venable, K.B.
  (2007): Solving and learning a tractable class of soft temporal constraints:
  Theoretical and experimental results.
\newblock AI Communications \textbf{20}(3), 181--209

\bibitem{Koubarakis97}
Koubarakis, M. (1997): From local to global consistency in temporal constraint
  networks.
\newblock Theoretical Computer Science \textbf{173}(1), 89 -- 112.
\newblock \doi{10.1016/S0304-3975(96)00192-2}

\bibitem{Lanz14b}
Lanz, A., Reichert, M.: Enabling time-aware process support with the atapis
  toolset.
\newblock In: L.~Limonad, B.~Weber (eds.) Proceedings of the BPM Demo Sessions
  2014, \emph{CEUR Workshop Proceedings}, vol. 1295, pp. 41--45. CEUR (2014)

\bibitem{LanzWR12}
Lanz, A., Weber, B., Reichert, M. (2012): Time patterns for process-aware
  information systems.
\newblock Requirements Engineering \textbf{19}(2), 113--141.
\newblock \doi{10.1007/s00766-012-0162-3}

\bibitem{LifshitsPavlov:2007}
Lifshits, Y., Pavlov, D. (2007): Potential theory for mean payoff games.
\newblock Journal of Mathematical Sciences \textbf{145}(3), 4967--4974.
\newblock \doi{10.1007/s10958-007-0331-y}

\bibitem{Mendling10}
Mendling, J., Reijers, H.A., van~der Aalst, W.M.P. (2010): Seven process
  modeling guidelines ({7PMG}).
\newblock Information and Software Technology \textbf{52}(2), 127--136.
\newblock \doi{10.1016/j.infsof.2009.08.004}

\bibitem{MerlinF76}
Merlin, P., Farber, D.J. (1976): Recoverability of communication
  protocols--implications of a theoretical study.
\newblock Communications, IEEE Transactions on \textbf{24}(9), 1036--1043.
\newblock \doi{10.1109/TCOM.1976.1093424}

\bibitem{MorrisMV01}
Morris, P., Muscettola, N., Vidal, T.: Dynamic control of plans with temporal
  uncertainty.
\newblock In: Proceedings of the 17th International Joint Conference on
  Artificial Intelligence - Volume 1, IJCAI'01, pp. 494--499. Morgan Kaufmann
  Publishers Inc., San Francisco, CA, USA (2001)

\bibitem{Pani:2001tb}
Pani, A., Bhattacharjee, G. (2001): Temporal representation and reasoning in
  artificial intelligence: A review.
\newblock Mathematical and Computer Modelling \textbf{34}(1–2), 55--80.
\newblock \doi{10.1016/S0895-7177(01)00049-8}

\bibitem{Papadimitriou:1994}
Papadimitriou, C.H.: Computational Complexity.
\newblock Addison-Wesley (1994)

\bibitem{pgsolver}
pgsolver: The pgsolver collection of parity game solvers.
\newblock https://github.com/tcsprojects/pgsolver (2013)

\bibitem{Kumar05}
Satish~Kumar, T.K.: On the tractability of restricted disjunctive temporal
  problems.
\newblock In: ICAPS 2005 - Proceedings of the 15th International Conference on
  Automated Planning and Scheduling, pp. 110--119 (2005)

\bibitem{Sch08}
Schewe, S.: An optimal strategy improvement algorithm for solving parity and
  payoff games.
\newblock In: M.~Kaminski, S.~Martini (eds.) Computer Science Logic,
  \emph{LNCS}, vol. 5213, pp. 369--384. Springer (2008).
\newblock \doi{10.1007/978-3-540-87531-4_27}

\bibitem{ScheweTV15}
Schewe, S., Trivedi, A., Varghese, T.: Symmetric strategy improvement.
\newblock In: M.M. Halld{\'{o}}rsson, K.~Iwama, N.~Kobayashi, B.~Speckmann
  (eds.) Automata, Languages, and Programming - 42nd International Colloquium,
  {ICALP} 2015, Kyoto, Japan, July 6-10, 2015, Proceedings, Part {II},
  \emph{Lecture Notes in Computer Science}, vol. 9135, pp. 388--400. Springer
  (2015).
\newblock \doi{10.1007/978-3-662-47666-6_31}

\bibitem{SmithFJ00}
Smith, D., Frank, J., Jónsson, A. (2000): Bridging the gap between planning
  and scheduling.
\newblock Knowledge Engineering Review \textbf{15}(1), 47--83.
\newblock \doi{10.1017/S0269888900001089}

\bibitem{StergiouK00}
Stergiou, K., Koubarakis, M. (2000): Backtracking algorithms for disjunctions
  of temporal constraints.
\newblock Artificial Intelligence \textbf{120}(1), 81--117.
\newblock \doi{10.1016/S0004-3702(00)00019-9}

\bibitem{TVP03}
Tsamardinos, I., Vidal, T., Pollack, M.E. (2003): Ctp: A new constraint-based
  formalism for conditional, temporal planning.
\newblock Constraints \textbf{8}(4), 365--388.
\newblock \doi{10.1023/A:1025894003623}

\bibitem{VidalF99}
Vidal, T., Fargier, H. (1999): Handling contingency in temporal constraint
  networks: from consistency to controllabilities.
\newblock Journal of Experimental and Theoretical Artificial Intelligence
  \textbf{11}(1), 23--45.
\newblock \doi{10.1080/095281399146607}

\bibitem{ZwickPaterson:1996}
Zwick, U., Paterson, M. (1996): The complexity of mean payoff games on graphs.
\newblock Theoretical Computer Science \textbf{158}(1--2), 343--359.
\newblock \doi{10.1016/0304-3975(95)00188-3}

\end{thebibliography}
